\documentclass[nojss]{jss}

\author{Saptarshi Chakraborty\\ Memorial Sloan Kettering Cancer Center \And
        Samuel W. K. Wong\\ University of Waterloo}
\title{\pkg{BAMBI}: An \proglang{R} package for Fitting Bivariate Angular Mixture Models}

\Plainauthor{Saptarshi Chakraborty, Samuel W. K. Wong} 
\Plaintitle{BAMBI: An R package for Fitting Bivariate Angular Mixture Models} 
\Shorttitle{\pkg{BAMBI}: Bivariate Angular Mixture Models} 

\Abstract{
  Statistical analyses of directional or angular data have applications in a variety of fields, such as geology, meteorology and bioinformatics. There is substantial literature on descriptive and inferential techniques for univariate angular data, with the bivariate (or more generally, multivariate) cases receiving more attention in recent years. More specifically, the bivariate wrapped normal, von Mises sine and von Mises cosine distributions, and mixtures thereof, have been proposed for practical use. However, there is a lack of software implementing these distributions and the associated inferential techniques. 
  In this article, we introduce \pkg{BAMBI}, an \proglang{R} package for analyzing bivariate (and univariate) angular data. We implement random data generation, density evaluation, and computation of theoretical summary measures (variances and correlation coefficients) for the three aforementioned bivariate angular distributions,  as well as two univariate angular distributions: the univariate wrapped normal and the univariate von Mises distribution.  The major contribution of \pkg{BAMBI} to statistical computing is in providing Bayesian methods for modeling angular data using finite mixtures of these distributions. We also provide functions for visual and numerical diagnostics and Bayesian inference for the fitted models. In this article, we first provide a brief review of the distributions and techniques used in \pkg{BAMBI}, then describe the capabilities of the package, and finally conclude with demonstrations of mixture model fitting using \pkg{BAMBI} on the two real datasets included in the package, one univariate and one bivariate.
}

\Keywords{angular data, mixture models, bivariate data, von Mises distribution, wrapped normal distribution, \proglang{R},
Hamiltonian Monte Carlo, MCMC, Gibbs Sampler}
\Plainkeywords{angular data, mixture models, bivariate data, von Mises distribution, wrapped normal distribution, R, Hamiltonian Monte Carlo, MCMC, Gibbs Sampler} 

\Address{
  Saptarshi Chakraborty\\
  Department of Epidemiology and Biostatistics\\
  Memorial Sloan Kettering Cancer Center\\
  E-mail: \email{chakrabs@mskcc.org}\\ \\ 
  Samuel W.K. Wong \\ 
  Department of Statistics and Actuarial Science \\
  University of Waterloo \\
  E-mail: \email{samuel.wong@uwaterloo.ca}
}

\usepackage{amsthm,amsmath}
\usepackage[utf8]{inputenc}
\usepackage{makeidx}
\usepackage{amsfonts}
\usepackage{amssymb}
\usepackage{gensymb}
\usepackage{caption}
\usepackage[labelformat=simple]{subcaption}
\usepackage{bbm}  
\usepackage{lipsum}
\usepackage{enumitem}  
\usepackage[english]{babel} 
\usepackage{bm} 
\usepackage[margin = 1in]{geometry}
\usepackage[T1]{fontenc}	 
\usepackage{lmodern}
\usepackage{amsthm} 
\usepackage{graphicx} 
\usepackage{ltablex}
\usepackage{booktabs} 
\usepackage{multicol} 
\usepackage{multirow} 
\usepackage{tocstyle} 
\usepackage{filecontents} 
\usepackage[toc,page]{appendix}
\usepackage[section]{placeins} 

\newcommand\numbereqn{\addtocounter{equation}{1}\tag{\theequation}}
\newcommand{\N}{\text{N}}
\newcommand{\wn}{\text{WN}}
\newcommand{\bwn}{\text{WN}_2}
\newcommand{\vm}{\text{vM}}
\newcommand{\vms}{\text{vM}_2^s}
\newcommand{\vmc}{\text{vM}_2^c}

\newcommand{\pb}{\bm{p}}
\newcommand{\qb}{\bm{q}}
\newcommand{\rb}{\bm{r}}
\newcommand{\yb}{\bm{y}}
\newcommand{\tru}{\text{true}}
\newcommand{\psib}{\bm{\psi}}
\newcommand{\alphab}{\bm{\alpha}}
\newcommand{\etab}{\bm{\eta}}
\newcommand{\zetab}{\bm{\zeta}}
\newcommand{\Psib}{\bm{\Psi}}
\newcommand{\thetab}{\bm{\theta}}
\newcommand{\mub}{\bm{\mu}}
\newcommand{\omegab}{\bm{\omega}}
\newcommand{\Z}{\mathbb{Z}}
\newcommand{\R}{\mathbb{R}}
\newcommand{\ft}{\tilde{f}}
\newcommand{\pt}{\tilde{p}}

\newcommand{\qt}{\tilde{q}}
\newcommand{\lt}{\tilde{l}}
\newcommand{\dat}{\text{data}}
\newcommand{\bambi}{\pkg{BAMBI }}
\newcommand{\one}{\mathbbm{1}}
\DeclareMathOperator{\var}{var}
\newcommand{\js}{\text{JS}}
\newcommand{\fl}{\text{FL}}

\newtheorem{prop}{Proposition}[section]

\allowdisplaybreaks[1]

\begin{document}

\section{Introduction} \label{intro}

Statistical analyses of angular or directional data have found applications in a variety of fields, such as geology (earth's magnetic poles), meteorology (wind directions) and bioinformatics (backbone structures of proteins). Directional data can be univariate or multivariate, and one way of representing such data is via angles measured on a circle $[0, 2\pi)$ (element-wise when multivariate), and hence the name \textit{angular}. Angular methods are also applicable to any interval that wraps around (e.g., $[0, L)$ or $[-L/2, L/2)$ for some $L > 0$) when transformed to the circle $[0, 2\pi)$.  The wraparound condition on the support invalidates direct applicability of many standard statistical methods.  There is substantial literature devoted to the development of descriptive and inferential techniques for directional data (see, e.g.,  \cite{mardia:jupp:2009, mardia:1972, fisher:1995}), with the traditional univariate case as the primary focus, although the bivariate case is gaining increasing interest \cite{singh:2002, mardia:2007} along with the emergence of new applications. Bivariate angular data can now be found in  a variety of modern scientific problems, with many notable applications arising from the field of computational biology \citep{mardia:2007, boomsma:2008, lennox:2009, bhattacharya:2015}. A major area of research in protein bioinformatics involves modeling and predicting protein 3-D structures, which requires proper handling of the paired backbone torsion angles. Formal analyses of these bivariate angle pairs thus require rigorous statistical techniques and models.

A unique feature in the modeling of  directional data is the use of angular probability distributions, or  mixtures thereof (see Section~\ref{mixmodels}), which are inherently different from their linear (Euclidean) counterparts because of the wraparound nature of their supports.  Bayesian methods provide flexible tools for analyzing and modeling  such data. First, one may incorporate prior information, if available, into modeling. Second, one may use powerful computational methods, i.e., Markov chain Monte Carlo (MCMC, see Section~\ref{bayesintro}) for sampling from the posterior, to fit such models and assess the fitted models.  Third, one may readily compute posterior quantities of interest while coherently accounting for uncertainty in the model parameters.  Within this context, this package was developed for fitting \textbf{B}ivariate \textbf{A}ngular \textbf{M}ixtures using \textbf{B}ayesian \textbf{I}nference (hence, the package name \pkg{BAMBI}).  In \bambi we implement the two most popular angular distributions, namely the wrapped normal (or Gaussian) and the von Mises distributions, and consider both univariate and bivariate versions of these.  \bambi  provides functionality for modeling univariate and bivariate angular data using these distributions, and for fitting finite mixture models of these distributions.   We first introduce the basics of these distributions and mixture models. It should be noted that the bivariate distributions considered in this paper have support $[0, 2\pi)^2$ (i.e., on a \textit{torus}), which are distinct from those defined on the surface of the \textit{unit sphere}, such as the von Mises - Fisher distribution.

\subsection{Wrapped Normal Distributions} \label{wnmodels}
For univariate continuous data, the angular analogue of the normal distribution on the real line is the wrapped normal distribution obtained by \textit{wrapping} a normal random variable around the unit circle  (see, e.g., \cite{jona-lasinio:2012}). Formally, let $X$ be a normal random variable with mean $\mu$ and variance $\sigma^2 > 0$. Then the distribution of $\psi = X \mod 2\pi$ is called the \emph{wrapped normal distribution} with mean $\mu$ and variance $\sigma^2$ and is denoted by $\wn(\mu, \sigma^2)$. The density of $\psi \sim \wn(\mu, \sigma^2)$ is given by:
\begin{align}\label{wn_var_rep}
f_\wn(\psi|\mu, \sigma) = \frac{1}{\sigma \sqrt{2\pi}} \sum_{\omega \in \Z} \exp \left[-\frac{1}{2 \sigma^2}(\psi - \mu - 2\pi \omega)^2 \right]; \quad \psi \in [0, 2\pi)
\end{align}
where $\Z$ denotes the set of all integers. Since the density contains a summation over entire $\Z$, without loss of generality, we let $\mu \in [0, 2\pi)$ to ensure identifiability. Figure~\ref{fig:wnormdenmodel} displays the univariate wrapped density with $\mu = \pi$ and $\kappa = 0.01, 1$ and 10, which shows that the density is symmetric around $\mu$ and becomes more concentrated as $\kappa$ increases.

The multivariate generalization of the above distribution is straightforward \citep{jona-lasinio:2012}.  The distribution of a random vector $\psib = (\psi_1, \cdots, \psi_p)^\top$ with probability density
\begin{align} \label{mvwn_var_rep}
\frac{1}{ \sqrt{|\Sigma|(2\pi)^p}} \sum_{\omegab \in \Z^p} \exp \left[-\frac{1}{2}\left(\psib - \mub - 2\pi \omegab \right)^\top \Sigma^{-1} \left(\psib - \mub - 2\pi \omegab \right) \right]; \quad \psib \in [0, 2\pi)^p
\end{align}
with $\mub \in [0, 2\pi)^p$ and $\Sigma$  positive definite, is called the $p$-variate wrapped normal distribution with mean vector $\mub$ and variance matrix $\Sigma$, denoted by $\wn_p(\mub, \Sigma)$.  Although (\ref{wn_var_rep}) and (\ref{mvwn_var_rep}) are the most common parameterizations of the wrapped normal distributions found in the literature, to facilitate comparability with the von Mises distribution (defined in Section~\ref{vmmodels}), we shall use the equivalent representation(s) obtained through the re-parameterization(s) $\kappa = 1/\sigma^2$ and $\Delta = \Sigma^{-1}$.  \pkg{BAMBI} handles the univariate and bivariate cases, namely $p = 1$ and $p = 2$. Thus, the form of the univariate  wrapped normal density we use is
\begin{align} \label{wn}
f_\wn(\psi|\mu, \kappa) = \sqrt{\frac{\kappa}{2\pi}} \sum_{\omega \in \Z} \exp \left[-\frac{\kappa}{2}(\psi - \mu - 2\pi \omega)^2 \right]; \quad \psi \in [0, 2\pi)
\end{align}
with $\mu \in [0, 2\pi)$ and $\kappa > 0$; and that of the bivariate density is
\begin{align*} \label{bvwn}
&\quad f_{\wn_2}(\psi_1, \psi_2 | \mu_1, \mu_2, \kappa_1, \kappa_2, \kappa_3) \\
&= \frac{\sqrt{\kappa_1 \kappa_2 - \kappa_3^2}}{2\pi} \sum_{(\omega_1, \omega_2) \in \Z^2}  \exp  \left[-\frac{1}{2} \left\lbrace  \kappa_1 (\psi_1 - \mu_1 - 2\pi \omega_1)^2 + \kappa_2 (\psi_2 - \mu_2 - 2\pi \omega_2)^2  \right. \right.\\
& \qquad \qquad \qquad \qquad \qquad \qquad \qquad \quad \left. \left. + 2 \kappa_3 (\psi_1 - \mu_1 - 2\pi \omega_1) (\psi_2 - \mu_2 - 2\pi \omega_2) \right\rbrace \right] \numbereqn
\end{align*}
where $\psi_1, \psi_2, \mu_1, \mu_2 \in [0, 2\pi)$, $\kappa_1, \kappa_2 > 0$ and $\kappa_3^2 \leq \kappa_1\kappa_2$, obtained by letting $\mub = (\mu_1, \mu_2)^\top$ and
\[
\Delta = \begin{pmatrix}
\kappa_1 & \kappa_3 \\
\kappa_3 & \kappa_2
\end{pmatrix}.
\]
Figure~\ref{fig:wnormdenmodel} plots the univariate wrapped normal densities with $\mu = \pi$ and $\kappa = 0.01, 1$ and 10, which shows that the density is symmetric around $\mu$ and becomes more concentrated as $\kappa$ increases. 
\begin{figure}[h]
	\centering
	\subcaptionbox{Wrapped normal density \label{fig:wnormdenmodel}}%
	[.485\linewidth]{\includegraphics[height=3.2in, width = 3.2in]{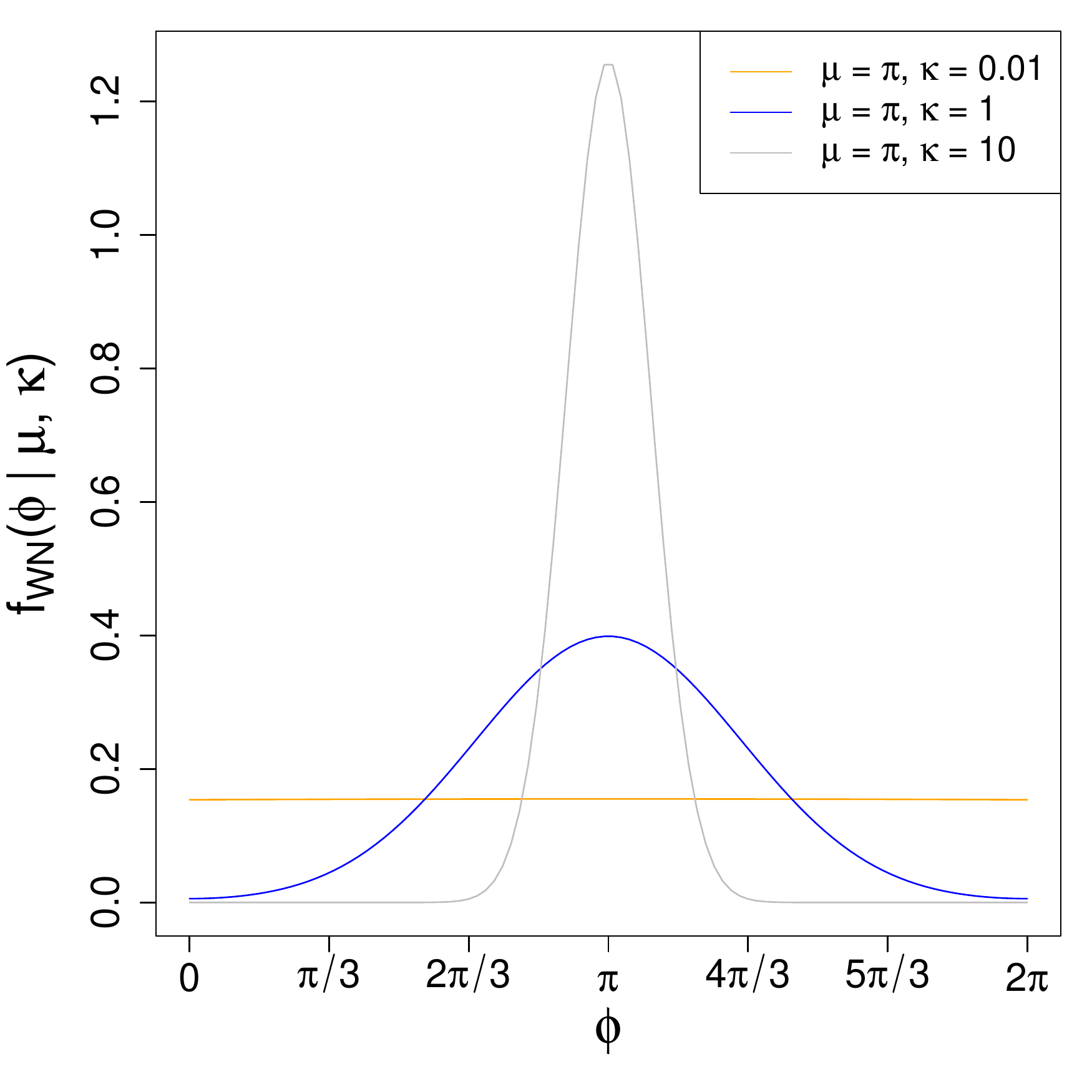}}
	\hfill
	\subcaptionbox{von Mises density \label{fig:vmdenmodel}}%
	[.485\linewidth]{\includegraphics[height=3.2in, width = 3.2in]{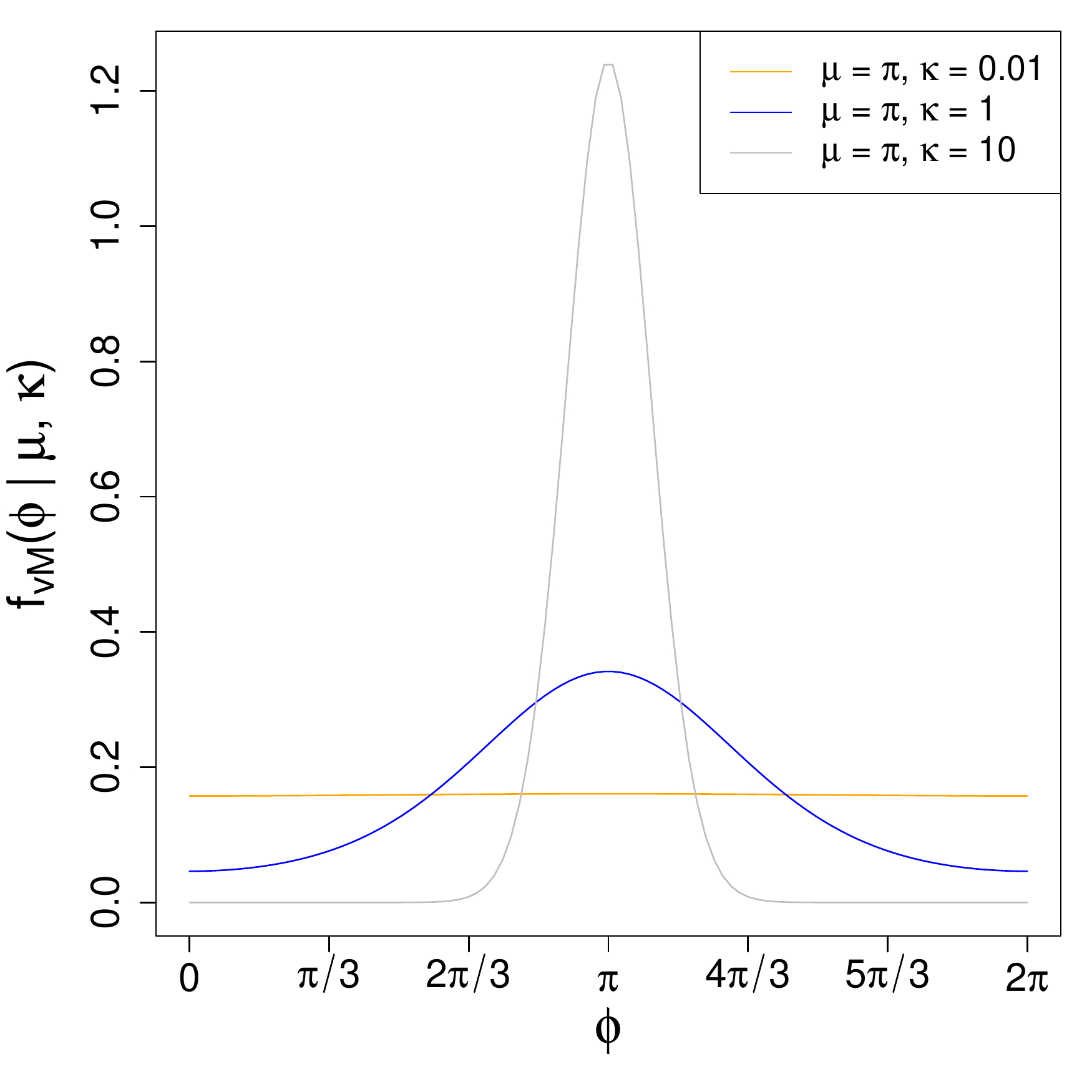}}
	\caption{Univariate wrapped normal density $f_\wn(\phi|\mu, \kappa)$ and univariate von Mises density $f_\vm(\phi|\mu, \kappa)$ with $\mu = \pi$ and different $\kappa$'s.}
	\label{fig:unvariant_den_plots}
\end{figure}

Similarly, the bivariate wrapped normal  density is also symmetric around $(\mu_1, \mu_2)$ and becomes more concentrated as $\kappa_1$ and/or $\kappa_2$ increases, while the parameter $\kappa_3$ regulates the association between the random coordinates.
This can be visualized from Figure~\ref{fig:wnorm2denmodel} displaying the surfaces of the density created via \pkg{BAMBI} function \code{surface_model},  for different parameter combinations. (Codes for generating these plots can be found in the replication \proglang{R} script for this paper.) The upper panels of Figure~\ref{fig:wnorm2denmodel} show how the density becomes more concentrated when $\kappa_1$ and $\kappa_2$ are increased (while keeping $\kappa_3$ fixed). In contrast,  the lower panels of Figure~\ref{fig:wnorm2denmodel} display density surfaces showing how the association between the random coordinates changes (from positive to negative), when $\kappa_3$ is changed (from negative to positive, since $\kappa_3$ is the diagonal element of the \emph{inverse} covariance matrix) while keeping $\kappa_1$ and $\kappa_2$ are fixed. 

\begin{figure}[h]
	\centering
	\includegraphics[width=\linewidth]{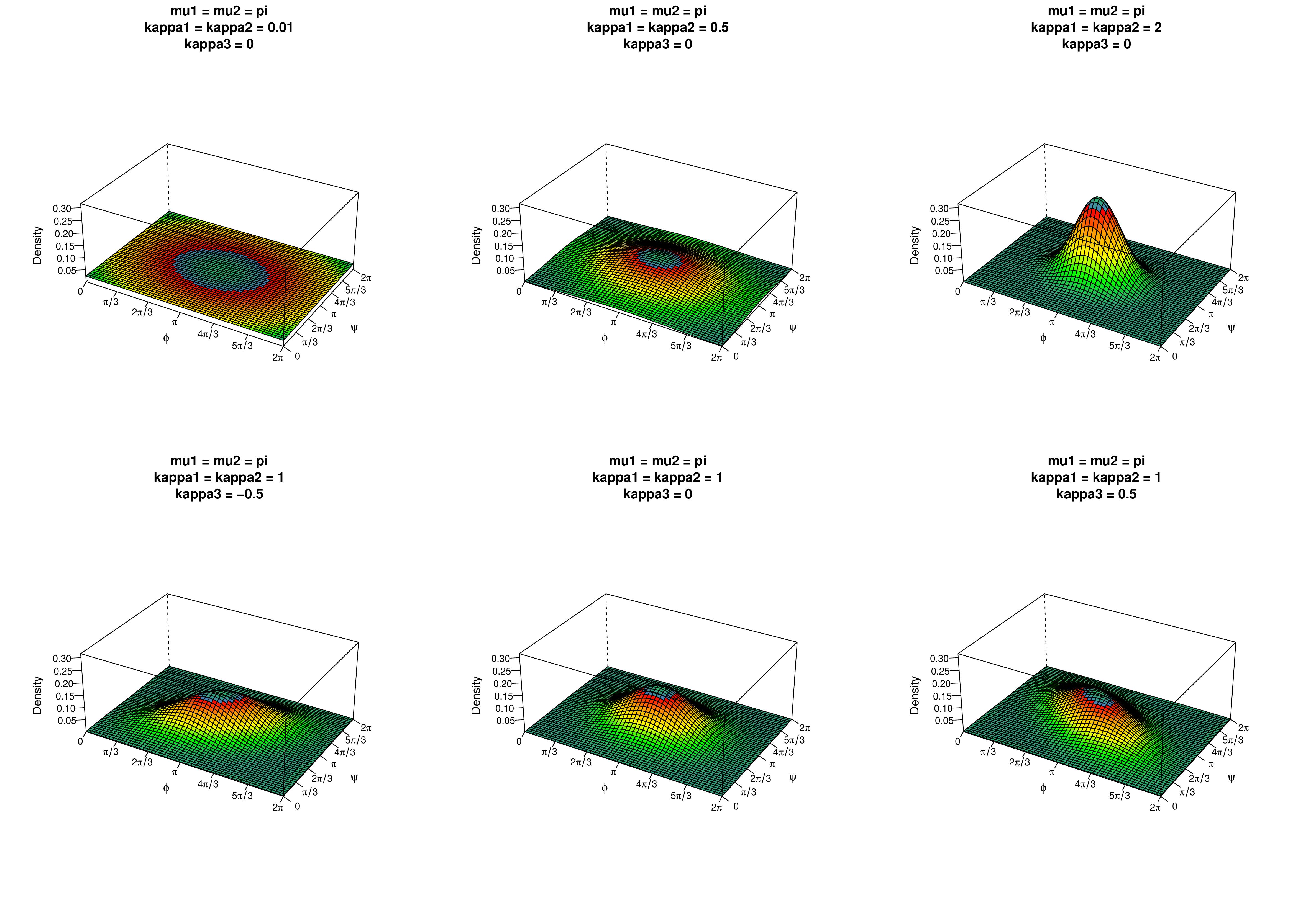}
	\caption{Bivariate wrapped normal density for $\mu_1 = \mu_2 =  \pi$ and various $\kappa_1$, $\kappa_2$, $\kappa_3$.}
	\label{fig:wnorm2denmodel}
\end{figure}

Note that when $\kappa \rightarrow 0$ (or $\Delta \rightarrow 0_{2 \times 2}$) then the distribution of $\psi = X \mod 2\pi$ converges to the uniform distribution over $[0, 2\pi)$ (or $[0, 2\pi)^2$). Hence, we shall include the cases $\kappa = 0$ and $\kappa_1 = \kappa_2 = \kappa_3 = 0$ in the support of these parameters, and define the associated densities by their limits.

The precision parameter $\kappa$ ($\kappa_1, \kappa_2$ in the bivariate case) is (are) conceptually similar to the concentration parameters in the von Mises distribution (see Section~\ref{vmmodels}). Therefore to aid comparability, we shall call $\kappa$ ($\kappa_1$ and $\kappa_2$) the concentration parameter(s) of the univariate (bivariate) wrapped normal model. In \pkg{BAMBI}, evaluation of univariate and bivariate wrapped normal densities are implemented through the function \code{dwnorm} and \code{dwnorm2} respectively. Random data  from these models can be generated using \code{rwnorm} and \code{rwnorm2} respectively.

\subsection{von Mises distributions} \label{vmmodels}
Wrapped normal models have a high computational cost in practice. Although the sum over $\Z$ in the expression for the density can be well-approximated by a sum over the set $A = \{-3,-2,-1,0,1,2,3\}$ (i.e., 3 integer displacements, covering $\pm$ 3 standard deviations from the mean), it can be seen that the number of terms in the sum grows exponentially as the dimension increases. For instance, in the bivariate case, even if $\Z$ is approximated  by set $A$, the (double) sum in the density  consists of 49 terms.

Because of this difficulty, the von Mises distribution is an alternative that is widely used; it is able to approximate the wrapped normal while being less computationally intensive \citep[p. 36]{mardia:jupp:2009}. A random variable $\psi$ is said to follow the von Mises distribution (also called  the circular normal distribution, \citet{jammalamadaka:2001})  with mean parameter $\mu$ and concentration parameter $\kappa$, denoted $\psi \sim \vm(\mu, \kappa)$, if $\psi$ has the density
\begin{align} \label{vm}
f_\vm(\psi \mid \mu, \sigma) = \frac{1}{2\pi I_0(\kappa)} \exp(\kappa \cos(\psi-\mu)); \quad \psi \in [0, 2\pi)
\end{align}
where $\mu \in [0, 2\pi)$, $\kappa \geq 0$ and $I_r(\cdot)$ denotes the modified Bessel function of the first kind and order $r$.  Letting $\kappa = 0$ makes (\ref{vm}) the uniform density over $[0, 2\pi)$, and when $\kappa \to \infty$, (\ref{vm}) converges to a normal density. An intuitive explanation of the latter result follows from the fact that when  the concentration parameter $\kappa$ is large, $\psi - \mu \approx 0$, so that $\cos (\psi - \mu) \approx 1-(\psi-\mu)^2/2$, which makes the exponent in the density (\ref{vm}) approximately proportional to the  $\N(\mu, (1/\sqrt{\kappa})^2)$ density. A formal proof can be found in  \citet[Proposition 2.2]{jammalamadaka:2001}. 

Figure~\ref{fig:vmdenmodel} plots the von Mises densities with $\mu = \pi$ and $\kappa = 0.01, 1$ and 10, which shows that the density is symmetric around $\mu$ and becomes more concentrated as $\kappa$ increases, and that the density is broadly similar to the associated univariate wrapped normal density.

A multivariate generalization for the univariate von Mises distribution is however not as straightforward as the wrapped normal distribution, as there is not a unique way of defining a  multivariate distribution with univariate von Mises-like marginals.  In the bivariate case, two versions of the bivariate von Mises distribution have been suggested for practical use, namely the sine model \citep{singh:2002} and the cosine model \citep{mardia:2007}. They are comparable to the bivariate normal model both in terms of number of parameters (five), and the interpretability of those parameters.  Other generalizations with more parameters have been studied theoretically \citep{mardia:1975, rivest:1988}.

Let $\psib = (\psi_1, \psi_2)^\top$ be a random vector on $\R^2$ with support $[0, 2\pi)^2$. Then $\psib$ is said to follow the (bivariate) von Mises sine distribution with mean parameters $\mu_1, \mu_2$, concentration parameters $\kappa_1, \kappa_2$, and association parameter $\kappa_3$, denoted $\psib \sim \vms(\mu_1, \mu_2, \kappa_1, \kappa_2, \kappa_3)$), if $\psib$ has the probability density
\begin{align*} \label{vms}
&\quad f_{\vms}(\psi_1, \psi_2 \mid \mu_1, \mu_2, \kappa_1, \kappa_2, \kappa_3) \\
&= {C_s(\kappa_1, \kappa_2, \kappa_3)} \exp [ \kappa_1 \cos(\psi_1 - \mu_1) + \kappa_2 \cos(\psi_2 - \mu_2) + \kappa_3 \sin(\psi_1 - \mu_1) \sin(\psi_2 - \mu_2) ] \numbereqn
\end{align*}
where $\kappa_1, \kappa_2 \geq 0$, $-\infty < \kappa_3 < \infty$, $\mu_1, \mu_2 \in [0, 2\pi)$ and the normalizing constant is given by
\begin{align} \label{c_vms}
C_s(\kappa_1, \kappa_2, \kappa_3)^{-1} = 4 \pi^2 \sum_{m = 0} ^\infty  {{2m}\choose{m}} \left(\frac{\kappa_3^2}{4 \kappa_1 \kappa_2}\right)^m I_m(\kappa_1) I_m(\kappa_2).
\end{align}

In contrast, $\psib$ is said to follow the (bivariate) von Mises cosine distribution with mean parameters $\mu_1, \mu_2$, concentration parameters $\kappa_1, \kappa_2$, and association parameter $\kappa_3$, denoted $\psib \sim \vmc(\mu_1, \mu_2, \kappa_1, \kappa_2, \kappa_3)$, if $\psib$ has the probability density~\footnote{\cite{mardia:2007} define the density with $-\kappa_3$ instead of $\kappa_3$ in the exponent. However, that makes the normalizing constant equal to $C_c(\kappa_1, \kappa_2, -\kappa_3)$ in our current notation (i.e., in the form shown in (\ref{c_vmc})) and not $C_c(\kappa_1, \kappa_2, \kappa_3)$ as given in the paper. See Appendix~\ref{append_c_vmc_proof} for a proof.}
\begin{align*} \label{vmc}
&\quad f_{\vmc}(\psi_1, \psi_2 \mid \mu_1, \mu_2, \kappa_1, \kappa_2, \kappa_3) \\
&=  C_c(\kappa_1, \kappa_2, \kappa_3) \exp [ \kappa_1 \cos(\psi_1 - \mu_1) + \kappa_2 \cos(\psi_2 - \mu_2) + \kappa_3 \cos(\psi_1 - \mu_1 - \psi_2 + \mu_2) ]. \numbereqn
\end{align*}
Here, similar to the sine model, $\kappa_1, \kappa_2 \geq 0$, $-\infty < \kappa_3 < \infty$, $\mu_1, \mu_2 \in [0, 2\pi)$ and the normalizing constant is given by
\begin{align} \label{c_vmc}
C_c(\kappa_1, \kappa_2, \kappa_3)^{-1} = 4 \pi^2 \left\lbrace I_0(\kappa_1) I_0(\kappa_2) I_0(\kappa_3) + 2 \sum_{m = 0} ^\infty I_m(\kappa_1) I_m(\kappa_2) I_m(\kappa_3)  \right \rbrace.
\end{align}

From (\ref{vms}) and (\ref{vmc}) it is easy to see that when $\kappa_3 = 0$, both the von Mises sine and cosine densities become products of univariate von Mises densities, implying independence between the two random coordinates. In addition, when $\kappa_1$ and $\kappa_2$ are also zero, both densities become uniform over $[0, 2\pi)^2$. \cite{singh:2002} and \cite{mardia:2007} provide explicit forms for the marginal and conditional distributions in the sine and cosine models; the conditional distributions in both sine and cosine models are univariate von Mises, whereas the marginal distributions, although not von Mises, are symmetric around $\mu_1$ and $\mu_2$. 

One key difference between the bivariate wrapped normal model and the bivariate von Mises models is that  $\kappa_3^2$ is not required to be bounded above by $\kappa_1 \kappa_2$ in the latter, and thus can take any value in $(-\infty, \infty)$. Consequently, the densities can be bimodal; \cite{mardia:2007} show that the sine (cosine) joint density is unimodal if  $\kappa_3^2 < \kappa_1 \kappa_2$ ($\kappa_3 \geq - \kappa_1\kappa_2/(\kappa_1+\kappa_2)$), and bimodal otherwise. This flexibility gives the two bivariate von Mises distributions  richer sets of possible contour plots and the ability to model a larger class of angular data.

\begin{figure}[h]
	\centering
	\includegraphics[width=\linewidth]{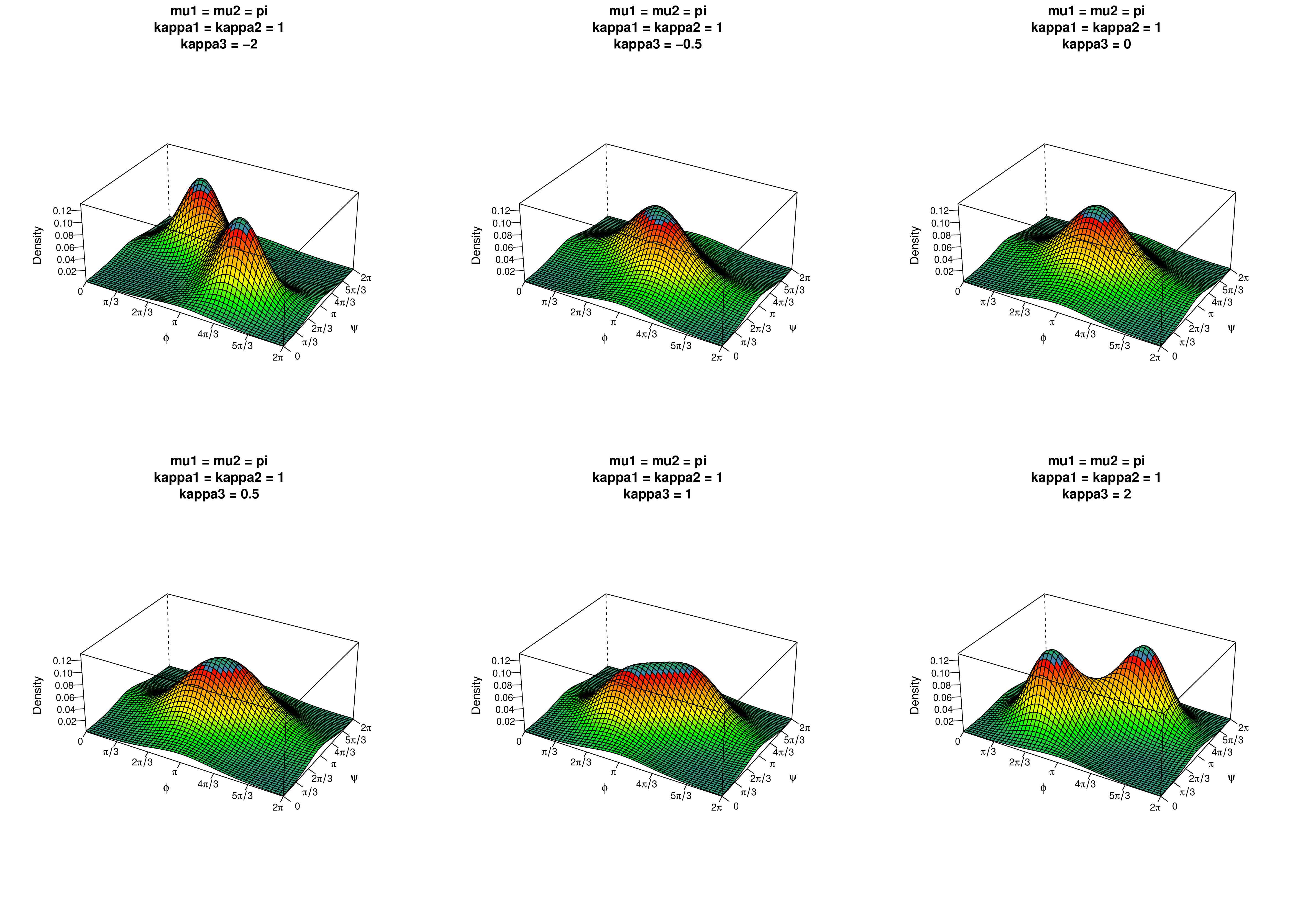}
	\caption{Von Mises sine density for $\mu_1 = \mu_2 =  \pi$, $\kappa_1 = \kappa_2 = 1$ and various $\kappa_3$.}
	\label{fig:vmsindenmodel}
\end{figure}

\begin{figure}[h]
	\centering
	\includegraphics[width=\linewidth]{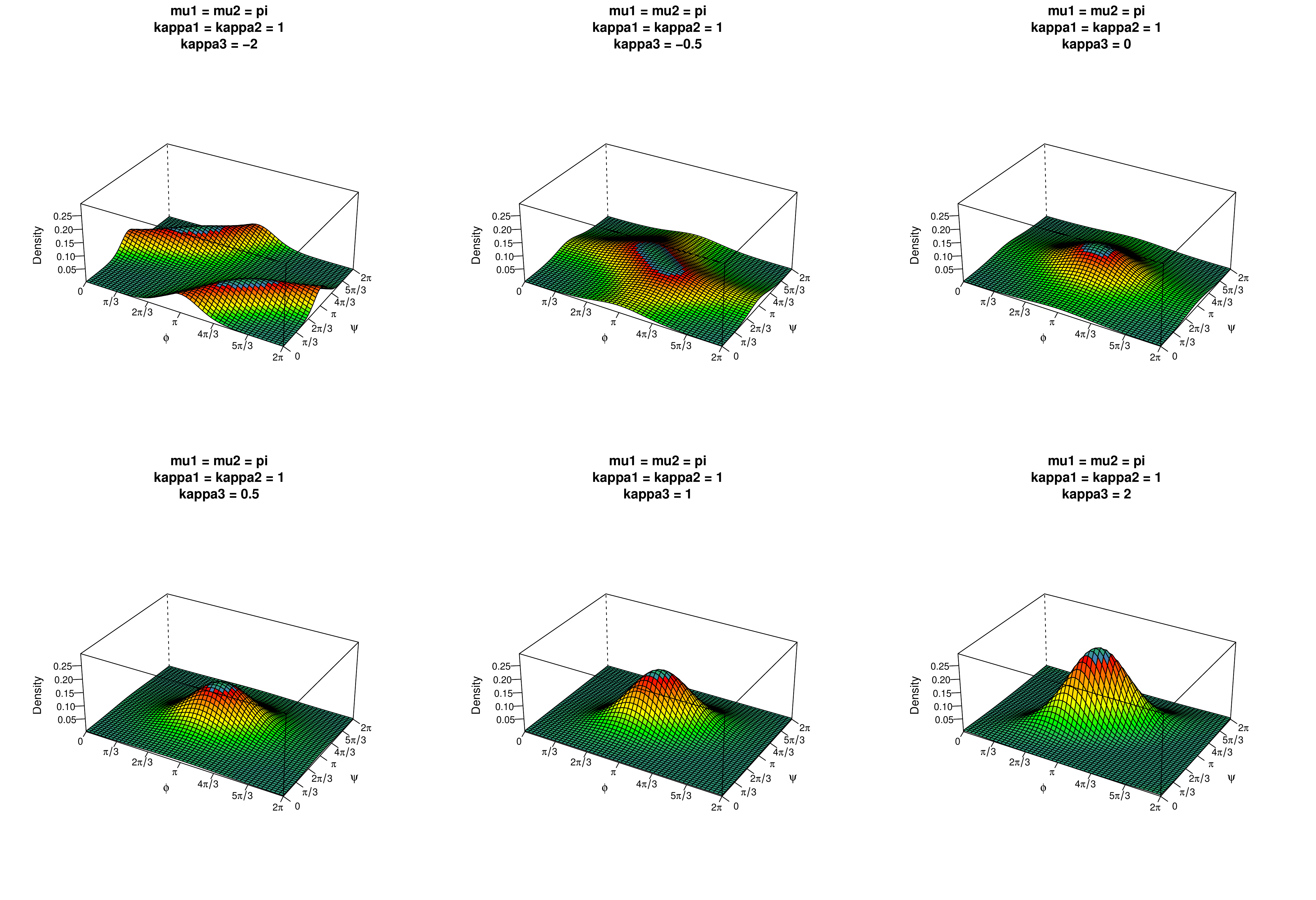}
	\caption{Von Mises cosine density for $\mu_1 = \mu_2 =  \pi$, $\kappa_1 = \kappa_2 = 1$ and various $\kappa_3$.}
	\label{fig:vmcosdenmodel}
\end{figure}

Figures~\ref{fig:vmsindenmodel} and \ref{fig:vmcosdenmodel} display the surfaces of the von Mises sine and von Mises cosine densities respectively with $\mu_1 = \mu_2 = \pi$, $\kappa_1 = \kappa_2 = 1$ and various $\kappa_3$'s. From Figure~\ref{fig:vmsindenmodel}, it can be seen that the density is bimodal when $\kappa_3 = \pm 2$ (or more generally for $|\kappa_3| \geq 1$ when $\kappa_1 = \kappa_2 = 1$), and unimodal when $|\kappa_3| < 1$. It can also be seen that the density surface (or the contours) of a sine model with $\kappa_3 = \xi$ is essentially a mirror image of that with $\kappa_3 = -\xi$, for any $\xi \in (-\infty, \infty)$; see, e.g., the upper-left and the lower-right panels of Figure~\ref{fig:vmsindenmodel}. Such is however, not the case for the cosine density, as depicted in Figure~\ref{fig:vmcosdenmodel}. The cosine density is bimodal when $\kappa_3$ is very negative ($\kappa_3 \leq -0.5$ when  $\kappa_1 = \kappa_2 = 1$, see, e.g., the upper-left and upper-middle panels of Figure~\ref{fig:vmcosdenmodel}), and is unimodal otherwise. Moreover, flipping the sign of $\kappa_3$ does not yield density surfaces (or contours) that are mirror images of each other. 

An interesting feature of both sine and cosine densities is that they both approximate the regular bivariate normal density (on $\R^2$) when the concentration parameters $\kappa_1$ and $\kappa_2$ are large, and the densities are unimodal (\citet[Section 2]{singh:2002}, \citet[Theorem 1]{mardia:2007}). This property is analogous to the univariate von Mises distribution. A heuristic explanation of this result again follows from the fact that when the distributions are unimodal and $\kappa_1, \kappa_2$ are large, then $\phi_1$ and $\phi_2$ are highly concentrated around $\mu_1$ and $\mu_2$. This means $\phi_i - \mu_i \approx 0$ so that $\sin(\phi_i - \mu_i) \approx (\phi_i - \mu_i)$ and $\cos(\phi_i - \mu_i) \approx 1 - (\phi_i - \mu_i)^2/2$ for $i=1,2$.

\subsection{Summary Measures for Univariate and Bivariate Angular Distributions} \label{summ_meas}
Circular summary measures are useful for describing various aspects of angular distributions. The circular mean or mean direction (see \cite{jammalamadaka:2001}) of an angular random variable $\psi$ is defined as
\[
E_c (\psi) = \arctan \left[ \frac{E(\sin \psi)}{E(\cos \psi)}\right]
\]
and the circular variance of $\psi$ is given by
\[
\var_c(\psi) = 1 - E[\cos (\psi - E_c(\psi))].
\]  
Note that $0 \leq \var_c(\psi) \leq 1$.

When considering the joint distribution of paired angular random variables $(\phi, \psi)$, their association can be measured using circular correlation. Multiple parametric circular correlation coefficients have been proposed in the literature, and here we consider two of them.  Let $\mu_1$ and $\mu_2$ be the circular means of $\psi_1$ and $\psi_2$ respectively. Then the Jammalamadaka-Sarma (JS) circular correlation coefficient \citep{jammalamadaka:1988} is defined as
\begin{equation} \label{rho_js_defn}
\rho_{\js}(\psi_1, \psi_2) = \frac{E\left[\sin (\psi_1 - \mu_1) \sin (\psi_2 - \mu_2) \right]}{\sqrt{E\left[\sin^2 (\psi_1 - \mu_1) \right] E\left[\sin^2  (\psi_2 - \mu_2)\right]}}.
\end{equation} 
Now let $(\psi_1^{(1)}, \psi_2^{(1)})$ and $(\psi_1^{(2)}, \psi_2^{(2)})$ be independent and identically distributed (IID) copies of $(\psi_1, \psi_2)$. Then the Fisher-Lee (FL) circular correlation coefficient \citep{fisher:1983} is defined by
\begin{equation} \label{rho_fl_defn}
\rho_{\fl}(\psi_1, \psi_2) = \frac{E\left[\sin \left(\psi_1^{(1)} - \psi_1^{(2)}\right) \sin \left(\psi_2^{(1)} - \psi_2^{(2)}\right) \right]}{\sqrt{E\left[\sin^2 \left(\psi_1^{(1)} - \psi_1^{(2)}\right) \right] E\left[\sin^2  \left(\psi_2^{(1)} - \psi_2^{(2)}\right) \right]}}.
\end{equation} 
Both $\rho_{\js}$ and $\rho_{\fl}$ have properties  similar to the ordinary correlation coefficient. In particular, $\rho_{\js}, \rho_{\fl} \in [-1, 1]$ and they are equal to $1 (-1)$ under perfect positive (negative) toroidal-linear (\emph{T-linear}) relationship \citep{fisher:1983,jammalamadaka:1988}.

Note that all distributions considered in \bambi have circular mean(s) equal to the respective mean parameter(s). For the univariate models, the circular variances are just functions of the associated concentration parameter (see \cite{mardia:jupp:2009}). In particular, if $\psi \sim \wn(\mu, \kappa)$ then $\var_c(\psi) = 1 - \exp(-\sigma^2/2)$ with $\sigma^2 = 1/\kappa$, and for $\psi \sim \vm(\mu, \kappa)$, $\var_c(\psi) = 1 - I_1(\kappa)/I_0(\kappa)$. For a bivariate wrapped normal model with  $\Sigma = (\Sigma_{ij}) = \Delta^{-1}$, the marginal circular variance of the first coordinate is $1 - \exp(-\Sigma_{11}/2)$, $\rho_{\fl} = \sinh(2\Sigma_{12}) /\sqrt{\sinh(2\Sigma_{11}) \sinh(2\Sigma_{22})}$ and $\rho_{\js} = \sinh(\Sigma_{12}) /
\sqrt{\sinh(\Sigma_{11}) \sinh(\Sigma_{22})}$ \citep{fisher:1983, jammalamadaka:1988}, where $\sinh$ denotes the hyperbolic sine function. For bivariate von Mises models (both sine and cosine forms), these expressions, provided in Appendix~\ref{appen_vmsin_vmocs_varcor}, are much more complicated, and involve infinite series of product of modified Bessel functions (see \cite{singh:2002, chakraborty:wong:2018:circular}). In \bambi we implement circular variances and correlation coefficients  for all the three bivariate models considered in this article. In addition, a function for calculating sample circular correlation coefficients is also provided, where the sample analogs of $\rho_{\js}$ and $\rho_{\fl}$, along with two other non-parametric circular correlation coefficients are considered (see Section~\ref{summ_model_fn}). 

\subsection{Mixture Models} \label{mixmodels}

Mixture models are convex combinations (\textit{mixtures}) of two or more probability distributions, and provide a semi-parametric approach to modeling complex datasets with multiple noticeably distinct clusters. Mixture models of both univariate and multivariate (non-wrapped) normal distributions are well studied in the literature (e.g., see \cite{lindsay:1995}), and implemented in many statistical packages, such as the \proglang{R} \citep{r} packages \pkg{mixtools} \citep{mixtools}, \pkg{mclust} \citep{mclust:2012, mclust:2002}, and \pkg{Rmixmod} \citep{pkg_Rmixmod}.  However, these are not applicable to mixture models for angular data.  This is a key motivation for our creation of \pkg{BAMBI}, which considers finite mixture models of univariate and bivariate angular distributions (the single function \code{fit_angmix} handles the fitting of all such models; see Section~\ref{sec_bambi_fn} Category~\ref{fit_fn}). 

Let $K$ denote the number of components (where $K$ is finite), $\{ f(\cdot \mid \thetab_j): j = 1, \cdots, K \}$ denote the component densities ($f$ can be univariate or bivariate) with $\thetab_j$ denoting the parameter vector associated with the $j$-th component, and let $ \pb = (p_1, \cdots, p_K)^T$  denote the vector of mixing proportions (or weights) with $p_j \geq 0$ and  $\sum_{j=1}^K p_j = 1$. Then the mixture density is defined as
\begin{align} \label{mix_generic}
\ft(\cdot \mid \pb; \thetab_1, \cdots, \thetab_K) = \sum_{j=1}^K p_j f(\cdot \mid \thetab_j)
\end{align}
In practice, the number of components $K$ necessary to fit the data is usually unknown,  and thus should be estimated on the basis of the data itself.  (See Section~\ref{compsizedet} for a discussion on number of components estimation.) 

An important special case of the general mixture model (\ref{mix_generic}) is the mixture of product components, also called a conditional independence model. Here, one assumes each multivariate component density $f(\cdot \mid \thetab_j)$ to be a product of univariate densities; specifically for the bivariate angular models considered in \pkg{BAMBI}, this is achieved by letting $\kappa_3=0$ in each component. Note that a mixture of product components \emph{does not} imply independence in the final mixture density. In fact, such a model can reasonably approximate a wide class of more complicated models, while being computationally less involved (see \cite{grim:2017});  however, one often needs a larger $K$ compared to a general (non-product) mixture model to achieve similar results, thus offsetting some of the potential computational gains. In \bambi a  product component mixture can be fitted via \code{fit_angmix} by setting the argument \code{cov.restrict = "ZERO"} (see Section~\ref{sec_bambi_fn} under Category~\ref{fit_fn}).

It is also noteworthy to mention the aspect of bimodality of bivariate von Mises distributions in the context of mixture modeling. In practice, often each  component of a mixture model is used to represent one single (unimodal) cluster in data. However, as discussed in Section~\ref{vmmodels}, both von Mises sine and cosine models can be bimodal depending on the values of the concentration and association parameters. When bimodality is present in some of the component specific densities, the final mixture model can be harder to interpret. To avoid this issue, it is possible to restrict the parameter spaces associated with the concentration and association parameters (by letting $\kappa_3^2 < \kappa_1\kappa_2$ in the sine model, and $\kappa_3 \geq -\kappa_1\kappa_2 / (\kappa_1+\kappa_2)$ in the cosine model) in these angular models to force unimodality in each component specific density.  Consequently, a larger $K$ may be needed to achieve similar results, which increases model complexity. In  \bambi we provide an option of having only unimodal von Mises component densities. This is achieved by setting the logical argument \code{unimodal.component = TRUE} in \code{fit_angmix} (defaults to \code{FALSE}). See the  discussion in Section~\ref{sec_bambi_fn}, Category~\ref{fit_fn}.   

\subsection{Related Work and Motivation for BAMBI}

\subsubsection{Literature}
Several papers have addressed inferential problems relating to mixtures of bivariate angular distributions.  \cite{mardia:2007} consider the mixture of bivariate von Mises cosine distributions, and suggest an EM algorithm for frequentist estimation of the associated parameters. Their approach is used in \cite{boomsma:2008} in the context of modeling protein backbone angles. In other work,  \cite{lennox:2009} consider a Bayesian non-parametric model involving an infinite mixture of von Mises sine distributions. In \bambi we focus on classical finite mixtures, providing a unifying framework for Bayesian estimation of all three bivariate angular models presented earlier.


\subsubsection{Software}

To the best of our knowledge, no previous packages or libraries handle finite mixture modeling for univariate or bivariate angular data, whether in \proglang{R} or otherwise. In fact, the only available software (as of the time of writing this manuscript) that  has functionality for bivariate von Mises models  is the C++ library \proglang{mocapy++} in the context of Dynamic Bayesian Networks  \citep[][mentioned in \citealp{mardia:2007}]{paluszewski:2010}. However, \proglang{mocapy++} does not implement bivariate wrapped normal models.  


The overarching goal of \pkg{BAMBI} is to create a unified platform that implements descriptive and  inferential statistical tools required to analyze bivariate and univariate angular data. First, \bambi provides functions for  density evaluation, computation of various summary measures (such as circular mean, variance and correlation coefficient), and random data generation from bivariate and univariate angular models and their mixtures. Second, it has functions for fitting these models to real angular data using Bayesian methods. Third, it implements a number of post-processing steps required in any Bayesian statistical analysis. For example, visual and numerical assessment of the goodness of fits can be done using a number of native \bambi functions, as well as \pkg{coda} package functions, which are applicable on \bambi outputs (\code{angmcmc} objects) through a convenient  \code{as.mcmc.list} method.  Furthermore, \bambi has functions for model selection as well as random data generation and density evaluation from fitted models, which are useful in posterior predictive analyses.



It is to be noted that while it is possible to use general-purpose MCMC samplers such as \proglang{stan} \citep{stan_rpack}, \proglang{JAGS} \citep{plummer:2003} and \proglang{WinBUGS} \citep{lunn:2000} for fitting the angular mixture models considered in \pkg{BAMBI}, there are important motivations for developing specialized implementations for these models.
First, special care needs to be taken while handling the normalizing constants in the von Mises sine and cosine densities, which contain infinite series of product of Bessel functions that can be numerically unstable for some ranges of parameter values; such cases are handled in \pkg{BAMBI} via (quasi) Monte Carlo approximations. Second, computations for Bayesian mixture modeling benefit from using a latent allocation structure, as done in \pkg{BAMBI} (see Section~\ref{gs}), which allows independent sampling of  the component specific parameters.  Such an approach cannot be used in \proglang{stan} due to the discreteness of the allocation (p. 79, Section~6.2 of the reference manual v2.18.0); consequently, the dimensionality of the parameter vector can hinder convergence of MCMC sampling for mixtures with many components. 
In contrast, \proglang{JAGS}/\proglang{WinBUGS} allows incorporation of discrete latent allocation; however, 
their sampling techniques 
do not make use of the gradient of the target (log) posterior density. As discussed in Section~\ref{hmc}, Hamiltonian Monte Carlo uses the gradient  and hence is typically more efficient for sampling from intractable distributions.
Finally, the  analytic gradients necessary for efficient MCMC sampling in these models are built into \pkg{BAMBI}.

\subsection{Organization of the Paper}

The remainder of this article is organized as follows. In Section~\ref{meth}, we  review Bayesian methods for fitting angular mixture models to data. In Section~\ref{bambipack} we describe the capabilities of \pkg{BAMBI}, by describing all functions and datasets available in \pkg{BAMBI}, and providing brief overviews on their usage. Following, in Section~\ref{illus} we illustrate angular mixture modeling on datasets included in \pkg{BAMBI}.  The paper concludes with a brief summary and possible directions for future development in Section 5. A derivation for the von Mises cosine model normalizing constant, formulas for circular variances and correlation coefficients in the von Mises sine and cosine models, analytic forms of gradients needed for efficient MCMC sampling (discussed in Section~\ref{hmc}), and MCMC parameter traceplots associated with one of the examples considered in Section~\ref{illus} are  provided in the Appendices.

\section{Methods} \label{meth}

\subsection{Overview}
We adopt a Bayesian approach for fitting angular mixture models to data. Let $\Psib^\top = (\psib_1, \cdots, \psib_n)$ be the  data matrix (or  data vector in the univariate case)  with each $\psib_i$ being a bivariate vector of angles (or a univariate angle) $[0, 2\pi)^2$ (or in  $[0, 2\pi)$).  We are interested in fitting a mixture density of the form (\ref{mix_generic}) for a given number of components $K$. For example, in bivariate wrapped normal mixtures, the density for the $j$-th component is given by $f_j  \equiv f_{\bwn} (\cdot \mid \thetab_j) =: f_{\bwn, j}$, where $\thetab_j^\top = (\kappa_{1j}, \kappa_{2j}, \kappa_{3j}, \mu_{1j} ,\mu_{2j})$ denotes the vector of (model) parameters for the $j$th component, $j = 1, \cdots, K$, and the mixture density is given by $\ft_{\bwn} = \sum_{j=1}^K p_j f_{\bwn, j}$. For a specified $K$, our objective is to estimate the parameter vector $\etab^\top = (\thetab^\top, \pb^\top)$,  which consists of the model parameters $\thetab^\top = (\thetab^\top_1, \cdots, \thetab^\top_K)$ and the mixing proportions $\pb^\top = (p_1, \cdots, p_K)$, based on $\Psib$. Often, $K$ itself will also need to be estimated. In the following, we review some commonly used techniques in Bayesian mixture model fitting.

\subsection{Bayesian mixture modeling} \label{bayesintro}

Under a Bayesian framework a prior distribution must be specified for the parameter vector, which can be non-informative (or diffuse) if \textit{a priori} information is unavailable. Let $\pi(\thetab, \pb)$ denote the joint prior density for $\etab$. Often the prior distributions of $\thetab$ and $\pb$ are assumed to be independent so that (with a slight abuse of notation; here $\pi(y)$ stands for the appropriate prior density of the random variable $y$) $\pi(\thetab, \pb) = \pi(\thetab) \pi(\pb)$. Moreover, parameters from different components are often assumed to be independent, so that $\pi(\thetab) = \prod_{j=1}^{K} \pi(\thetab_j)$. Let $L(\Psib \mid \thetab, \pb) = \prod_{i = 1}^n \ft(\psib_i \mid \thetab, \pb)$ denote the likelihood function of the data. Then the posterior density of $\etab$ given the data is
\begin{align} \label{post_den_gen}
\pi(\thetab, \pb \mid  \Psib) \propto L(\Psib \mid \thetab, \pb) \: \pi(\pb) \prod_{j=1}^{K} \pi(\thetab_j),
\end{align}
which is the basis for Bayesian inference on $\etab$. It is to be noted that the prior densities $\pi(\thetab_j)$'s all need to be proper in order to ensure that the posterior density $\pi(\thetab, \pb \mid  \Psib)$ is proper \citep[see, e.g.,][Section~2.2]{diebolt:1994}.   Specific comments about the choice of priors used in the current setting are provided in Section~\ref{prior}. Note that the associated posterior mean, median or mode, commonly used as point estimates of the parameters, are not available  in closed form for our distributions of interest.  Additionally, $\pi(\thetab, \pb \mid \Psib)$ is intractable for directly simulating IID samples, and thus some kind of Markov Chain Monte Carlo (MCMC) technique is used in practice as an alternative.  Starting from some initial point, an MCMC algorithm generates a Markov chain which has the target posterior density $\pi(\thetab, \pb \mid \Psib)$ as the invariant distribution. Various summary measures of the posterior distributions -- such as mean, mode (known as the \emph{maximum a posteriori} or MAP parameter value), and quantiles --  can then be approximated based on the MCMC realizations. In practice, the MCMC algorithm must be run long enough for the Markov chain to converge, so that the realizations approximately follow the target posterior distribution.  For this purpose the chain is given a \emph{burn-in} period, where the initial iterations are discarded.




In \bambi the function \code{fit_angmix} fits a Bayesian angular mixture model with a specified number of components, and the function \code{fit_incremental_angmix} fits angular mixtures with incremental number of components to determine an optimum number of components. In the following we briefly review the MCMC generation techniques \textit{Gibbs sampler (GS)}, \textit{Metropolis Hastings} and \textit{Hamiltonian Monte Carlo (HMC)}, and describe how they are used for sampling from the posterior distributions of model parameters and mixing proportions in these two \bambi functions.

\subsection{Gibbs sampler (GS)} \label{gs}

The Gibbs sampler (GS) \citep{geman:1984, gelfand:1990} breaks the Markov chain updates for the parameter vector into blocks.  For example, when $\etab = (\etab_1, \etab_2)$ the GS generates the $N$-th state of the Markov Chain $(\etab_1^{(N)}, \etab_2^{(N)})$ from the previous state $(\etab_1^{(N-1)}, \etab_2^{(N-1)})$ with the steps

\begin{enumerate}
\item Generate $\etab_1^{(N)}$ from $\pi(\etab_1 \mid \etab_2^{(N-1)},\dat)$.
\item Generate $\etab_2^{(N)}$ from $\pi(\etab_2 \mid \etab_1^{(N)},\dat)$.
\end{enumerate}
The GS is most effective when it is easy to sample from the (full) conditional posterior densities $\pi(\etab_1 \mid \etab_2,\dat)$ and $\pi(\etab_2 \mid \etab_1,\dat)$.  Note that when $\etab_1$ and $\etab_2$ are vectors, this is sometimes called the blocked Gibbs sampler.

For mixture models, an efficient Gibbs sampling step for the mixing proportions $\pb$ (when $K > 1$) can be obtained by adopting a so-called Data Augmentation scheme, where one introduces (``augments'') unobserved data to make the conditional distributions simpler \citep{diebolt:1994}. Here, we introduce (hidden) component indicators $\zetab_i^\top = (\zeta_{i1}, \cdots, \zeta_{iK})$ corresponding to each observation $\psib_i$ where $\zeta_{ij}$ is 1 if the $i$th observation comes from the $j$th component, and 0 otherwise, for $i = 1, \cdots,n$ and $j=1,\cdots,K$. Thus, given $\zeta_{ij} = 1$, the density of $\psib_i$ is simply $f\left(\psib_i \mid \thetab_j\right)$, i.e., the density in the $j$-th component evaluated at $\psib_i$. Moreover,  $n_j := \sum_{i=1}^n \zeta_{ij}$ is the total number of observations coming  from this density. It is customary to assume a Dirichlet$(\alphab)$ prior for $\pb$, where $\alphab^\top = (\alpha_1, \cdots, \alpha_K)$ with $\alpha_j > 0$ for all $j$, so that $\pi(\pb) \propto \prod_{j=1}^K p_j^{\alpha_j-1}$. Note that $\alpha_j = 1$ for all $j$ represents the uniform prior. Let $Z^\top = (\zetab_1^\top, \cdots, \zetab_n^\top)$ and let $\thetab^{(N-1)}$, $\pb^{(N-1)}$ and $Z^{(N-1)}$ be the $(N-1)$th MCMC realizations of $\thetab$, $\pb$ and $Z$ respectively. Then the $N$th realization of $\pb$ (and $Z$) are obtained as follows:
\begin{enumerate}
\item For $i = 1, \cdots,n$, generate $\zetab_i^{(N)}$ from Multinomial$\left(1; \pt_{i1}^{(N-1)}, \cdots, \pt_{iK}^{(N-1)}\right)$ independently, and define $n_j^{(N)} := \sum_{i=1}^n \zeta_{ij} ^{(N)}$, where
\begin{equation} \label{mem_probs}
\pt_{ij}^{(N-1)} = {\frac{p_j^{(N-1)} f\left(\psib_i \mid \thetab_j^{(N-1)}\right)}{\sum_{h=1}^K p_h^{(N-1)} f\left(\psib_i \mid \thetab_h^{(N-1)}\right)} }
\end{equation}
are the posterior membership probabilities.

\item Generate $\pb^{(N)}$ from Dirichlet$\left(\alpha_1+n_1^{(N)}, \cdots, \alpha_K+n_K^{(N)}\right)$.
\end{enumerate}

Thus when $K > 1$, the latent allocation $\zetab_i$'s generated during the Gibbs sampling step for $\pb$ leads to simplifications that reduce the computational burden substantially.  Note that, conditional on $\zetab_i$'s, all $\psib_i$'s have independent single component densities $f(\cdot  \mid \thetab_{j_i})$, with $j_i$ being the non-zero position of $\zetab_i$. Thus, given $\zetab_i$'s, all $\thetab_j$'s are independent with only data points coming from component $j$ contributing to the respective likelihoods. Consequently $\thetab_j$'s can be sampled independently from their (component-specific) conditional posterior densities.

To complete the GS scheme for the mixture model, it remains to sample $\thetab_j$'s from $\pi(\thetab_j \mid Z, \Psib, \pb)$. As these distributions are still intractable for direct IID simulation, we use a Markov chain simulation technique for sampling, and then combine this step with the GS updates for $\pb$ and $Z$. In the following we describe two such Markov chain simulation techniques, and how they are used in \pkg{BAMBI}.

\subsection{Metropolis-Hastings Algorithm} \label{rwmh}

The Metropolis-Hastings algorithm  \citep{metropolis:1953, hastings:1970}  is simple and widely-used for Markov chain simulation. 
Formally, let $x$ be the current state of a Markov chain $\Phi$ with stationary density $q$. Let $\qt(\cdot \mid x)$ be a proposal density defined on the state space of $\Phi$ that is easy to sample from.  Then the next state $x'$ of the Markov chain $\Phi$ is obtained as follows:
\begin{enumerate}
	\item Generate $x^*$ from $\qt_x$.
	\item Define $r(x^*, x) = \min\left\{1, \frac{q(x^*)}{q(x)} \: \frac{\qt(x \mid x^*)}{\qt(x^* \mid x)} \right\}$, and define the next state $x'$ equal to $x^*$ with probability $r(x^*, x)$ and equal to $x$ with probability $1-r(x^*, x)$.
\end{enumerate}

The random walk variant of Metropolis-Hastings (RWMH) uses a proposal density $\qt(\cdot \mid x)$ that is symmetric about $x$; e.g., by taking $\qt(\cdot \mid x)$ to be the density of $Y_x = x + Y_0$,  where $Y_0$ is a normal random variable with mean zero. Under RWMH, $\qt(x \mid x^*) = \qt(x^* \mid x)$, and hence $r(x^*, x) = \min\left\{1, \frac{q(x^*)}{q(x)} \right\}$, thus simplifying computations. In \pkg{BAMBI}, RWMH is implemented with independent normal proposals.

Note that the variance of the density $\qt(\cdot \mid x)$ strongly affects the acceptance probabilities $r(x^*, x)$. Convergence of the Markov chain will be slow if the variance of $\qt(\cdot \mid x)$ is too large or too small. \cite{roberts:2001} suggest maintaining an acceptance rate of 20-30\% as a general rule-of-thumb. In \bambi we provide an auto-tuning feature that implements adaptive tuning during the burn-in period.  Briefly, the acceptance rate and scale of the sampled parameters  are monitored at regular intervals, and the proposal variances are adjusted accordingly (see the documentation of \code{fit_angmix} for details).  We limit adaptation to the burn-in period, so that the desired properties of the final MCMC samples are retained.


\subsection{Hamiltonian/Hybrid Monte Carlo (HMC)} \label{hmc}
Simple RWMH can become quite inefficient in multi-dimensional problems.  A powerful alternative to RWMH when the gradient of the posterior density has an analytical form is  Hamiltonian (also called \textit{Hybrid}) Monte Carlo (HMC) \citep{duane:1987, neal:1996}.  HMC makes use of the gradient of the log posterior density and an auxiliary random variable, and incorporates tools from molecular dynamics to furnish proposal states coming from high posterior density regions. This allows a much faster exploration of the state space than a RWMH scheme. A gentle and detailed introduction to HMC with applications to statistical problems can be found in \cite{neal:2011}.  Briefly, in HMC first an auxiliary random variable $\rb$ called \textit{momentum} is considered along with the variable of interest (vector of model parameters $\thetab$ in our case), which is classically called the \textit{position} in physical problems, denoted by $\qb$.\footnote{In classical HMC literature, the auxiliary variable is denoted by $\pb$; however, we will keep that notation for mixing proportions.} Furthermore, two energy functions $U(\qb)$  and $K(\rb)$ are introduced, followed by a Hamiltonian function $H(\qb, \rb)$ which is usually the sum of those two energies, i.e., $H(\qb, \rb) = U(\qb) + K(\rb)$. $U(\qb)$, called the  \textit{potential energy}, is defined as the negative log posterior density of $\qb$ (plus any fixed constant) in MCMC applications, and  $K(\rb)$, called the \textit{kinetic energy}, is usually defined as $K(\rb) = \rb^\top M^{-1} \rb$ for some fixed positive definite matrix $M$.  This form for $K(\rb)$ corresponds to the negative log  density (plus a constant) of the zero-mean normal distribution with variance matrix $M$. In practice, $M$ is typically taken to be diagonal, often the identity matrix (as used in \pkg{BAMBI}), or a scalar multiple of the identity matrix. Let $\nabla U(\qb)$ denote the gradient vector of $U(\qb)$ with respect to $\qb$. Further, let $\epsilon > 0$ be a small real number, called the \textit{step-size}, and $L\geq 2$, a positive integer, called the number of  \textit{leapfrog steps}.  Then one step of HMC that updates (via \textit{leapfrog} method) the current state $\qb$ to the next state $\qb'$ can be described as follows:
\begin{enumerate}
	\item Generate $\rb$ from N$(\bm{0}, M)$ and let $\qb^{(0)} = \qb$ and $\rb^{(0)} = \rb - (\epsilon/2) \nabla U(\qb^{(0)})$.
	\item  For $t = 1, \cdots, L$ define $\qb^{(t)} = \qb^{(t-1)} + \epsilon \: \rb^{(t-1)}$ and $\rb^{(t)} = \rb^{(t-1)} - (\epsilon/\gamma_l) \nabla U(\qb^{(t)})$, where $\gamma_l = 1$ for $l = 1, \cdots, L-1$ and $\gamma_L = 2$.
	\item Let $\qb^* = \qb^{(L)}$ and $\rb^* = -\rb^{(L)}$, and define $\beta(\qb^*,\rb^*; \qb,\rb) = \min \left \lbrace 1,  \exp \left[ H(\qb^*, \rb^*) - H(\qb, \rb) \right] \right \rbrace$.
	\item Finally, define the new state $\qb'$ equal to $\qb^*$ with probability $\beta(\qb^*,\rb^*; \qb,\rb)$, and equal to $\qb$ with probability $1-\beta(\qb^*,\rb^*; \qb,\rb)$.
\end{enumerate}
Special care needs to be taken for the cases where the variables being sampled are constrained: for our angular models, $\mu_i$'s are angles in $[0,2\pi)$, and the (raw) concentration parameters are positive. See  section 5.5.1.5 of \cite {neal:2011}  for more details.

Since HMC approximates the dynamics by discretization, the step-size $\epsilon$ needs to be sufficiently small for the proposals to have a high acceptance rate. However,  if $\epsilon$ is too small, convergence of the Markov chain will be slow. Thus, $\epsilon$ requires tuning to obtain a reasonable acceptance rate ($\sim$40-90\%, with 65\% being optimal, as suggested by \cite{neal:2011}).   In \bambi we provide an auto-tune feature for $\epsilon$ similar to the one for the proposal standard deviation in RWMH (see Section~\ref{rwmh}), which adaptively tunes $\epsilon$ during burn-in to ensure a reasonable acceptance rate (60-90\% by default).

Care is required for choosing the number of leapfrog steps $L$, since a $L$ that is too large or too small can lead to poor convergence.  While setting an appropriate $L$ can be challenging for high dimensional parameter vectors, here the independence of components $\pi(\thetab_j \mid Z, \Psib, \pb)$ means that only two (for univariate models) or five (for bivariate models) parameters need to be sampled at a time.  Thus, the default $L=10$ used in \pkg{BAMBI}, which works well empirically, suffices for mixtures with any number of components. As suggested in \cite{neal:2011, Mackenzie:1989}, randomly choosing $\epsilon$ and $L$ from some fairly small interval at the beginning of every HMC step may improve convergence of the chain. In \bambi $\epsilon$ is by default randomly chosen at each iteration from an interval of the form $(\epsilon_0(1-\delta), \epsilon_0(1+\delta))$ for a fixed $\epsilon_0 > 0$ (can be auto-tuned in \pkg{BAMBI}) and a given $\delta \in (0, 1)$, while $L$ is kept fixed. However, these settings can be changed; in particular, $L$ can also be randomly chosen  from the set of integers contained in an interval $(L_0/\exp(A), L_0 \exp(A))$ for some given $L_0 > 0$ and $A > 0$, or both $\epsilon$ and $L$ can be specified to be non-random. See the documentation of \code{fit_angmix} for more details. 

When properly tuned, HMC can achieve faster convergence and better exploration of the target density than RWMH, for a similar computational cost. Note that the computational cost for each HMC iteration is higher due to $L$ gradient evaluations, however, HMC usually requires fewer iterations to reach stationarity and successive samples have lower autocorrelation.  Hence, HMC is our recommended sampling approach in \pkg{BAMBI}.  HMC, while powerful, does not solve all the challenges associated with MCMC sampling algorithms; in particular, both RWMH and HMC can get trapped in local modes. One possible remedy is to use multiple independent chains, see Section~\ref{initial_vals}.

By default, \bambi uses HMC to sample $\thetab$.  All angular densities considered here,  both univariate and bivariate, admit analytic gradients for efficient programming implementation. Expressions for the conditional log posterior density and its gradients are provided in the following section.

\subsection{Using RMWH or HMC for angular mixture models}
Consider the mixture model (\ref{mix_generic}) with density $f(\cdot \mid \thetab_j)$ for the $j$-th component, $j=1,\cdots,K$. It follows that given the component indicators $Z$, information on $\pb$ is superfluous, and the complete-data (i.e., given $\Psib$ and $Z$) likelihood for $\thetab = (\thetab_1, \cdots, \thetab_K)$ is given by:
\[
\mathrm{likelihood}(\thetab \mid Z, \Psib) \propto \prod_{i = 1}^n \prod_{j=1}^K  f(\psib_i  \mid \thetab_j)^{\zeta_{ij}}.
\]
Recall that the joint prior density of $\thetab$ is $\prod_{j=1}^K \pi(\thetab_j)$. Hence, the complete-data posterior density of $\thetab$ is given by:
\[
\pi(\thetab \mid Z, \Psib) \propto \left\{\prod_{i = 1}^n \prod_{j=1}^K  f(\psib_i  \mid \thetab_j)^{\zeta_{ij}}\right\} \prod_{j=1}^K \pi(\thetab_j).
\]
Therefore, by taking the logarithm, the complete-data log posterior density (LPD) for $\thetab = (\thetab_1, \cdots, \thetab_K)$ given the component indicators $Z$ is obtained as
\begin{align*} \label{lpd_gen_full}
\lt_{\text{complete-data}}(\thetab) := \log \pi(\thetab \mid Z, \Psib) &= C + \sum_{i=1}^n \sum_{j=1}^{K} \zeta_{ij} \log f(\psib_i  \mid \thetab_j) + \sum_{j=1}^{K} \log \pi(\thetab_j) \\
&= C + \sum_{j=1}^{K} \left\{  \sum_{i=1}^n \zeta_{ij}  \log f(\psib_i  \mid \thetab_j) + \log \pi(\thetab_j) \right\} \\
&= C + \sum_{j=1}^{K} \left\{  \sum_{i \colon \zeta_{ij} = 1}  \log f(\psib_i  \mid \thetab_j) + \log \pi(\thetab_j) \right\} \numbereqn
\end{align*}
where $C$ is a constant free of $\thetab$. The above expression shows that conditional on $Z$, $\thetab_j$'s are independent, and that the complete-data log posterior density of $\thetab_j$ is of the form
\begin{align} \label{lpd_gen}
\lt_j(\thetab_j) = C_j + \sum_{i \colon \zeta_{ij} = 1}  \log f(\psib_i  \mid \thetab_j) + \log \pi(\thetab_j)
\end{align}  
where $C_j$'s are  constants (free of $\thetab$).  Given the current GS draw of $Z$, samples from the conditional posterior density $\lt_j$ in (\ref{lpd_gen}) can therefore be drawn independently for all $j=1,\cdots, K$. For each $j \geq 1$,  we let $\lt_j$ play the role of the target density $q$ (see Section~\ref{rwmh}) in RWMH, or let $-\lt_j$ play the role of  the potential energy $U$ (see Section~\ref{hmc}) in HMC; the gradient of $U$ with respect to $\thetab_j$, $\nabla U$, is therefore the negative of the gradient $\nabla \lt_j$. From (\ref{lpd_gen}), it follows that
\begin{subequations}
\begin{align}
\nabla \lt_j(\thetab) &=  \left( \sum_{i \colon \zeta_{ij} = 1}  \frac{\partial \log f(\psib_i  \mid \thetab_j)} {\partial \thetab_j}\right) + \frac{\partial \log \pi(\thetab_j)} {\partial \thetab_j} \label{grad_lpd_gen_logform} \\
&=  \left(\sum_{i \colon \zeta_{ij} = 1} \frac{1}{f(\psib_i  \mid \thetab_j)} \cdot \frac{\partial f(\psib_i  \mid \thetab_j)} {\partial \thetab_j}\right) + \frac{\partial \log \pi(\thetab_j)} {\partial \thetab_j} \label{grad_lpd_gen_origform}
\end{align}	 
\end{subequations}
For the von Mises distributions (both univariate and bivariate), form (\ref{grad_lpd_gen_logform}) is easier to work with, whereas form (\ref{grad_lpd_gen_origform}) is more useful for the wrapped normal distributions. Full analytic expressions for all model specific gradients are provided in Appendix~\ref{grads}.

Note that parameters with a non-negative support are often sampled more efficiently on the log scale; we use this strategy for sampling the concentration parameters $\kappa$ (in univariate models) and $\kappa_1, \kappa_2$ (in bivariate models). 


\subsection{Choice of priors} \label{prior}
Selection of prior constitutes an important step in Bayesian analyses, as they play a key role in the final inference. This is comparatively more standard for the component-specific model parameters $\thetab$. As discussed, proper prior distributions for the model parameters are required to ensure posterior propriety.  For the mean parameters $\mu$ (in univariate models) and $\mu_1, \mu_2$ (in bivariate models), their prior distributions can be taken to be a member of the same family of the distribution which are being used in the mixture model (e.g., von Mises sine prior for $(\mu_1, \mu_2)$ in a von Mises sine mixture model) to aid conjugacy. \cite{lennox:2009} use this conjugate prior for the mean parameters in their von Mises sine (infinite) mixture model. Note that conjugacy for the mean parameter is not achievable except in trivial cases in the wrapped normal distributions (both univariate and bivariate).  In \bambi we set a uniform prior over $[0, 2\pi)$ (if univariate) or $[0, 2\pi)^2$ (if bivariate)  for the mean parameter(s), which can be viewed as a special case of the von Mises and wrapped normal distributions (see Sections~\ref{wnmodels} and \ref{vmmodels}). Conjugacy is also possible for the concentration and association parameters, e.g., \cite{lennox:2009} consider such a family for von Mises sine model. However, that approach does not aid sampling, as the resulting unnormalized densities involve infinite sums of products of modified Bessel functions.  
As a simple alternative, we suggest using independent normal distributions with zero mean as the prior for the association parameter $\kappa_3$, as well as for the log of the concentration parameters $\kappa$, $\kappa_1$, and  $\kappa_2$ (i.e.,  the prior for concentration parameters are log normal). These prior distributions can be made informative or diffuse through appropriate choices of the variance hyper-parameter. Priors are assigned independently to each parameter, and truncation is performed to reflect any specified constraints in the model (such as $\kappa_3^2 < \kappa_1 \kappa_2$ in a bivariate wrapped normal model, and a von Mises sine model with unimodal density). 

Care is required in the selection of prior for the mixing proportions $\pb$, as an ill-chosen prior may result in very poor fits. This is particularly true when $K$ is too large (i.e., the mixture is overfitted). Note that overfitting is a necessary step when the true number of components is unknown and needs to be estimated, see Section~\ref{compsizedet} for more details.  
It is customary to assume a Dirichlet$(\alphab)$ prior for $\pb$, where $\alphab^\top = (\alpha_1, \cdots, \alpha_K)$ with $\alpha_j > 0$, often with the special case  $\alpha_j = \alpha_0$ for all $j$. When the mixture is overfitted, the asymptotic results in \cite{rousseau:2011} show that $\alpha_j$'s strongly influence how the spurious mixture components are handled by the limiting posterior density. In particular, if $\max_{j=1,\cdots, K} \alpha_j < d /2$, where $d = \dim \thetab_j$, then the spurious components vanish asymptotically. On the other hand, if $\min_{j=1,\cdots, K} \alpha_j > d /2$, then the spurious components asymptotically get  superimposed on some of the existing components with positive mixing proportions \citep[Section~10.3.1]{fruhwirth:2011}. The subsequent estimation of $K$ depends on which way the overfitting is handled by the posterior density (see Section~\ref{compsizedet}); thus $\alpha_j$'s all need to be appropriately either small or large \citep[Section~10.3.2]{fruhwirth:2011}.  A uniform prior with $\alpha_j=\alpha_0 = 1$ is a rather poor choice in this regard. In \bambi estimation of $K$ is done assuming the use of $\alpha_j > d/2$ for all $j$ in conjunction with a model selection criterion;  our default is $\alpha_j = \alpha_0 =  (r +r(r+1)/2)/2 + 3$ as used in \citet[Section~10.3.4]{fruhwirth:2011}, where $r$ denotes the dimension of the data, i.e., $r$ is 1 or 2 according as whether the model is univariate or bivariate (and consequently, all $\alpha_j$'s are either $4$ or $5.5$).

\subsection{Estimating the number of components $K$ from data} \label{compsizedet}

Suppose the data were generated from a mixture of $K_\tru$ (non-empty, non-identical) components. In practice, $K_\tru$ will not be known, and therefore mixture modeling requires estimating the appropriate number of components from the data.

In the Bayesian setting, the estimation of $K_\tru$ requires an overfitted mixture model, i.e., one that has spurious or superfluous components. There are two ways of introducing superfluous components to overfit a mixture model, and the subsequent estimation of $K_\tru$ should reflect which way is taken. First, the superfluous components can be arbitrarily introduced at regions with no data points (``leave some groups empty''), and assigned zero mixing proportions. Then, the number of non-empty components in the fitted mixture provides a good estimate of $K_\tru$. Second, the spurious components can be superimposed on some of the existing components (``let two component-specific parameters be identical''), and assigned positive mixing proportions. Here, the total number of components in the fitted mixture, after accounting for model complexity (via some model selection criterion), provides a reasonable estimate for $K_\tru$.  Note that the prior distribution of the mixing proportion $\pb$ affects the way overfitting is handled by the posterior, and hence the associated prior hyper-parameters need to be wisely chosen (see Section~\ref{prior}). A detailed discussion on the estimation of the number of components can be found in \citet[Section~10.3.1]{fruhwirth:2011}. 

In \bambi we assume that the superfluous components are introduced in the second (``let  two component-specific parameters be identical'') way. Consequently, $K_\tru$ is estimated by first incrementally fitting the data with one additional component (starting from $K=1$), until a model with $K+1$ component fails to improve upon the previous fit with $K$ component (as determined through a model selection criterion); that value of $K$ is then used as an estimate of $K_\tru$. There exist multiple model selection criteria in the literature; we review six such criteria implemented in \pkg{BAMBI} and comment on their applicability in MCMC simulations. In the following, $\etab = (\thetab, \pb)$ denotes the entire parameter vector, and $\yb =(y_1, \ldots, y_n)$ is the vector/matrix of $n$ independent  observations.

\begin{enumerate}
    	\item \textit{Watanabe-Akaike Information Criterion (WAIC)} (\cite{watanabe:2013, gelman:2014:paper}). Given the dataset $y_1, \cdots, y_n$, the Markov chain realizations $\{\etab_1, \cdots, \etab_N \}$ of the parameter vector, and the pointwise densities $\{p(y_i \mid \etab_s): i = 1, \cdots, n; \: s = 1, \cdots, N \}$, define the computed log pointwise posterior predictive density
    	\[
    	\text{LPPD} = \sum_{i=1}^n \log \left(\frac1N \sum_{s=1}^N p(y_i \mid \etab_s) \right).
    	\]
    	Then WAIC is defined as
    	\[
    	\text{WAIC} = \text{LPPD} - p_W
    	\]
    	where $p_W$ is a correction term to adjust for effective number of parameters. Two forms for the adjustment terms are proposed in the literature, both being approximations based on Bayesian cross validation. In the first approach, (computed) $p_W$ is defined as
    	\[
    	p_W = 2 \sum_{i=1}^n \left[ \log \left(\frac1N \sum_{s=1}^N p(y_i \mid \etab_s) \right) -\frac1N \sum_{s=1}^N \log \:p(y_i \mid \etab_s)  \right]
    	\]
    	whereas, in the second approach, (computed) $p_W$ is defined by $p_W = \sum_{i=1}^n \widehat{\text{var}}\: \log \: p(y_i \mid \etab)$, where for $i = 1, \cdots, n$, $\widehat{\text{var}}\: \log \: p(y_i \mid \etab)$ denotes the estimated variance of $p(y_i \mid \etab)$ based on the realizations $\etab_1, \cdots, \etab_N$.
    	
    	\item \textit{Leave One Out Cross Validation Information Criterion (LOOIC) \citep{vehtari:2017}}. Under the same set-up as WAIC, the LOOIC is defined as
    	\[
    	\text{LOOIC} = \sum_{i=1}^n \log \left( \frac{ \sum_{s=1}^N w_i^s p(y_i \mid \etab_s)}{\sum_{s=1}^N w_i^s} \right)
    	\]
    	where for each $s = 1, \cdots, N$, $\bm{w}^s =(w_1^s, \cdots, w_n^s)$ is a vector of importance sampling weights, typically calculated via the Pareto smoothed importance sampling method \citep[PSIS;][]{vehtari:2015}. Because of the importance sampling weights,  LOOIC can be more stable in practice than WAIC. See \cite{vehtari:2017} for a gentle and thorough introduction to both WAIC and LOOIC, including applications and case studies. 
    	
    	Because both WAIC and LOOIC are based on the mixture likelihood and do not explicitly depend on the sampled model parameters (thus, remain unaffected by the presence of multiple permutation and non-permutation modes), in \pkg{BAMBI} we recommend using either of these two criteria for selecting the number of mixture components.  Both WAIC and LOOIC are made available in \bambi via their implementations in the \proglang{R} package \pkg{loo} \citep{loo_rpack}, which also provides a \code{compare()} function for comparing WAICs/LOOICs based on estimated difference in expected log predictive density (elpd). In BAMBI, during an incremental model fitting via \code{fit_incremental_angmix} with \code{crit = 'WAIC'} or \code{crit = 'LOOIC'}, a test of hypothesis $H_{0K}:$ elpd for the fitted model with  $K$ components $\geq$ elpd for the fitted model with $K + 1$ components, is performed at every $K \geq 1$. The test-statistic used for the test is an approximate z-score based on the normalized estimated elpd difference between the two models (obtained from \code{compare()}, which provides estimated elpd difference along with its standard error). The incremental fitting stops if  $H_{0K}$ cannot be rejected (at level \code{alpha}, defaults to 0.05) for some $K \geq 1$; this $K$ is then regarded as the optimum number of components. 

	\item \textit{Marginal Likelihood (ML)}. Marginal likelihood is arguably the most natural and intuitive model selection criterion that is used in the Bayesian paradigm. As the name suggests, marginal likelihood is the likelihood obtained by integrating out the parameters from the joint density of the data and parameters, and provides a natural way of measuring the ``marginal'' effect of data. In the context of Bayesian model selection, marginal likelihood provides a way of selecting an optimum model in that the model with largest marginal likelihood provides the best fit.  Given the likelihood $L(\yb \mid \etab)$ and prior density $\pi(\etab)$, marginal likelihood is the constant (dependent only on the data):
	\[
	m(\yb) = \int_{\mathcal{E}} L(\yb \mid \etab) \pi(\etab) \: d\etab
	\]
	where $\mathcal{E}$ denotes the support of the parameter vector $\etab$. Note that $m(\yb)$ is the reciprocal of the normalizing constant required to define the posterior density. Evaluation of the marginal likelihood $m(\yb)$ in practice however is typically challenging, as it tends to be a high-dimensional intractable integral (as in our case). Multiple estimation techniques based on samples from the posterior density $\pi(\etab \mid \yb)$ have been proposed in the literature; in \bambi we implement bridge sampling \citep{meng:wong:1996, meng:schilling:2002}.
	Briefly, the key idea is to first write $m(\yb)$ as
	\[
	m(\yb) = \frac{\int_\mathcal{E} h(\etab) L(\yb \mid \etab) \pi(\etab) g(\etab)\:d\etab}{\int_\mathcal{E} h(\etab) g(\etab)  \pi(\etab \mid \yb) \:d\etab}
	= \frac{E_{g} \left[h(\etab) L(\yb \mid \etab) \pi(\etab)\right] }{E_{\pi(\cdot \mid \yb)} \left[h(\etab) g(\etab) \right] }
	\]
	where $g$ is a density, called the \emph{proposal density}, and $h$ is a function, called the \emph{bridge function}. Then one approximates the above ratio by
	\[
	\hat m(\yb) = \frac{\frac{1}{n_2} \sum_{j=1}^{n_2} h(\etab_j^*) L(\yb \mid \etab_j^*) \pi(\etab_j^*)}{ \frac{1}{n_1} \sum_{j=1}^{n_1}h(\etab_j^\dagger) g(\etab_j^\dagger)}
	\]
	where $\etab_1^\dagger, \cdots, \etab_{n_1}^\dagger$ are MCMC samples from the posterior density $\pi(\cdot \mid \yb)$ and $\etab_1^*, \cdots, \etab^*_{n_2}$ are samples from the proposal density $g$. Note that $h$ and $g$ play crucial roles in the estimation of $m(\yb)$, and must be optimally chosen for accurate results. See \cite{gronau:2017} for a gentle and detailed tutorial on bridge sampling. In \pkg{BAMBI}, marginal likelihood can be used to select the optimal number of mixture component in a \code{fit_incremental_angmix} run, by specifying \code{crit = 'LOGML'}. This will ensure computation of the log marginal likelihood via bridge sampling for every mixture model during the incremental run, and the model attaining the first minimum negative log marginal likelihood will be treated as the optimum model. 
	
	It should however be noted  that for mixture models, optimal selection of $h$ and $g$ is typically difficult  due to the multi-modality of the posterior density; see \citet[Chapter 5]{fruhwirth:2006}, for a review of some of the available methods. In  \pkg{BAMBI}, marginal likelihood is computed by leveraging the function \code{bridge_sampler} from the \proglang{R} package \pkg{bridgesampling} \citep{bridgesampling_rpack}, and the authors of \pkg{bridgesampling} warn against the use of  \code{bridge_sampler} in mixture models \citep[see the Discussion section in][]{bridgesampling_rpack}. As such, use of this method in \bambi is not recommended, even though the functionality is available. 

	\item \textit{Akaike Information Criterion (AIC)} \citep{akaike:1974}.  Let $\hat L$ be the maximum value of the likelihood function for the model and  let $m$ be the number of estimated parameters in the model. Then AIC is defined as
	\[
	\text{AIC} = -2 \log \hat{L} + 2m.
	\]
	
	\item \textit{Bayesian Information Criterion (BIC)} \citep{schwarz:1978}. Under the same setup, if $n$ denotes the number of data points, BIC is defined by
	\[
	\text{BIC} = - 2 \log \hat{L} + m \log(n).
	\]
	Observe that both AIC and BIC depend on the maximum value $\hat L$, which, in general, is not directly available in MCMC simulations. A possibly suboptimal estimate of the global maximum is given by the maximum value of the likelihood function computed at the MCMC samples. 
	
	During model selection, the model with minimum AIC (or BIC) can be treated as the optimal model. In \bambi AIC/BIC can be used for selecting optimum number of components in an \code{fit_incremental_angmix} run by specifying \code{crit = 'AIC'} or \code{crit = 'BIC'}. This will ensure computation of AIC/BIC of every mixture model fitted during the incremental fitting; the model attaining the first minimum AIC/BIC will be treated as the optimum model.
	
	It should to be noted, however, that AIC and BIC are both based on asymptotic normality results that \emph{do not} hold for  mixture models with multiple modes, and hence their use in selecting the number of mixture components may lead to inconsistent results. Thus, though implemented, using AIC or BIC is not recommended in \pkg{BAMBI}.  
	
	\item \textit{Deviance Information Criterion (DIC)} \citep{spigelhalter:2002}. DIC is another model selection criterion, which, similar to AIC and BIC, is based on an asymptotic result for large samples. Let   $D(\etab ) = -2\log p(\dat \mid \etab )$  denote the deviance, where $\etab$ denotes the vector of all parameters in the model and $p(\dat \mid \etab)$ denotes the likelihood. Let $\{\etab_1, \cdots, \etab_N \}$ denote the MCMC realizations of the parameters. Define (estimated) effective number of parameters $p_D$ by  $p_D = \bar{D}(\etab) - D({\bar \etab})$, where $\bar{D}(\etab) = N^{-1} \sum_{s=1}^N D(\etab_s)$  and $\bar{\etab} = N^{-1} \sum_{s=1}^N \etab_s$. Another commonly used form for $p_D$ is given by $p_D = \widehat{\text{var}}\: D(\etab) / 2$, \citep{gelman:2014} where $\widehat {\text{var}}\: D(\etab)$ denotes the estimated variance of $D(\etab)$ based on $\etab_1, \cdots, \etab_N$. Then DIC is defined as
	\[
	\text{DIC} = p_D + \bar D(\etab) = D (\bar \etab) + 2 p_D.
	\]
	In \bambi AIC/BIC can be used for selecting optimum number of components in an \code{fit_incremental_angmix} run by specifying \code{crit = 'DIC'}. This will ensure computation of DIC of every mixture model fitted during the incremental fitting; the model attaining the first minimum DIC will be treated as the optimum model.
	
	It should be noted that use of DIC can be unstable in practice. For example,  if the first form of $p_D$, i.e., $p_D = \bar{D}(\etab) - D({\bar \etab})$) is used, DIC becomes heavily dependent on the plug-in estimator $\bar \etab$. However, in Bayesian mixture modeling the posterior mean is not always a suitable plug-in estimator for the parameter vector as it may lie between different modes of the posterior density \citep{plummer:2008}. The problem is exacerbated by the presence of label switching in the MCMC samples (see Section~\ref{labswitch}). Moreover, depending on how the information on latent component indicators are handled, multiple versions of DIC can be constructed here. \citet{celeux:2006} consider no less than eight variants, but are unable to recommend any of them for practical use. Likewise we caution against the use of DIC, although the functionality is available in \pkg{BAMBI}. 

\end{enumerate}

\subsection{Label switching} \label{labswitch}
Label switching is a fundamental aspect of Bayesian mixture modeling that requires proper care. Briefly, when exchangeable priors $\pi(\thetab_j)$'s are placed on the model parameters $\thetab_j$'s, the resulting posterior distribution becomes invariant to permutation in the component label $j$'s. As a result, the posterior density  consists of symmetric or permutation modes that are identical up to permutation of component labels. A well mixing  MCMC algorithm will explore these permutation modes,  causing the component labels to switch over the course of an MCMC simulation. This phenomenon is called label switching, and is required in MCMC-based Bayesian mixture modeling to justify convergence. A fundamental limitation of 
MCMC-based Bayesian mixture modeling is that the chains may become trapped at local modes, rather than fully exploring the symmetric modes. One possible remedy is to run multiple independent chains to improve exploration of the posterior.  Alternatively, one may embed a deliberate random relabeling step into the sampler, i.e., adopt a so-called \emph{permutation sampling} \citep{fruhwirth:2001} scheme: after each draw of the  random allocation, components are relabeled according to a random permutation of $1, \cdots, K$. This is, in fact, a specific example of a \emph{sandwich algorithm} \citep{Meng:1999, Yu:2011, hobert:2011}, where a computationally inexpensive step (of drawing a random permutation of $\{1,\dots, K\}$, and then relabeling the components according to that random permutation) is \emph{sandwiched} in between the two steps (drawing the allocation vector, and drawing the component-specific parameters) of a Data Augmentation algorithm. Sandwich algorithms often converge faster  than the original Data Augmentation algorithm; in the case of permutation sampling, the chain is forced to visit the permutation modes (and potentially the non-permutation modes) more frequently. These potential improvements in convergence  would not be achieved by simply randomly switching labels of the MCMC samples post-hoc. In permutation sampling, inclusion of the random label switching step results in a modified MCMC algorithm that  is theoretically proven to be at least as good as the original Data Augmentation algorithm in terms of convergence rates \citep{khare:hobert:2011}. 
However, care must be taken if RWMH or HMC updates for the component specific parameters are adaptively tuned according to the scales and variabilities of the sampled model parameters (to do so properly requires keeping track of each component label). In \bambi permutation sampling can be done after burn-in, by setting \code{perm_sampling = TRUE} (defaults to \code{FALSE}) in a \code{fit_angmix} call. 


Although label switching is required for MCMC convergence in Bayesian mixture modeling, its presence in MCMC samples makes inference on the different components via posterior means or quantiles challenging (note that MAP estimation is not affected). A number of techniques have been proposed to handle this problem; see, e.g.,~\cite{jasra:2005} and \cite{rodriguez:2014}. The available methods either need to be applied during MCMC sampling (\textit{on-line}) or after simulating the entire chain (\textit{post processing}).

Several post-processing techniques that undo label switching are implemented in the useful \proglang{R} package \pkg{label.switching} \citep{label.switching:2016}, and in \pkg{BAMBI}, we provide a wrapper called \code{fix_label} for the main \code{label.switching} function from that package. All the methods available in \code{label.switching} are appropriately implemented in \code{fix_label}, which takes an \code{angmcmc} object (see Section~\ref{sec_angmcmc}) as input, and may require additional user inputs depending on the method used. The Kullback-Leibler divergence based method by \citet{stephens:2000} (\code{method = 'STEPHENS'}) is used by default if the permutation sampling is performed during original MCMC run; otherwise, the default method is the data-based algorithm of \citet{rodriguez:2014} (\code{method = 'DATA-BASED'}); neither requires any additional input other than an \code{angmcmc} object.

\subsection{Initialization of the parameters, and the use of multiple chains for faster convergence}
\label{initial_vals}

MCMC algorithms can converge faster if the initial values are chosen well.
The function \code{fit_angmix}, when called without supplying starting parameter values, will automatically initialize the latent allocation to the mixture components. The default (and recommended) option is initialization via a $k$-means algorithm:  toroidal angle pairs are first projected onto the surface of a unit sphere, and then Cartesian coordinates of the projected spherical points are clustered.  Random initial allocation is also provided as an option, but is not recommended as it may lead to slow convergence.  Once an initial allocation is obtained, component specific parameters are  estimated via method of moments (see \cite{jammalamadaka:2001, singh:2002, mardia:2007} for more details on these estimators), and the mixing proportions are estimated by the sample proportions.  

When explicit starting values of the model parameters and mixing proportions are provided to \code{fit_angmix},  no initial allocation is necessary. This is particularly useful for estimating the number of components, when mixture models are being fitted incrementally (e.g., via \code{fit_incremental_angmix}).  Under incremental model fitting the parameters of a $K+1$ component mixture can be  initialized directly from the parameter estimates from a $K$ component mixture; the  extra component is simply taken as a ``copy'' of an existing component (preferably the one with the largest mixing proportion), and the associated mixing proportion is distributed equally between the two identical components. This method is expected to work well when the posterior density handles overfitting in the ``let two component-specific parameters be identical'' way  (see Sections~\ref{prior} and \ref{compsizedet}), which is the approach taken in \pkg{BAMBI}. As such, this is the default method of initializing parameters in \code{fit_incremental_angmix}  when $K > 2$; however $k$-means allocation followed by moment estimation can also be used, by setting \code{prev_par = FALSE}. 

Finally, we note that even with good initial values, MCMC samplers can still get trapped in local modes for a large number of iterations,  rather than fully exploring the posterior density. 
One possible remedy is to run multiple independent chains to improve exploration of the posterior, which is implemented in \pkg{BAMBI}. The argument \code{n.chains} (set to 3 by default) specifies the number of independent chains to run in \code{fit_angmix}. These chains can be run in parallel, see Section~\ref{sec_bambi_fn}~Category~IV for more details.

\section{BAMBI Package} \label{bambipack}
This section overviews the functionalities of \pkg{BAMBI}. At the core of the package is the \code{angmcmc} object, which is created when a model fitting function is used. In the following we first describe the \code{angmcmc} objects, then describe the datasets included in \pkg{BAMBI}, and finally discuss the functions available in \pkg{BAMBI} and comment on their usability. However, this is not an exhaustive manual; all functions in \pkg{BAMBI} include \proglang{R} documentations, which should serve as the definitive resources.

\subsection{"angmcmc" objects} \label{sec_angmcmc}

\code{angmcmc} objects are classified lists belonging to the S3 class \code{angmcmc} that are created when the function \code{fit_angmix} is used. An \code{angmcmc} object contains a number of elements, including the dataset and its dimension (i.e., univariate or bivariate), the model being fitted, the tuning parameters used,  MCMC samples of the parameter vector, and at each iteration the (hidden) component indicators for data points,  log-likelihood and log posterior density values (up to additive constants). When printed, an \code{angmcmc} object returns a brief summary of the function arguments used and the acceptance rate of the proposal states (in HMC and RWMH). An \code{angmcmc} object can be used as an argument for the diagnostic and post-processing functions available in \pkg{BAMBI} for making further inferences.

\subsection{Datasets}

\bambi contains two illustrative datasets, namely \code{wind} (univariate) and \code{tim8} (bivariate), each measured in the radian scale $[0, 2\pi)$.

\begin{enumerate}
	\item[wind] The \code{wind} dataset consists of 239 observations on wind direction (originally measured in 10s of degrees, and then converted into radians) measured at Saturna Island, British Columbia, Canada during October 1-10, 2016 (obtained from Environment Canada website). There was a severe storm during October 4-7 in Saturna Island, which caused significant fluctuations in wind direction.
	
    \item[tim8] The dataset \code{tim8}  consists of 490 pairs of backbone dihedral angles $(\phi, \psi)$ for \textit{Triose Phosphate Isomerase} (8TIM). The three dimensional structure of  8TIM is available from the protein data bank (PDB).  The  protein is an example of a TIM barrel, a common type of protein fold exhibiting alternating
    $\alpha$-helices and $\beta$-sheets. The backbone angles for this protein were obtained by using the DSSP software (\cite{touw:2015, kabsch:1983}) on the PDB file for 8TIM, and then converted into radians. 			
\end{enumerate}

\subsection{Functions} \label{sec_bambi_fn}
In \pkg{BAMBI}, all five models described in Section~\ref{intro}, namely the univariate von Mises (\code{vm}), univariate wrapped normal (\code{wnorm}), bivariate von Mises sine (\code{vmsin}), bivariate von Mises cosine (\code{vmcos}) and bivariate wrapped normal (\code{wnorm2}), and their (within same model) mixtures  are implemented. The functions in \bambi can be classified into six major categories:

\paragraph{Functions for evaluating density and generating random samples from an angular distribution.} \label{single_fn} 
  The functions \code{dvm}, \code{dwnorm}, \code{dvmsin}, \code{dvmcos} and \code{dwnorm2} evaluate the density and the functions \code{rvm}, \code{rwnorm}, \code{rvmsin}, \code{rvmcos} and \code{rwnorm2}  generate random samples from the models \code{vm}, \code{wnorm}, \code{vmsin}, \code{vmcos} and \code{wnorm2} respectively. The parameters of the models are specified as arguments; otherwise, default values (zero means, unit concentrations, and zero association) are used.
  
Density evaluations require computation of the normalizing constants, which for the \code{vmcos} model requires proper care, especially when $\kappa_1$, $\kappa_2$ or $|\kappa_3|$ is large. This is because the analytic expression involves infinite (alternating if $\kappa_3 < 0$) series of product of modified Bessel functions, which become numerically unstable when these parameters are large. As such, when $\kappa_3 < -5$ or $\max(\kappa_1, \kappa_2, |\kappa_3|) > 50$, the reciprocal of the integral for the normalizing constant is evaluated numerically using a (quasi) Monte Carlo method. By default, \code{n\_qrnd =} $10^4$ pairs of Sobol numbers are used for this purpose; however, \code{n\_qrnd}, or a two-column matrix \code{qrnd} containing random/quasi-random numbers between 0 and 1 can be supplied for this approximation\footnote{The user may perform a Monte Carlo approximation for the normalizing constant even when numerical evaluation of the analytic formula is stable, by changing \code{force\_approx\_const} to \code{TRUE} from its default value \code{FALSE}.}. For \code{vmsin} model, evaluation of the constant via its analytic form is much more stable, as the associated infinite series consists only of non-negative terms.  For univariate and bivariate wrapped normal models, the default absolute integer displacement for approximating the Wrapped Normal sum is 3, which can be changed to any value in $\{1, 2, 3, 4, 5\}$, through the argument \code{int.displ}. Note that \code{int.displ} regulates how many terms would be used to approximate the infinite sum present in the univariate and bivariate wrapped normal densities in \eqref{wn} and \eqref{bvwn}. For example, \code{int.displ} = $M$ implies that the infinite sum in the univariate wrapped normal density will be approximated by a finite sum of $2M +1$ values, with the summation index $\omega$ ranging over $\{0, \pm 1, \dots, \pm M \}$. For a bivariate wrapped normal density, setting \code{int.displ} = $M$ will ensure that the infinite double sum is approximated by a finite double sum, with the paired summation index $(\omega_1, \omega_2)$ ranging over $\{0, \pm 1, \dots, \pm M \}^2$.

Random data generation from the von Mises models (both univariate and bivariate) is done via   rejection sampling. In the univariate case, the von Mises random deviates are efficiently generated using a rejection sampling scheme from a wrapped Cauchy distribution \citep{best:fisher:1979, mardia:jupp:2009}. For the bivariate models, two forms of random samplings are implemented. In the first method, random deviates are generated via a naive bivariate rejection sampler with uniform proposal density (the majorization constant is numerically evaluated). In the second method (proposed in the web appendix of \cite{mardia:2007}), random deviates are first generated from the marginal distribution of one coordinate, then the other coordinate is drawn from the corresponding conditional distribution (which is von Mises in both models). The authors note that this latter scheme has a typical efficiency rate of at least 60\%. It is to be noted that while this scheme is usually more efficient than the naive rejection sampler (especially when the concentration is high), it does have an often substantial overhead due to the numerical computations required for determining appropriate proposal density parameters. These overheads often outweigh efficiency gains, especially if the sample size and/or concentration parameters are small. In \code{BAMBI}, therefore, the naive rejection sampler is used by default when the sample size is moderate or small (< 100), or when the concentration parameters are small (< 0.1).\footnote{When the concentration parameters are large, the density becomes concentrated in (a) very narrow region(s). As such, efficiency of the naive sampler, which draws proposal random deviates from a uniform density over the entire support, can be 15-20\% or less. However, even then the overall runtimes of the naive method are often still comparable to \citet{mardia:2007}'s method when the sample size is moderate or small.} 
For wrapped normal distributions (both univariate and bivariate) a random deviate is easily obtained by sampling from the unwrapped normal distribution (using \code{rnorm} if univariate, and \code{rvmnorm} from package \pkg{mvtnorm} if bivariate), and then wrapping into $[0, 2\pi)$.

\paragraph{Functions for evaluating density and generating random samples from a finite mixture model with a fixed number of components.} \label{mix_fn} Analogous to the functions for single component densities, the functions \code{dvmmix}, \code{dwnormmix}, \code{dvmsinmix}, \code{dvmcosmix} and \code{dwnorm2mix} evaluate the density and the functions \code{rvmmix}, \code{rwnormmix}, \code{rvmsinmix}, \code{rvmcosmix} and \code{rwnorm2mix}  generate random samples from mixtures of \code{vm}, \code{wnorm}, \code{vmsin}, \code{vmcos} and \code{wnorm2} respectively. All model parameters and mixing proportions must be provided as input arguments.

\paragraph{Functions for visualizing and  summarizing bivariate  models.} \label{summ_model_fn} To visualize the density for any of the three bivariate angular mixture models (with specified parameters and number of components) considered in this paper, the functions \code{surface_model} and \code{contour_model} can be used, which respectively plot the surface and the contour of a mixture density. To compute summary statistics for a single bivariate  angular distribution, the function \code{circ_varcor_model} can be used, which calculates the circular variance and  correlation coefficients (both Jammalamada-Sarma and Fisher-Lee forms, see Section~\ref{summ_meas}).  However, summarizing angular mixture models via circular variances and correlations is not recommended, as interpretations of the results can be challenging when multiple clusters are present in the data.\footnote{To calculate the circular variances and correlations for a mixture density, one can simulate from the density first, and then approximate the population quantities by their sample analogs on the basis of the simulated data.} The function \code{circ_cor} implements the sample Jammalamada-Sarma and Fisher-Lee circular correlation coefficients, as well as two forms of Kendall's tau \citep{fisher:1983, zhan:2017} as non-parametric measures. The sample circular variance can be computed using the \code{var.circular} function from \proglang{R} package \pkg{circular} \citep{circular_Rpack}.

\paragraph{Functions for fitting a single component model or a finite mixture model with specified number of components to a given dataset using MCMC.} \label{fit_fn} Given a dataset, using methods discussed in Section~\ref{meth}, the function \code{fit_angmix}  generates MCMC samples for parameters in an angular  mixture model with a specified number of components. Available models for bivariate input data (which must be supplied as a two-column matrix or data frame) are \code{vmsin}, \code{vmcos} and \code{wnorm2}, and for univariate data are \code{vm} and \code{wnorm}.  The argument \code{ncomp} specifies number of components  in the mixture model, with \code{ncomp = 1} representing the single component case (i.e.,~fitting a single density). A Gibbs sampler is used to generate latent component indicators, and conditional on this allocation the model parameters are sampled either by HMC (default), or by RWMH (can be specified through the argument \code{method}). A permutation sampling step can be added after burn-in by setting the logical argument \code{perm_sampling} to \code{TRUE}. The tuning parameters \code{epsilon} and \code{L} in HMC, and \code{propscale} in RWMH have pre-specified default values, and there is an auto-tuning feature for \code{epsilon} and \code{propscale} which is used by default, but can be turned off by setting the logical argument \code{autotune = FALSE}. The burn-in proportion can be specified through the argument \code{burnin.prop}, which is set to 0.5 by default. For HMC, the option to use random \code{epsilon} and \code{L} at each iteration is specified via the logical arguments \code{epsilon.random} and \code{L.random} respectively. In practice, using multiple chains is recommended, and the argument \code{n.chains} specifies the number of chains to be used. These chains can be run in parallel, if the logical argument \code{chains_parallel} is set to \code{TRUE}. The parallelization is implemented using \code{future_lapply} from \proglang{R} package \pkg{future.apply}; an appropriate \code{future::plan()} must be set in advance to ensure that the chains run in parallel   (otherwise the chains will run sequentially), see Section~\ref{illus} for an example. To retain reproducibility while running multiple chains in parallel, the \emph{same} RNG state is passed at the beginning of each chain. This is done by specifying \code{future.seed = TRUE} in \code{future.apply::future_lapply} call. Then at the beginning of the $i$-th chain, before drawing any parameters, $i$-many Uniform(0, 1) random numbers are generated using \code{runif(i)} (and then thrown away). This ensures that the RNG states across chains prior to sampling of the parameters are different (but reproducible), and hence, no two chains can become identical, even if they have the same starting and tuning parameters. This however, creates a difference between a \code{fit_angmix} call with multiple chains which is run sequentially by setting \code{chains_parallel = FALSE}, and another call which is run sequentially because of a sequential \code{plan()} (or no \code{plan()}), with \code{chains_parallel = TRUE}. In the former, \code{base::lapply} instead of \code{future_lapply} is used, which means that \emph{different} RNG states are passed at the initiation of each chain.

There are options for choosing prior hyperparameters. The prior for the association parameter $\kappa_3$ in bivariate models, and the log of the concentration parameters $\kappa$ (in univariate models), $\kappa_1$ and $\kappa_2$ (in bivariate models) are taken to be the normal distribution (i.e., the priors for $\kappa, \kappa_1, \kappa_2$ are log normal), all with zero mean. The default variance for the normal prior is 1000, which provide diffuse priors, although they can be set by the user via the argument \code{norm.var}.  A fixed non-informative Uniform$(0, 2\pi)$ prior is used for the mean parameters. The Dirichlet prior parameters $\alpha_j$'s for the mixing proportions $p_j$'s can be supplied through the argument \code{pmix.alpha}, which can either be a positive real number (same for all $\alpha_j$), or a vector of the same length as \code{pmix}. It is recommended that $\alpha_j$'s be chosen large for proper handling of overfitted mixtures;  following \citet[Section~1.3.2]{fruhwirth:2011}, all $\alpha_j$'s default to $(r+r(r+ 1)/2)/2 + 3$, where $r$ denotes dimension of the data (i.e., $r=1$ for univariate data, and $r=2$ for bivariate data). See Sections~\ref{prior} and \ref{compsizedet} for more details.

The argument \code{cov.restrict} specifies any (additional) restriction to be imposed on the component specific association parameters while fitting the model. The available choices are \code{"POSITIVE", "NEGATIVE", "ZERO"} and \code{"NONE"}. Note that when \code{cov.restrict = "ZERO"}, \code{fit_angmix} fits a mixture with product components. By default, \code{cov.restrict = "NONE"}, which does not impose any (additional) restriction.  When \code{model} is \code{"vmsin"} or \code{"vmcos"}, the component densities can be bimodal. However, one can restrict these densities to be unimodal, by setting the logical argument \code{unimodal.component} to \code{TRUE} (defaults to \code{FALSE}). For \code{"wnorm"} and \code{"wnorm2"} models, the default absolute integer displacement for approximating the Wrapped Normal sum is 3, which can be changed to any value in $\{1, 2, 3, 4, 5\}$, through the argument \code{int.displ}. For \code{"vmcos"} model, the normalizing constant is numerically approximated using quasi Monte Carlo method when analytic evaluation suffers from numerical instability. The arguments \code{qrnd} and \code{n_qrnd} can be used to alter the default settings used for these approximations. See the documentation of \code{fit_angmix} for more details.

The function \code{fit_angmix} creates an \code{angmcmc} object, which can be used for assessing the fit, post processing, and estimating parameters. 

\paragraph{Functions for assessing the fit.} Goodness of fit for MCMC-based Bayesian modeling depends on both convergence of the Markov chain and the appropriateness of the model used.  \bambi contains a number of functions that can be used to examine these two aspects. The functions \code{paramtrace} and \code{lpdtrace} respectively plot the parameter and log posterior density traces for visual assessment of convergence. These two functions are called together in the \code{plot} method for \code{angmcmc} objects. The \code{as.mcmc.list} method for \code{angmcmc} objects provides a convenient way of converting an \code{angmcmc} object to an \code{mcmc.list} object from package \pkg{coda}, which provides several additional functions for convergence diagnostics. 
 
Once convergence is justified, the appropriateness of the fitted model can be visually assessed by the S3 functions \code{densityplot} (from \pkg{lattice})   and \code{contour}. The first function plots the density surface (for bivariate data) or density curve (for univariate data) of the \code{fitted} mixture model, and the second plots the associated contour of a bivariate model. Note that these plots provide visual assessment of the goodness of fit by assuming the Markov chains have converged and the parameters can be estimated on the basis of the MCMC samples. As such, convergence of the MCMC samples must be ensured prior to this step. Otherwise these visual diagnostics will lead to misleading conclusions. 

The comparative goodness of fit for two mixture models can be assessed on the basis of model selection criteria implemented in \pkg{BAMBI}, namely, marginal likelihood, AIC, BIC, DIC, WAIC and LOOIC via the functions \code{bridge_sampler.angmcmc}, \code{AIC}, \code{BIC}, \code{DIC}, \code{waic.angmcmc} and \code{loo.angmcmc}. 
As with the diagnostic plots, one general caveat for using any of these model selection criteria is that one should ensure convergence of the associated Markov chain first; otherwise, the results may be misleading.

 \paragraph{Functions for post-processing and estimating parameters.}
  \pkg{BAMBI} provides several post-processing functions to aid inference on the basis of the generated MCMC samples. The function \code{add_burnin_thin} adds additional burn-in and/or thinning to an \code{angmcmc} object and the function \code{select_chains} extracts a subset of chains. These two functions can be helpful if convergence diagnostics indicate that some of the chains are poor mixing and/or require additional burn-in and thinning for convergence.
  
  As described in Section~\ref{labswitch}, care should be taken to ensure that there is no label switching in the MCMC output if inference is being made on the basis of posterior mean/median. If present, label switching can be fixed by applying the wrapper function \code{fix_label} on the \code{angmcmc} object, which will output another \code{angmcmc} object with label switching fixed.  
  
   Point estimates of the parameters are obtained using the \bambi function \code{pointest} on a fitted \code{angmcmc} object.  The function \code{pointest} calculates point estimates by applying \code{fn} on the MCMC samples, where \code{fn} is either a function, or a character string specifying the name of the function. Default for \code{fn} is \code{mean}, which computes posterior mean. If \code{fn} is \code{"MODE"} or \code{"MAP"} then the (MCMC-based approximate) MAP estimate is returned. Posterior quantiles can be estimated by (sample) quantiles of the MCMC realizations using the S3 function \code{quantile.angmcmc}. These quantiles can be used to construct credible sets. For example, if $\xi_\zeta$ denotes the $\zeta$-th (sample) quantile of the MCMC observations for $0 < \zeta < 1$, the central 95\% credible interval is given by $(\xi_{0.025}, \xi_{0.975})$. Both of these functions can be applied on specific parameters and/or component labels by setting the arguments \code{par.name} and \code{comp.label} accordingly. The S3 function \code{summary.angmcmc} prints (estimated) posterior means and the central 95\% credible intervals for all the parameters.
   
   The estimated latent allocation from an \code{angmcmc} object can be obtained using the function \code{latent_allocation}, which first estimates parameters via \code{pointest}, computes posterior membership probabilities (see (\ref{mem_probs})) for each data point, and then assigns each data point the class with largest membership probability.  The (estimated) log-likelihood of an \code{angmcmc} object can be extracted as a \code{logLik} object using the S3 function \code{logLik.angmcmc}. Note that there are two methods of obtaining the log-likelihood from an \code{angmcmc} object. In the default method (\code{method = 1}), the final  log-likelihood is computed by applying a function \code{fn} (defaults to \code{max}) on the iteration wise log-likelihood values obtained during the original MCMC sampling. On the other hand, if \code{method = 2}, first the parameters are estimated (using \code{pointest}), and then the log-likelihood is computed at the estimated parameters. 

   Density evaluations and random data generation from a fitted model can be done using the functions \code{d_fitted} and \code{r_fitted} respectively. Both functions take an \code{angmcmc} object as input and apply the appropriate model specific density evaluation and random data generation functions with the estimate $\hat \etab$  of the parameter vector $\etab$ obtained via \code{pointest}. The actual MCMC samples for one or more parameters in one or more components from one or more chains can be extracted via the function \code{extractsamples} on an \code{angmcmc} object for further analysis.

   \paragraph{Functions for incremental mixture model fitting and number of components estimation.}  Using the methods and model selection criteria described in Sections~\ref{prior} and \ref{compsizedet}, the  function \code{fit_incremental_angmix} fits mixture models with incremental number of components by calling \code{fit_angmix} at each step, and uses a Bayesian model selection criterion to determine an optimal number of components.  The arguments \code{start_ncomp} and \code{max_ncomp} provide the starting and maximum number of components to be used in the incremental fitting, which are respectively set to 1 and 10 by default. The available  model selection criterion to use (specified via the argument \code{crit}) are \code{'LOGML'}, \code{'AIC'}, \code{'BIC'}, \code{'WAIC'}, \code{'LOOIC'} and \code{'DIC'}, which is computed for every intermediate fit. The initial values of the starting model (or a model with $\leq 2$ components) are obtained by default using moment estimation on $k$-means clusters (they can also be directly supplied by the user). For the subsequent models (with number of components $\geq 3$), the initial values are by default obtained from the MCMC-based MAP parameter estimates for the previous model with one fewer component (see Section~\ref{compsizedet}).  This can be overridden by setting \code{prev_par = FALSE}, to use $k$-means clustering followed by moment estimation instead.   By default, only the ``best'' chain, i.e., the one  with maximum average log posterior density, is used for computation of model selection criterion, and parameter estimation (if \code{prev_par} is set to \code{TRUE}). This default helps safeguard against situations where some of the chains may get trapped at local optima. However, samples from all chains can be used for these computations by setting \code{use_best_chain = FALSE}. The function stops when \code{crit} achieves its first minimum, or when \code{max_ncomp} is reached, and returns a list with the following elements:
   \begin{itemize}
     \item \code{fit.best} is an \code{angmcmc} object corresponding to the optimum or \textit{best} fit.

     \item \code{crit.all} provides a vector (list) of model selection criterion values for each incremental model fitted.

     \item \code{crit.best} is the value of the model selection criterion for the model with optimal number of components.
     
     \item \code{maxllik.all} is the maximum (obtained from MCMC iterations)  log-likelihood for all fitted models.
       
     \item \code{maxllik.best} is maximum log-likelihood for the optimal model.

     \item \code{ncomp.best} is the optimal number of components associated with the ``best'' model.
     
     \item \code{fit.all}  is a list consisting of \code{angmcmc} objects for all number of components fitted during the model selection process. Any element of this list can be used as an argument for any function that takes \code{angmcmc} objects as input. However, this can be very memory intensive, and as such, by default not returned (can be returned by setting \code{return_all = TRUE}).

   \end{itemize}

The \code{angmcmc} object  corresponding to the best fit and the associated value of the model selection criterion  can also be extracted from the output of \code{fit_incremental_angmix}  extracted using the convenience functions \code{bestmodel} and \code{bestcriterion} respectively.


\section{Illustrations} \label{illus}

In this section we illustrate functionalities of \bambi by fitting mixture models to the angular datasets included in the package.  The following command
\begin{Sinput}
R> library(BAMBI)
\end{Sinput}
loads the package after it has been installed.  For reproducibility of the results presented, the same random seed 12321 is used for all of our examples.

\subsection{Fitting mixture models on the tim8 bivariate data}  \label{TIM8FIT}
The \code{tim8} dataset consists of 490 backbone torsion angle pairs $(\phi, \psi)$ for the protein 8TIM. The protein is an example of a TIM barrel, a common type of protein fold exhibiting alternating $\alpha$-helices and $\beta$-sheets. Its Ramachandran plot (i.e., a scatterplot of $(\phi, \psi)$ pairs) is generated using the following \proglang{R} command and shown in Figure~\ref{rama}.
\begin{Sinput}
R> plot(tim8, pch = 16,
+       xlim = c(0, 2*pi), ylim = c(0, 2*pi),
+       main = "8TIM", 
+       col = scales::alpha("black", 0.6))
\end{Sinput}

Note that this scatterplot projects the torus onto a 2D surface, which cannot show the wraparound nature of the angles.  This projection is not unique and affects the appearance of the scatterplot depending on how the angles are represented, e.g., in $[-\pi, \pi)$ instead of $[0, 2\pi)$. Moreover, one should be careful to note that the  top and bottom boundaries in these plots join together, as do the left and right boundaries. In Figure~\ref{rama}, about 3-6 visually distinct clusters can be seen; however, we need to note that the points around (4.5, 0) and (5.5, 6) in fact may form a single cluster.  Such features cannot be correctly modeled with statistical methods that ignore  circularity. Thus, this is a suitable example for illustrating the need for mixtures of (bivariate) angular distributions.


\begin{figure}[hptb]
	\centering
	\includegraphics[width=4in]{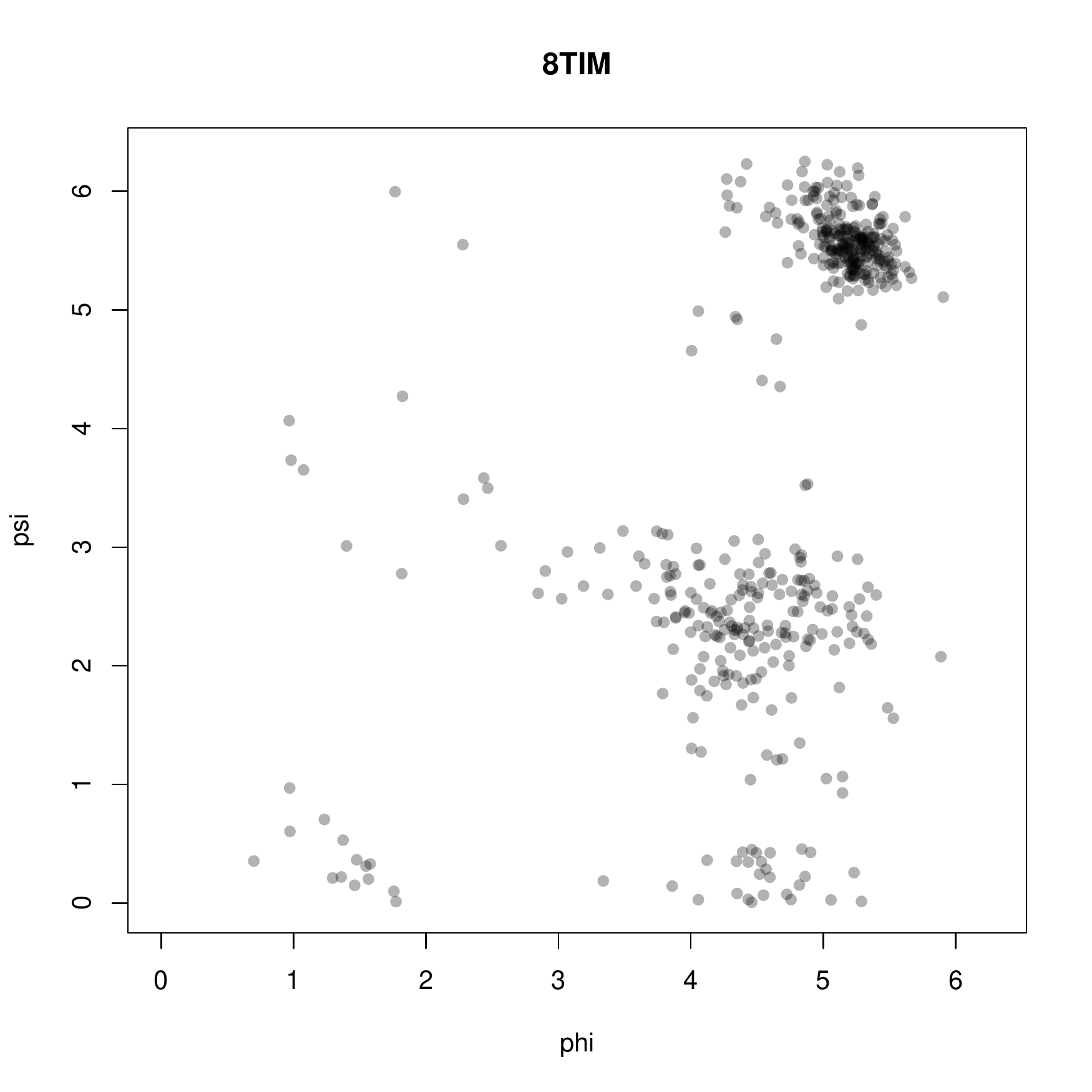}
	\caption{Ramachandran Plot for 8TIM.} \label{rama}
\end{figure}

To fit a bivariate mixture model with a specified number of components to the dataset, we can use the \pkg{BAMBI} function \code{fit_angmix} by specifying a model and the number of components to be used.  For example, the following \proglang{R} command fits a 4 component \code{vmsin} mixture by generating 3 MCMC chains with 20,000 samples each for the mixture model parameters. HMC  is used for sampling the model parameters, with tuning parameter \code{epsilon} adaptively tuned during burn-in (which is by default constituted by the first half of all iterations, i.e., first 10,000 iterations), and \code{L} taking its default value 10.  A \code{fit_angmix} call creates an \code{angmcmc} object, which can then be used for various post-processing tasks, including convergence assessment, parameter estimation and visualizing goodness of fit.

\begin{Sinput}
R> set.seed(12321)
R> fit.vmsin.4comp <- fit_angmix("vmsin", tim8, ncomp = 4, n.iter =  2e4,
+                                n.chains = 3)
\end{Sinput}

Note that in order for the independent chains to be run in parallel, an appropriate \code{plan()} from \proglang{R} package \pkg{future} needs to be set first; otherwise the chains will run sequentially. For example, running the commands 
\begin{Sinput}
R> library(future)
R> plan(multiprocess, gc = TRUE)
\end{Sinput}
before the \code{fit_angmix} call will ensure that the three chains are run in parallel, provided resources are available.  We suggest setting \code{gc = TRUE} in \code{plan}, to allow proper garbage collection from the parallel workers, even though it adds some overhead. This is because the parallel workers can end up leaving heavy memory footprints,  especially when mixture models are being fitted incrementally.


In the previous example with \code{fit_angmix}, the number of components $K$ was specified through \code{ncomp}. However, the ``true'' $K$ generally needs to be estimated from the data.  For this purpose, we use \code{fit_incremental_angmix} with \code{start_ncomp = 2}, which fits angular mixtures with incremental number of components (starting with 2 components), and uses a Bayesian model selection criterion to determine an optimal model. We use this function to fit optimal mixtures of \code{vmsin}, \code{vmcos} and \code{wnorm2} models separately. By specifying \code{n.chains = 3} and \code{n.iter = 20,000} in \code{fit_incremental_angmix}, each incremental model will be fitted using three chains with 20,000 iterations each (with first 10,000 iterations treated as burn-in, where \code{epsilon} in HMC is tuned).

By default, the algorithm uses MCMC-based MAP estimates from the preceding fitted model with one fewer component (if available) as starting parameter values (see Section~\ref{initial_vals}). We use the default leave one out cross validation information criterion (\code{'LOOIC'}) as the model selection criterion for determining the optimal model. 
When the function stops, we extract the best fitted model using  the function \code{bestmodel}, and assess convergence of the associated chains. After justifying convergence, we provide point and interval estimates of the parameters and visually examine goodness of fit.

\paragraph{Fitting the vmsin mixture model.}

We start with the \code{vmsin} model. The \proglang{R} commands are as follows:

\begin{Sinput}
R> set.seed(12321)
R> fit.vmsin <- fit_incremental_angmix(model="vmsin", data = tim8,
+                                      crit = "LOOIC",
+                                      start_ncomp = 2,
+                                      max_ncomp = 10,
+                                      n.iter = 2e4,
+                                      n.chains = 3)
\end{Sinput}

The algorithm stops at 5 components and determines the optimal number of components to be 4 on the basis of the \code{'LOOIC'} values. The MCMC-based maximum log-likelihood estimates for the intermediate models are  -945.3684, -853.3546, -803.4193, and -793.8154, which are steadily increasing. This is expected, since a ``smaller'' mixture should be nested within a ``larger'' mixture when properly fitted. We extract the optimum fitted model from the output of \code{fit_incremental_angmix} via  \code{bestmodel}:
\begin{Sinput}
R> fit.vmsin.best <- bestmodel(fit.vmsin)	
\end{Sinput}
	
Before estimating parameters, we first need to assess convergence and stationarity of the Markov chains. For this purpose, we first look at the (non-normalized) log posterior density (LPD) trace plots, which can be obtained using the \bambi function \code{lpdtrace}:
\begin{Sinput}
R> lpdtrace(fit.vmsin.best)
\end{Sinput}

\begin{figure}[hptb]
	\centering
	\includegraphics[height=4in, width = 4in]{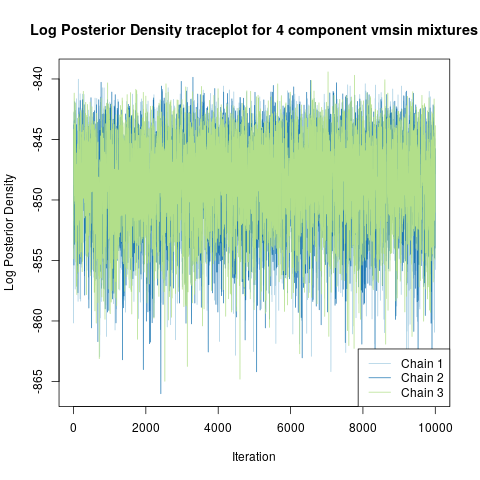}
	\caption{Log posterior trace plot for the Markov chain associated with the best fitted vmsin mixture model. \label{vmsinlpdtrace}}%
\end{figure}


The resulting plot displayed in Figure~\ref{vmsinlpdtrace} shows that all three chains have stabilized into similar LPD ranges  after burn-in, without noticeable trends or patterns.  Next, we look at the parameter traces, plotted using \bambi function \code{paramtrace}:
\begin{Sinput}
R> paramtrace(fit.vmsin.best)
\end{Sinput}
These trace plots are displayed in the panels of Figure~\ref{vmsin_paramtrace_plots_comp1} through Figure~\ref{vmsin_paramtrace_plots_comp4} in Appendix~\ref{vmsinparamtraces}, which show adequate signs of convergence and stationarity for samples within each chain. Stationarity of a chain can be formally tested using Geweke's convergence diagnostic  \citep{geweke:1991}, which tests equality (of means) of the first and last part of a Markov chain using   standard z-scores. The test is implemented in \proglang{R} package \pkg{coda}, and can be applied on \code{fit.vmsin.best}, first by converting it into a  \pkg{coda} \code{mcmc.list} object via S3 function \code{as.mcmc.list}:
\begin{Sinput}
R> mcmc.vmsin.best <- coda::as.mcmc.list(fit.vmsin.best)
\end{Sinput}
and then by applying the \pkg{coda} function \code{geweke.diag} on the output. We apply \code{geweke.diag} on \code{mcmc.vmsin.best} by setting both \code{frac1} and \code{frac2} equal to 0.5, to test the equality of the first and second halves. This produces a list of size 3 containing z-statistics for each chain. These values are displayed in  Figure~\ref{vmsindenplots} as barplots. The R commands are as follows:
\begin{Sinput}
R> geweke_res <- coda::geweke.diag(mcmc.vmsin.best, frac1 = 0.5, 
+                                  frac2 = 0.5)
R> par(mfrow = c(1, 3), mar = c(5, 6, 2, 1))
R> for(j in 1:3) {
+      barplot(geweke_res[[j]]$z, horiz = TRUE,
+              names = names(geweke_res[[j]]$z),
+              xlim = range(geweke_res[[j]]$z, -3, 3),
+              ylab="", xaxt='n', las = 2)
+      axis(1, las = 1)
+      title(main = paste("Chain", j), xlab = "geweke.diag")
+  }
\end{Sinput}

\begin{figure}[!htpb]
	\centering
	\includegraphics[width=\linewidth]{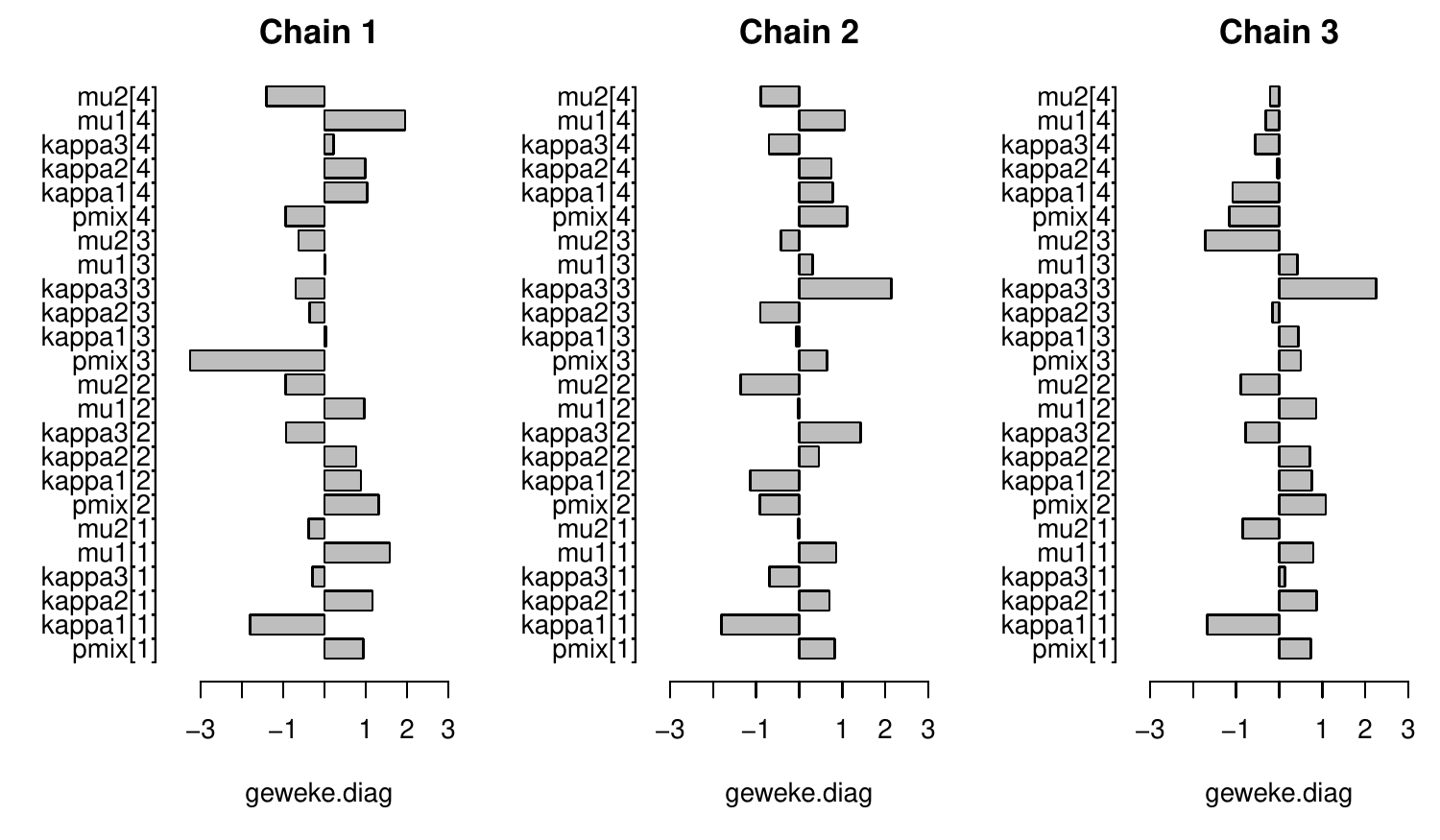}
	\caption{Geweke Diagnostic z scores for the three Markov chains associated with the best (4 component) vmsin mixture model.}
	\label{vmsingewekeplots}
\end{figure}

From these barplots, we see that all the z-scores more or less lie between $\pm 3$, thus indicating adequate similarity between the first and second halves of the chains. 

Package \pkg{coda} provides functions for a number of additional diagnostic plots and formal tests that may be used after conversion to a \code{mcmc.list} object.  Note that not all tests are applicable in every situation.  For example, the Gelman-Rubin test \citep{gelman:1992}, available via \pkg{coda} function \code{gelman.diag}, assumes normality of the posterior density; this assumption is clearly violated when the posterior density is multimodal. In this example, as well as for mixture models with several components in general, multimodality in the posterior density can be commonly observed.

Posterior multimodality is evident in the parameter traceplots for the current fit:  the sample values of the same parameter across the independent chains exhibit noticeable differences (see, e.g., the panels of Figure~\ref{vmsin_paramtrace_plots_comp2}). Note that this can be due to both permutation (i.e., label-switching) and non-permutation modes. These modes are from similar posterior density regions, as the  LPD traces show. This also indicates that various regions of the posterior density are being well explored by the three chains together. 

Next, we consider parameter estimation and assessing goodness of fit. Since  the combined MCMC samples are multimodal, the posterior mean point estimates from raw MCMC samples will not be meaningful, as they will lie in between modes. A simple alternative is to use the MCMC based approximate MAP estimate, which is unaffected by multimodality; proper care must be taken otherwise. The inferential difficulties associated with having permutation modes in the MCMC samples can be (potentially) solved by undoing label switching. The \bambi function \code{fix_label} (with the default settings) can be used for this purpose:
\begin{Sinput}
R> fit.vmsin.best <- fix_label(fit.vmsin.best)
\end{Sinput}
The parameter traceplots obtained (using \code{paramtrace} after undoing label switching) are displayed in Figures~\ref{vmsin_paramtrace_fix_plots_comp1} through \ref{vmsin_paramtrace_fix_plots_comp4} in Appendix~\ref{vmsinparamtraces_fix}. Compare these traceplots to the ones displayed in Appendix~\ref{vmsinparamtraces}. It can be seen that this procedure indeed removes some of the permutation modes, in that the parameter traces for the three independent chains now largely overlap. However, some non-unique modes are still present, as can be seen  in the traces of $\kappa_3$ and $\mu_2$ in component 4, displayed in panel (d) and (f) of Figure~\ref{vmsin_paramtrace_fix_plots_comp4}: e.g., $\kappa_3$ ``jumps'' between modes at approximately 2 and -2 over the course of the MCMC simulation. These modes might be genuine non-permutation modes, or permutation modes that \code{fix_label} is unable to resolve.

In \pkg{BAMBI},  parameter estimates are computed using the function \code{pointest}, which can find point estimates of the whole parameter vector, as well as its sub-vectors. Note that the function supports multiple methods of estimation. In particular, the argument \code{fn} in \code{pointest} specifies what function to evaluate on the MCMC samples for estimation. For example, \code{fn = mean} computes the MCMC posterior mean, while \code{fn = "MODE"} returns an MCMC based approximate MAP estimate. We use \code{pointest} to find the MAP and posterior mean estimates (after applying \code{fix_label}), and then note their differences. 
The \proglang{R} commands are as follows.
\begin{Sinput}
R> round(pointest(fit.vmsin.best, fn = "MODE"), 2)
\end{Sinput}
\begin{Soutput}
            1     2     3     4
pmix     0.43  0.16  0.35  0.06
kappa1  36.26  6.15  4.63  3.72
kappa2  27.93  2.25  7.80  0.00
kappa3 -12.03 -0.45 -0.58 -1.85
mu1      5.22  4.66  4.46  1.84
mu2      5.54  6.14  2.41  4.94
\end{Soutput}
\begin{Sinput}
R> round(pointest(fit.vmsin.best, fn = mean), 2)
\end{Sinput}
\begin{Soutput}
            1     2     3     4
pmix     0.43  0.17  0.34  0.06
kappa1  36.39  7.35  4.34  4.32
kappa2  29.19  2.06  8.12  0.07
kappa3 -12.74 -0.38 -1.08 -0.19
mu1      5.23  4.67  4.43  1.71
mu2      5.55  6.17  2.43  3.60
\end{Soutput}

We note that both the approximate MAP estimate and the posterior mean estimate reasonably agree on the first three components. However, they disagree on the remaining fourth component regarding the value of \code{mu2} and \code{kappa3}. This is not surprising, since the MCMC samples for these two parameters have (possibly non-permutation) multiple modes, as we saw earlier. In this case, their posterior mean estimates lie in between modes, and hence are not good point estimates. 
To visualize the differences between these estimates, we plot the  contours and  surfaces of the corresponding fitted model densities using the S3 functions \code{contour} (from \pkg{graphics}) and \code{densityplot} (from \pkg{lattice}) for \code{angmcmc} objects:
\begin{Sinput}
R> contour(fit.vmsin.best, fn = "MODE")
R> lattice::densityplot(fit.vmsin.best, fn = "MODE")
R> contour(fit.vmsin.best, fn = mean)
R> lattice::densityplot(fit.vmsin.best, fn = mean)
\end{Sinput}

\begin{figure}[!htpb]
	\centering
	\subcaptionbox{Approximate MAP estimate from MCMC samples. \label{vmsincon_raw}}%
	[.485\linewidth]{\includegraphics[height=3.2in, width = 3.2in]{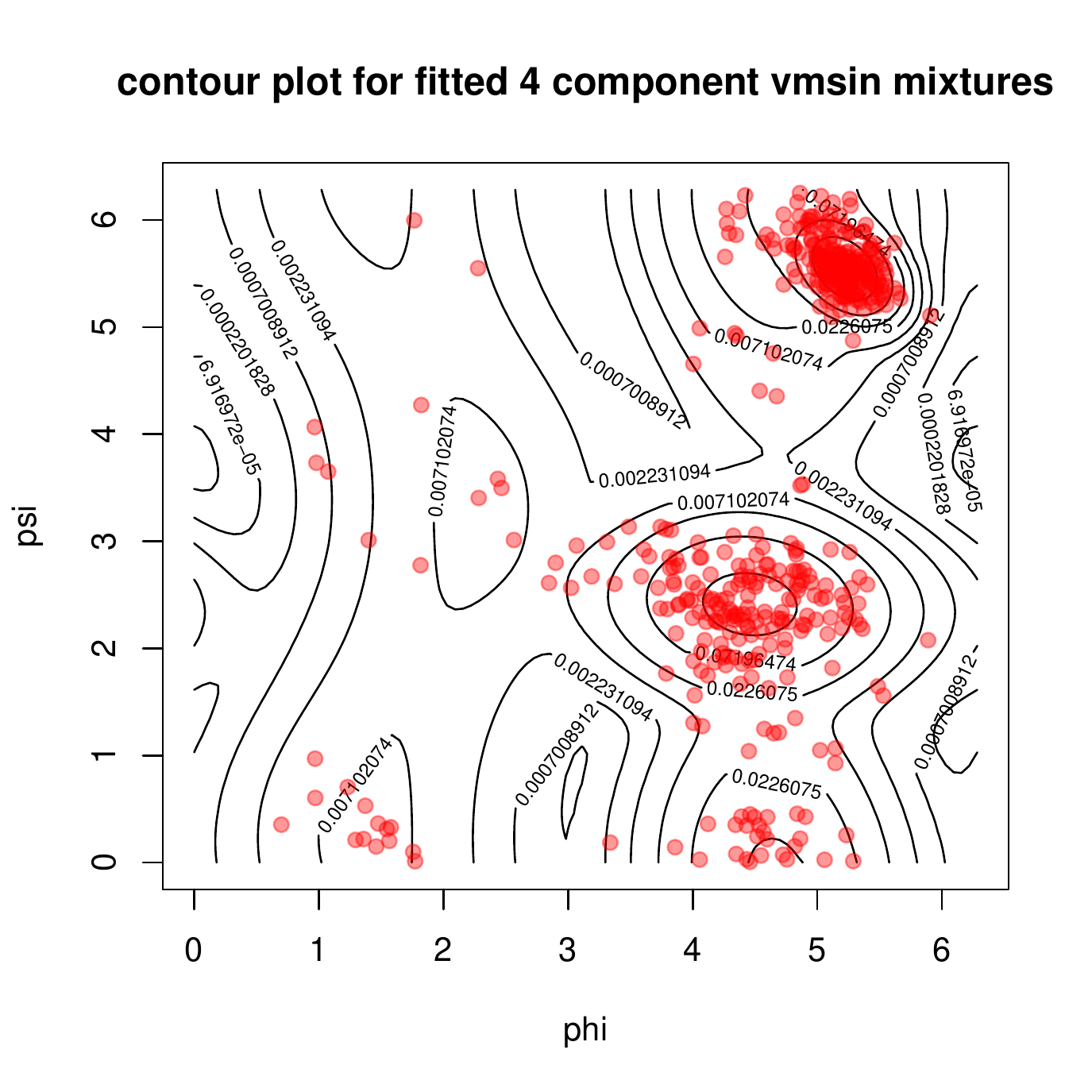}}
	\hfill
	\subcaptionbox{MCMC posterior mean after resolving label switching. \label{vmsincon}}%
	[.485\linewidth]{\includegraphics[height=3.2in, width = 3.2in]{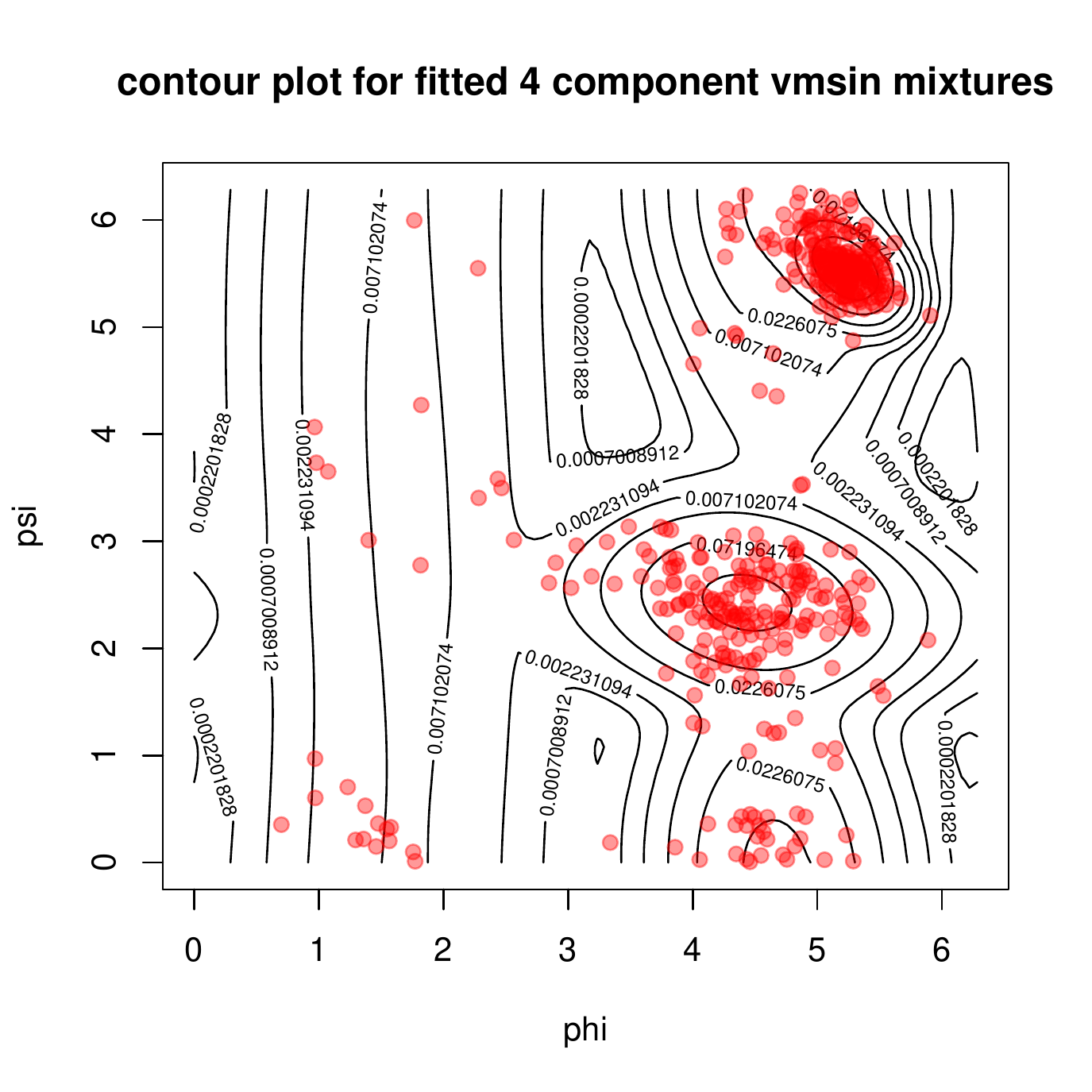}}
	\caption{Contour plots for fitted 4 component \code{vmsin} mixture model with parameters estimated via MCMC-based MAP estimation (left) and  posterior mean estimation (right).}
	\label{vmsinconplots}
\end{figure}

\begin{figure}[!htpb]
	\centering
	\subcaptionbox{Approximate MAP estimate from MCMC samples. \label{vmsinden_raw}}%
	[.485\linewidth]{\includegraphics[height=3.2in, width = 3.2in]{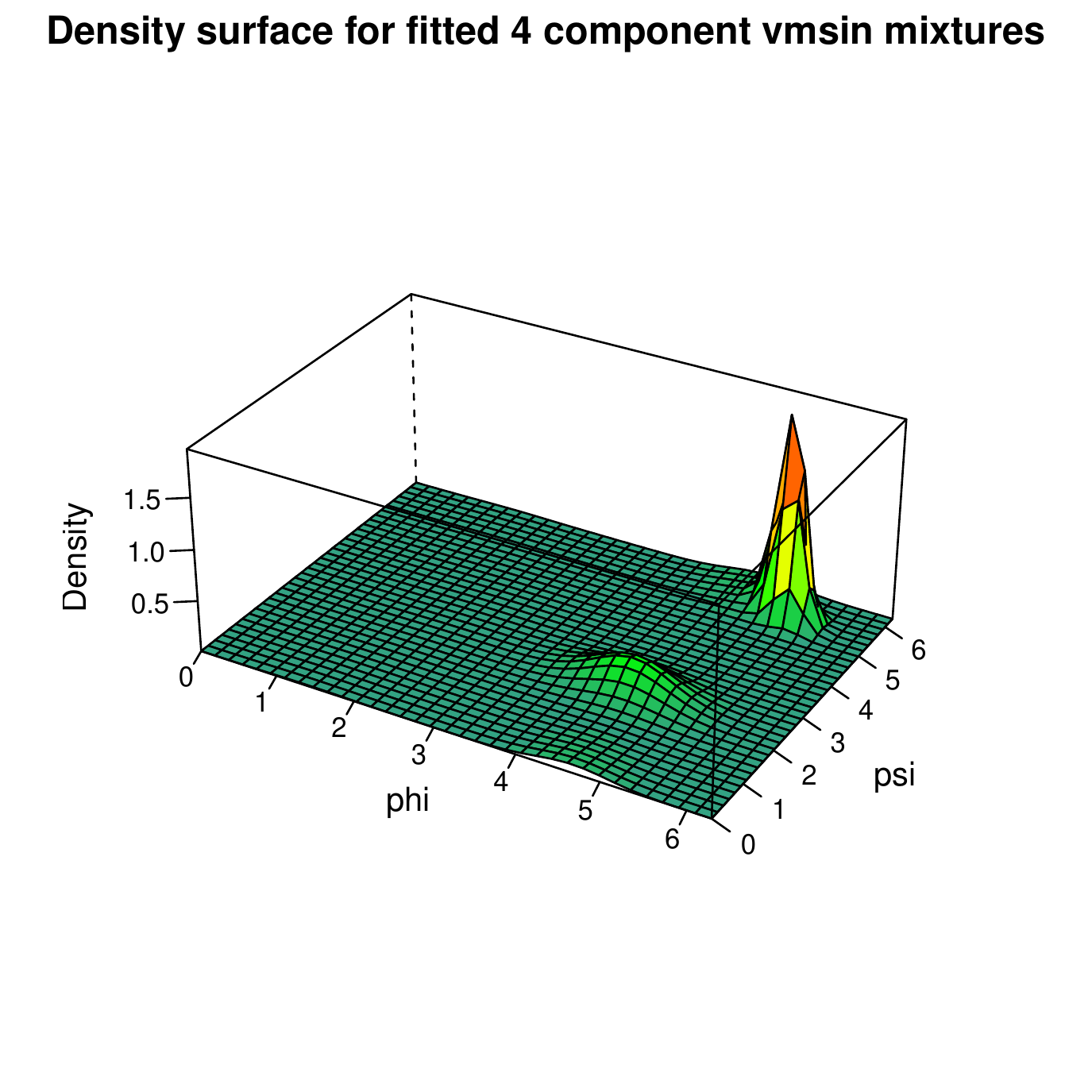}}
	\hfill
	\subcaptionbox{MCMC posterior mean after resolving label switching. \label{vmsinden}}%
	[.485\linewidth]{\includegraphics[height=3.2in, width = 3.2in]{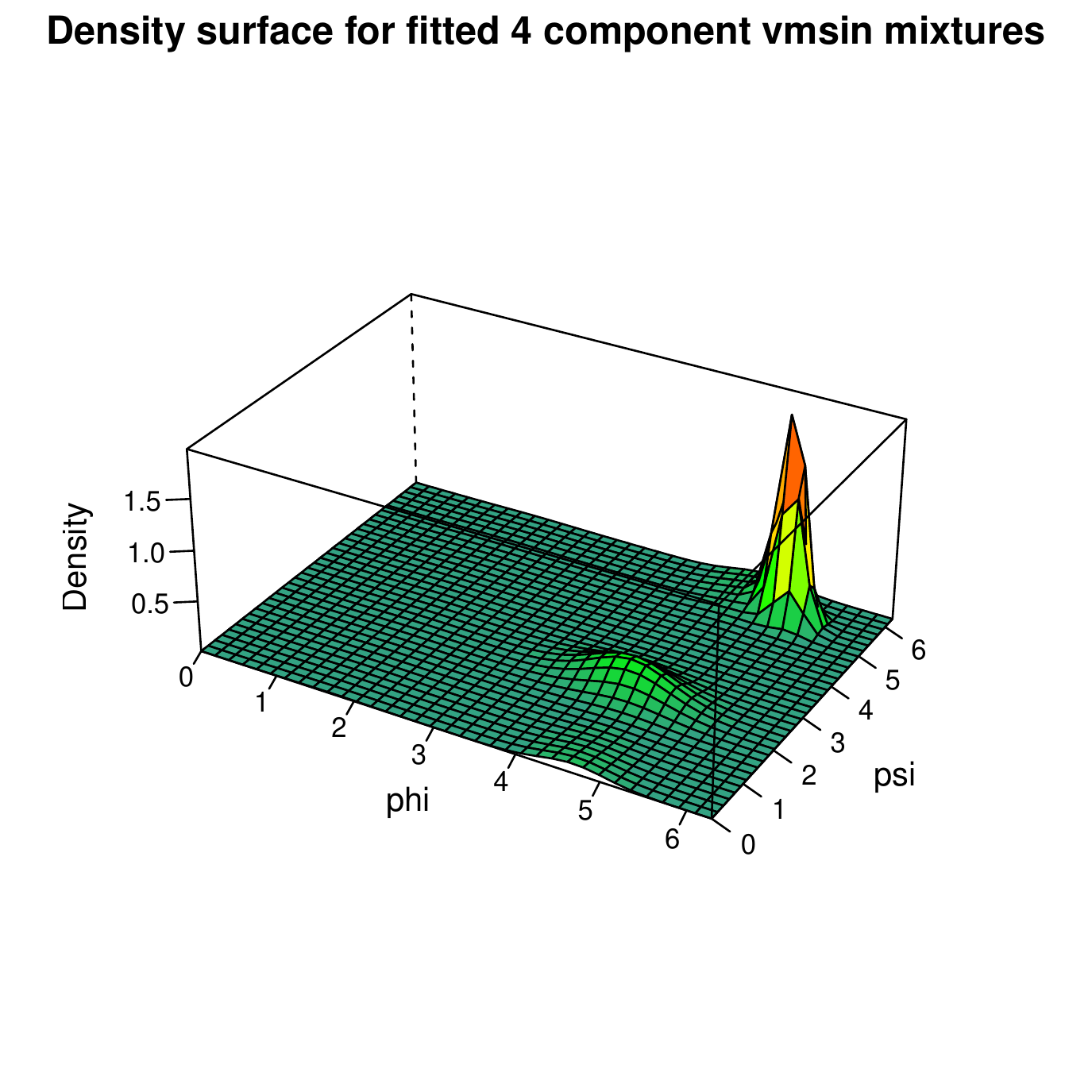}}
	\caption{Density surfaces for fitted 4 component \code{vmsin} mixture model with parameters estimated via MCMC-based MAP estimation (left) and  posterior mean estimation (right).}
	\label{vmsindenplots}
\end{figure}

The contour plots are shown in Figures~\ref{vmsincon_raw} and \ref{vmsincon}, and the density surfaces in Figures~\ref{vmsinden_raw} and \ref{vmsinden}. As can be seen, these plots are visually quite similar, despite the differences in the point estimates. This is due to the fact the component mostly affected by the existence of multiple modes has a low mixing proportion.  Nonetheless, in the current setting, the MAP estimate is the better of the two for the reasons described.  

Finally, we compute interval estimates of the parameters. This done by the S3 function \code{summary} for \code{angmcmc} objects, which computes the MCMC posterior mean along with a 95\% credible interval: 
\begin{Sinput}
R> summary(fit.vmsin.best)
\end{Sinput}

\begin{Soutput}
                            1                   2
pmix        0.43 (0.38, 0.48)   0.17 (0.12, 0.22)
kappa1   36.39 (29.01, 45.41)  7.35 (4.99, 10.68)
kappa2   29.19 (21.99, 38.96)   2.06 (1.21, 3.16)
kappa3 -12.74 (-18.49, -7.50) -0.38 (-1.76, 1.16)
mu1         5.23 (5.20, 5.26)   4.67 (4.54, 4.80)
mu2         5.55 (5.51, 5.58)   6.17 (5.99, 6.28)
                          3                    4
pmix      0.34 (0.29, 0.38)   0.064 (0.043, 0.09)
kappa1    4.34 (3.45, 5.43)     4.32 (1.64, 7.76)
kappa2   8.12 (6.00, 10.61) 0.073 (0.00012, 0.55)
kappa3 -1.08 (-2.29, 0.025)   -0.19 (-3.10, 3.05)
mu1       4.43 (4.35, 4.52)     1.71 (1.44, 2.05)
mu2       2.43 (2.37, 2.49)     3.60 (0.66, 5.97)
\end{Soutput}

%
%
%


We also use this example to illustrate \code{r{\_}fitted}, which generates random deviates from a fitted model, with parameters estimated using \code{pointest}. The corresponding function \code{d{\_}fitted} evaluates the density. These can be useful for posterior predictive checks. We draw observations from the best (4 component) fitted vmsin model, construct the Ramachandran plot for the generated dataset (exhibited in Figure~\ref{vmsinrama}) and compare it with the original Ramachandran plot. The following are the \proglang{R} commands used for this purpose.

\begin{Sinput}
R> set.seed(12321)
R> vmsin.data <- r_fitted(nrow(tim8), fit.vmsin.best, fn = "MODE")
R> plot(vmsin.data, xlab = "phi", ylab = "psi",
+       xlim = c(0, 2*pi), ylim = c(0, 2*pi), 
+       pch = 16, col = scales::alpha("black", 0.6))
R> title("Data generated from best fitted vmsin")
\end{Sinput}

\begin{figure}[!htpb]
	\centering
    \includegraphics[height=4in, width = 4in]{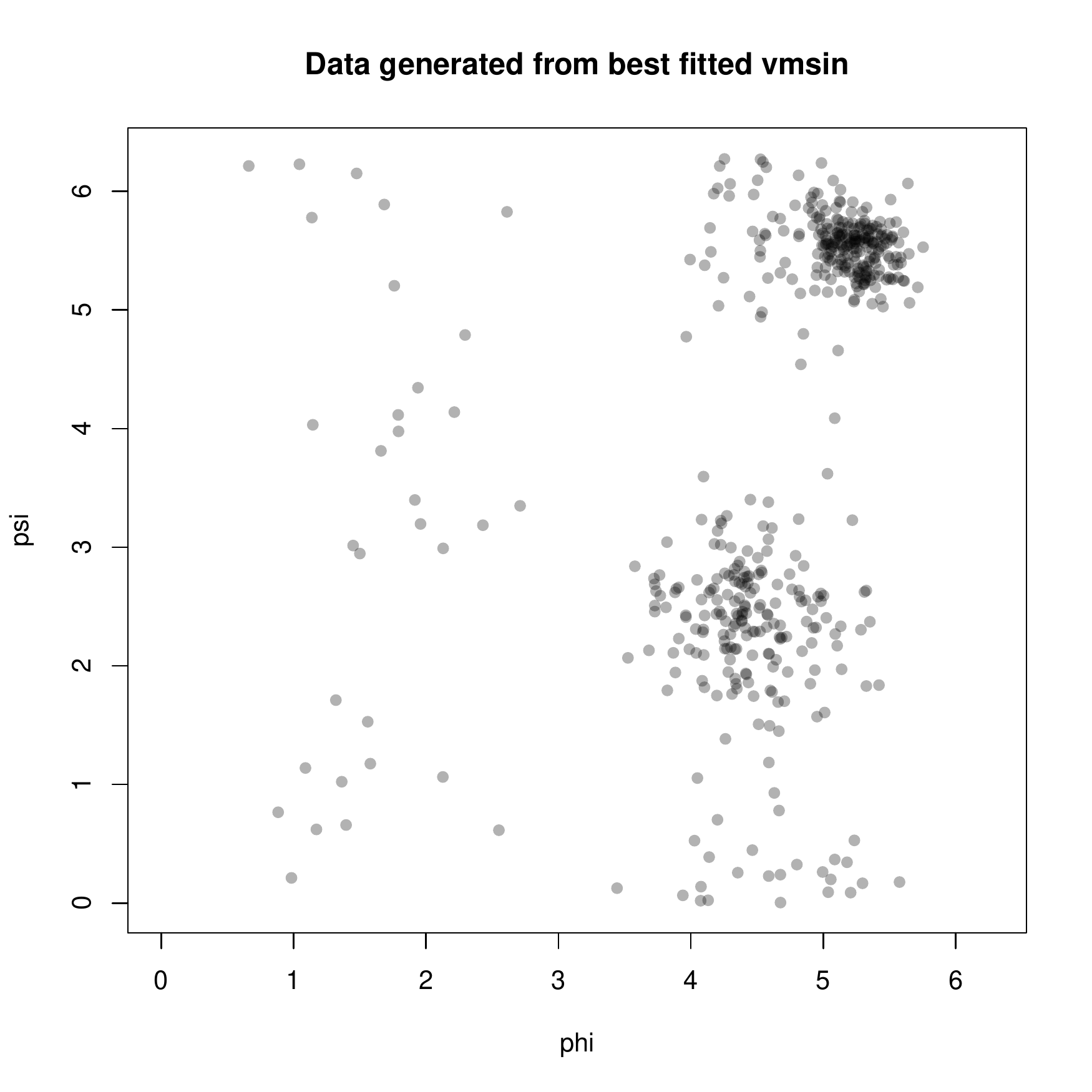}
	\caption{Ramachandran plot for the data generated from the best fitted vmsin mixture model. \label{vmsinrama}}
\end{figure}

Observe the similarity between Figures~\ref{rama} and \ref{vmsinrama}. The different clusters of the actual data points are reproduced well in the simulated data, which corroborate the goodness of fit.


\paragraph{Fitting the vmcos mixture model.}

Next, we fit \code{vmcos} mixtures to the data, by using \code{fit\_incremental\_angmix} with \code{model = "vmcos"}. We set \code{n\_qrnd = 1e4} (the default used in \code{dvmcos}), which specifies that 10,000 pairs of quasi-random Sobol numbers would be used to approximate the \code{vmcos} normalizing constant in cases where its analytic computaion is unstable. In a small dimensional problem with finite variance (such as ours), the Sobol sequence (or low discrepancy quasi-random sequences in general) often provides a better Monte Carlo approximation than (pseudo-) random sequences. In fact, for two dimensional problems, the rate of convergence for a Sobol sequence based Monte Carlo approximation is $O((\log N)^2/N)$ as opposed to $O(1/\sqrt{N})$ for an ordinary (pseudo-) random sequence based Monte Carlo approximation \citep{lemieux:2009}, where $N$ denotes the number of (quasi) random pairs used.  See the documentation of \code{dvmcos} for examples comparing analytic, quasi Monte Carlo, and ordinary Monte Carlo approximations  of the (normalizing constant of the) \code{vmcos} density. For \code{fit\_incremental\_angmix} (or more specifically in \code{fit\_angmix}),  Monte Carlo approximations based on $10^4$ pairs of Sobol numbers typically provide reasonable approximations, while keeping the computational burden moderate. 

The following are the \proglang{R} commands used for incrementally fitting  \code{vmcos}  mixture models.

\begin{Sinput}
R> set.seed(12321)
R> fit.vmcos <- fit_incremental_angmix(model="vmcos", data = tim8,
+                                      crit = "LOOIC",
+                                      start_ncomp = 2,
+                                      max_ncomp = 10,
+                                      n.iter = 2e4,
+                                      n.chains = 3,
+                                      use_best_chain = FALSE,
+                                      n_qrnd = 1e4)
\end{Sinput}

Similar to the \code{vmsin} case, here also the algorithm stops at 5 components and determines the optimal number of components to be 4. We first extract the ``best'' fitted model, via \code{bestmodel}:
\begin{Sinput}
R> fit.vmcos.best <- bestmodel(fit.vmcos)
\end{Sinput}
and then plot the log posterior and parameter traces. These plots show similar convergence properties, and are omitted for brevity. For parameter estimation, we compute both (approximate) MAP and posterior mean estimates (after undoing label switching), and plot the contour and surface of the associated fitted model densities. The following are the associated R commands:
\begin{Sinput}
R> fit.vmcos.best <- fix_label(fit.vmcos.best)
R> contour(fit.vmcos.best, fn = "MODE")
R> lattice::densityplot(fit.vmcos.best, fn = "MODE")
R> contour(fit.vmcos.best, fn = mean)
R> lattice::densityplot(fit.vmcos.best, fn = mean)
\end{Sinput}


\begin{figure}[!htpb]
	\centering
	\subcaptionbox{Approximate MAP estimate from MCMC samples. \label{vmcoscon_raw}}%
	[.485\linewidth]{\includegraphics[height=3.2in, width = 3.2in]{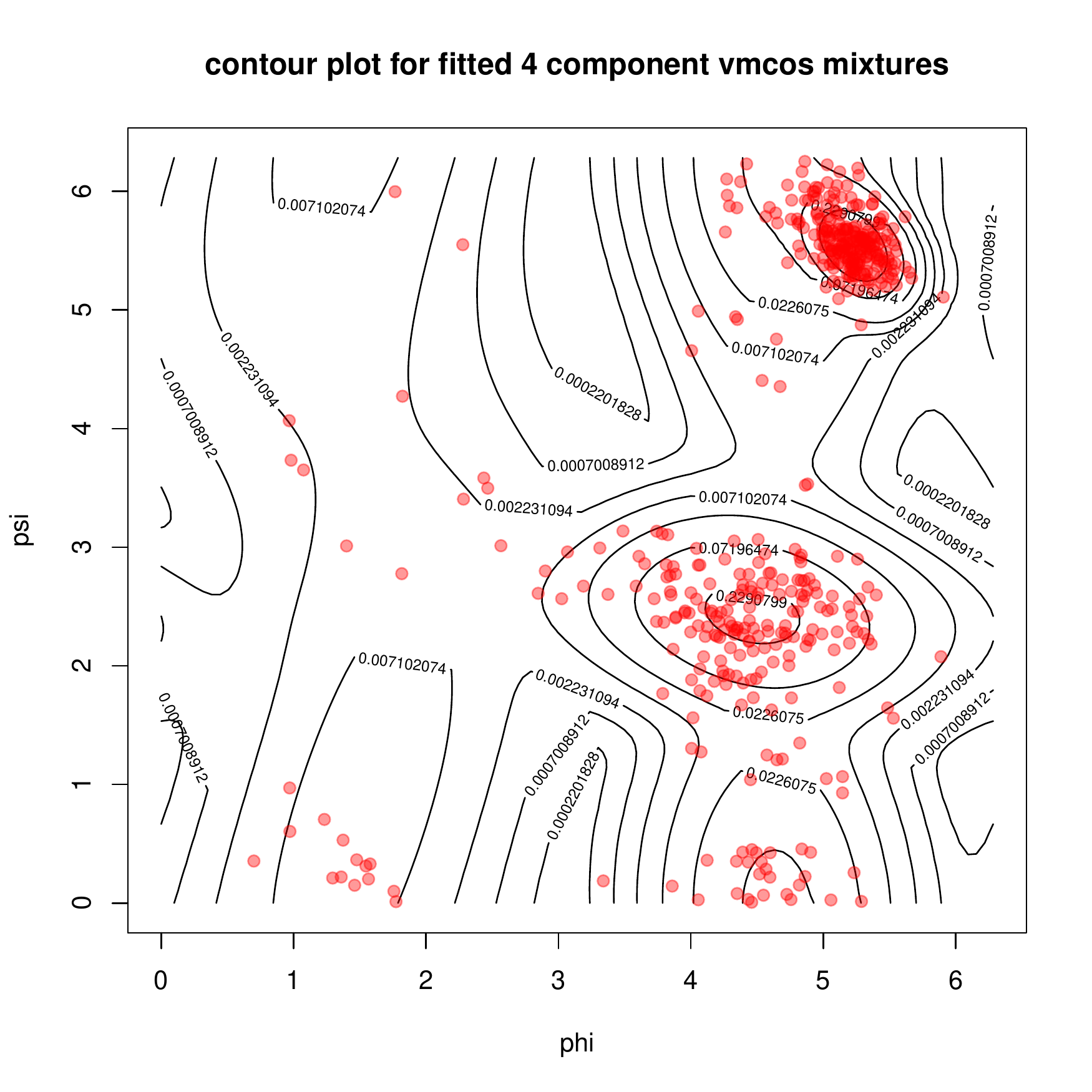}}
	\hfill
	\subcaptionbox{MCMC posterior mean after resolving label switching. \label{vmcoscon}}%
	[.485\linewidth]{\includegraphics[height=3.2in, width = 3.2in]{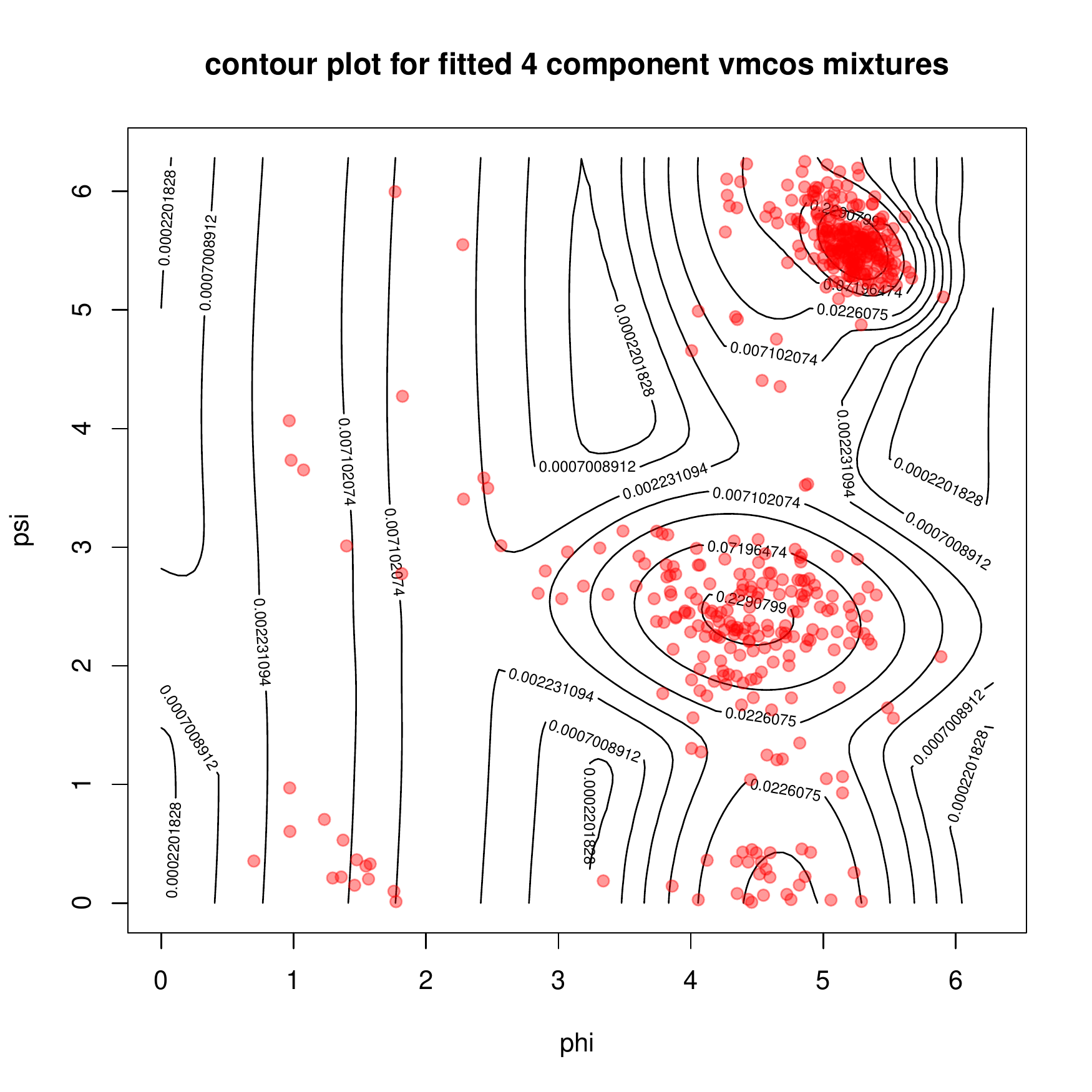}}
	\caption{Contour plots for fitted 4 component \code{vmcos} mixture model with parameters estimated via MCMC-based MAP estimation (left) and  posterior mean estimation (right)}
	\label{vmcosconplots}
\end{figure}

\begin{figure}[!htpb]
	\centering
	\subcaptionbox{Approximate MAP estimate from MCMC samples. \label{vmcosden_raw}}%
	[.485\linewidth]{\includegraphics[height=3.2in, width = 3.2in]{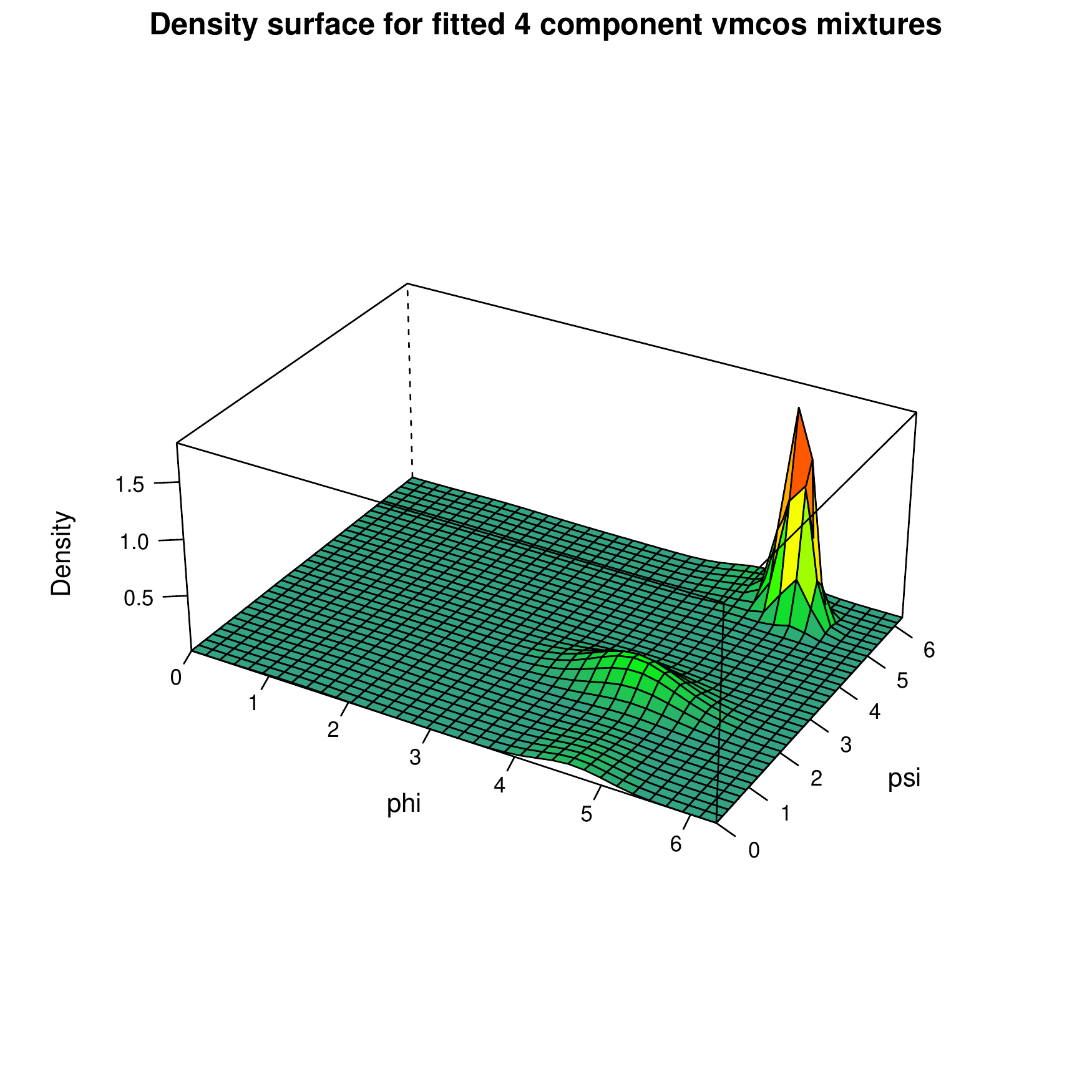}}
	\hfill
	\subcaptionbox{MCMC posterior mean after resolving label switching. \label{vmcosden}}%
	[.485\linewidth]{\includegraphics[height=3.2in, width = 3.2in]{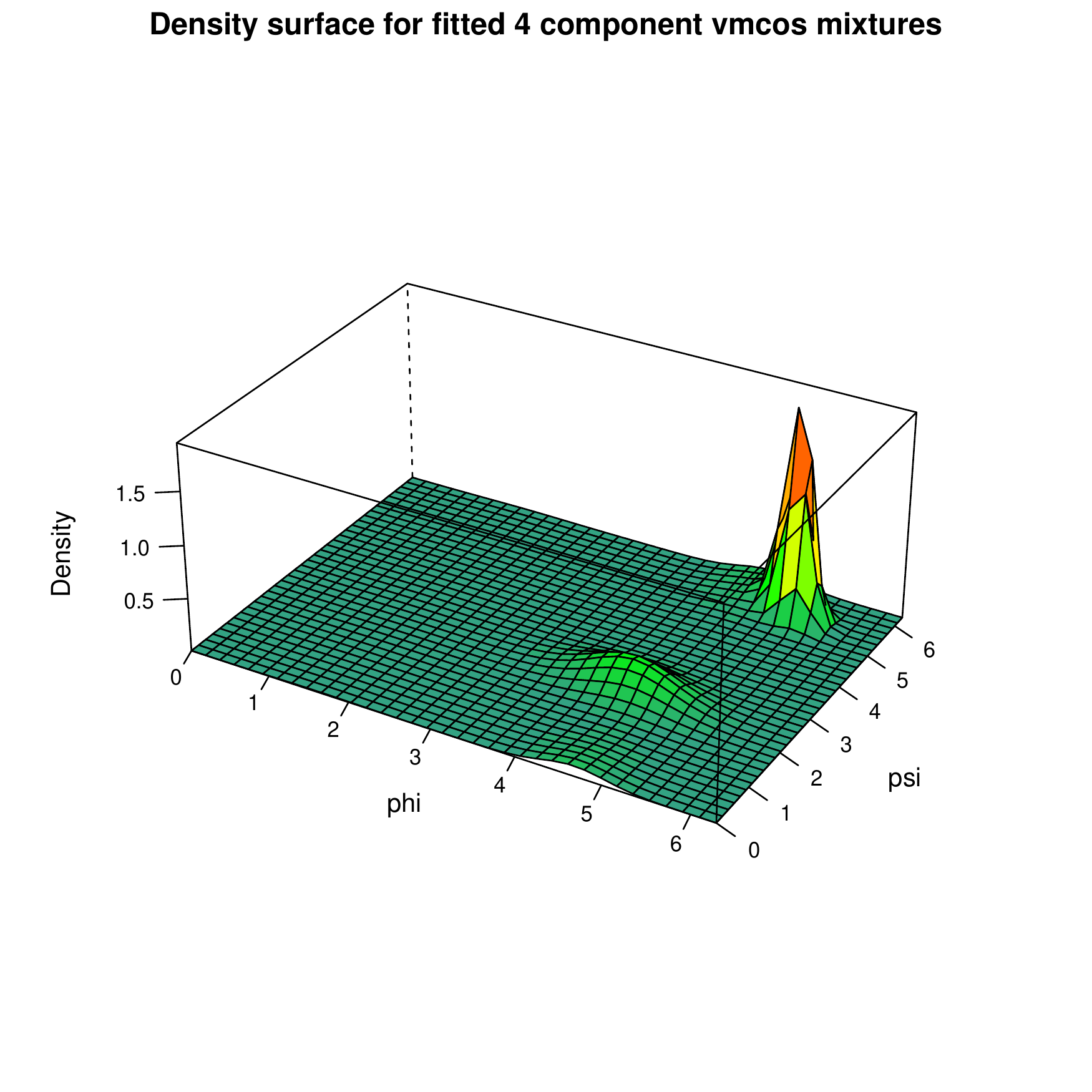}}
	\caption{Density surfaces for fitted 4 component \code{vmcos} mixture model with parameters estimated via MCMC-based MAP estimation (left) and  posterior mean estimation (right).}
	\label{vmcosdenplots}
\end{figure}

The fitted contours are displayed in Figures~\ref{vmcoscon_raw} and \ref{vmcoscon}, and the density surfaces are displayed in Figures~\ref{vmcosden_raw} and \ref{vmcosden}, are noticeably similar.  They  are also broadly similar to the ones obtained for the fitted \code{vmsin} models.  Estimated posterior means along with estimated 95\% credible intervals are   obtained using the S3 function \code{summary.angmcmc} as follows.

%

\begin{Sinput}
R> summary(fit.vmcos.best)
\end{Sinput}

\begin{Soutput}
                            1                     2
pmix        0.43 (0.38, 0.49)     0.16 (0.12, 0.21)
kappa1   48.70 (37.89, 61.88)    7.97 (5.11, 12.56)
kappa2   41.49 (30.79, 55.64)    2.15 (0.002, 4.57)
kappa3 -12.61 (-18.64, -7.35) -0.0077 (-1.90, 2.12)
mu1         5.23 (5.20, 5.26)     4.67 (4.54, 4.81)
mu2         5.55 (5.51, 5.58)     6.17 (5.99, 6.28)
                         3                    4
pmix     0.34 (0.29, 0.38) 0.065 (0.043, 0.095)
kappa1   5.16 (3.80, 6.69)    3.62 (0.90, 6.77)
kappa2  8.79 (6.36, 11.81) 0.19 (0.00013, 1.33)
kappa3 -0.98 (-2.15, 0.12)  -0.20 (-1.54, 0.90)
mu1      4.43 (4.34, 4.51)    1.59 (1.31, 1.93)
mu2      2.43 (2.36, 2.49)    3.13 (0.14, 6.13)
\end{Soutput}


\paragraph{Fitting the wnorm2 mixture model.}

Finally, we fit \code{wnorm2} mixtures to the data. The \proglang{R} commands used are as follows:

\begin{Sinput}
R> set.seed(12321)
R> library(future)
R> plan(multiprocess(workers = 3))
R> fit.wnorm2 <- fit_incremental_angmix(model="wnorm2", data = tim8,
+                                       crit = "LOOIC",
+                                       start_ncomp = 2,
+                                       max_ncomp = 10,
+                                       n.iter = 2e4,
+                                       n.chains = 3,
+                                       use_best_chain = FALSE)
\end{Sinput}

Here also, the function stops at 5 components and determines the optimal number of components to be 4.\footnote{It should be noted that the runtime for wrapped normal fitting is considerably longer than the von Mises sine models, due to the computational burden; see Section~\ref{wnmodels}.} As done in the previous two cases, after extracting the best model, we assess convergence via trace plots (omitted for brevity). We  find MCMC-based MAP and posterior mean estimates (after undoing label switching), and also find credible interval estimates. Finally we assess goodness of fit through fitted contour and density surfaces. The following are the R commands that perform these tasks.

\begin{Sinput}
R> fit.wnorm2.best <- bestmodel(fit.wnorm2)
R> lpdtrace(fit.wnorm2.best)
R> paramtrace(fit.wnorm2.best)
R> fit.wnorm2.best <- fix_label(fit.wnorm2.best)
R> contour(fit.wnorm2.best, fn = "MODE")
R> lattice::densityplot(fit.wnorm2.best, fn = "MODE")
R> contour(fit.wnorm2.best, fn = mean)
R> lattice::densityplot(fit.wnorm2.best, fn = mean)
R> summary(fit.wnorm2.best)
\end{Sinput}

\begin{figure}[!htpb]
	\centering
	\subcaptionbox{Approximate MAP estimate from MCMC samples. \label{wnorm2con_raw}}%
	[.485\linewidth]{\includegraphics[height=3.2in, width = 3.2in]{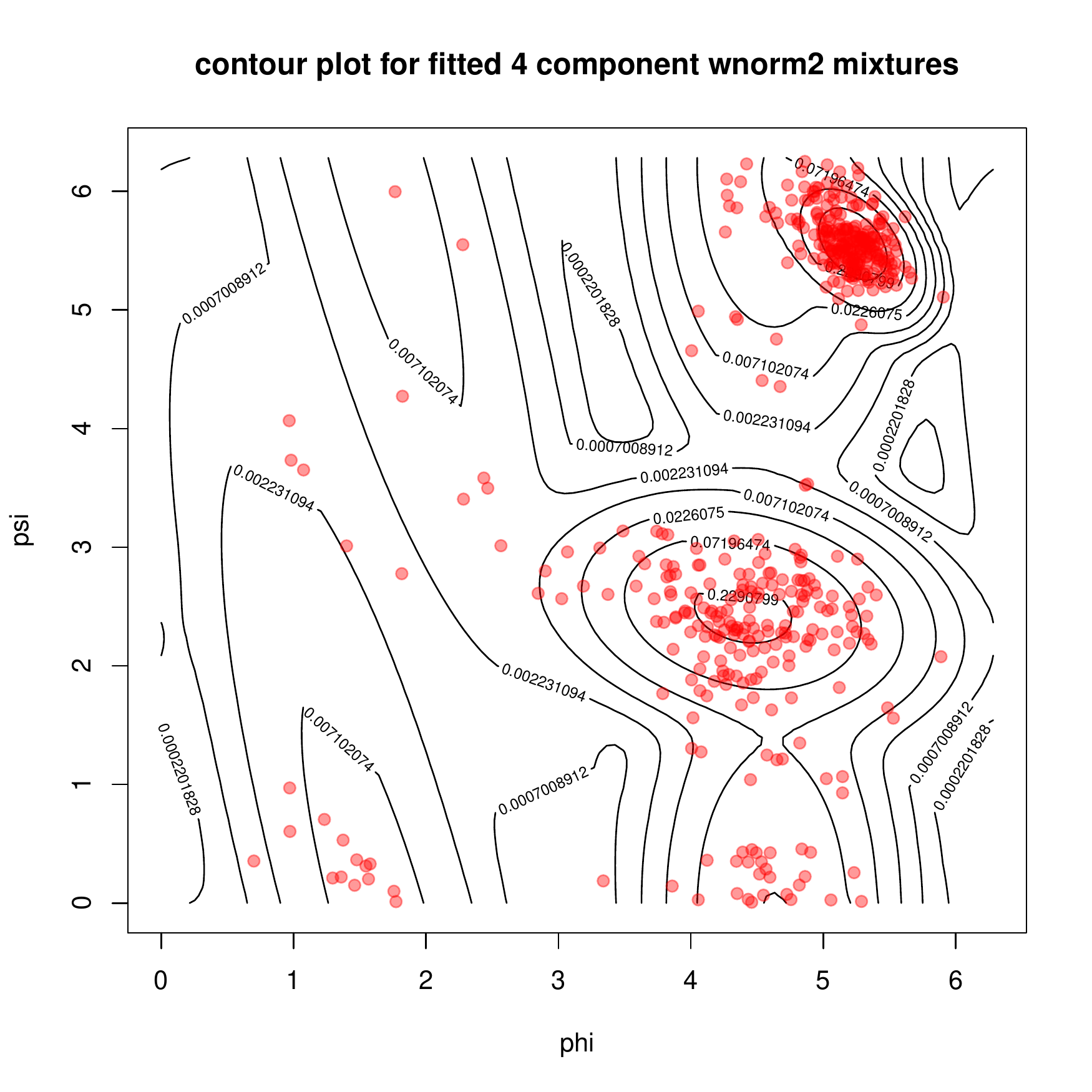}}
	\hfill
	\subcaptionbox{MCMC posterior mean after resolving label switching. \label{wnorm2con}}%
	[.485\linewidth]{\includegraphics[height=3.2in, width = 3.2in]{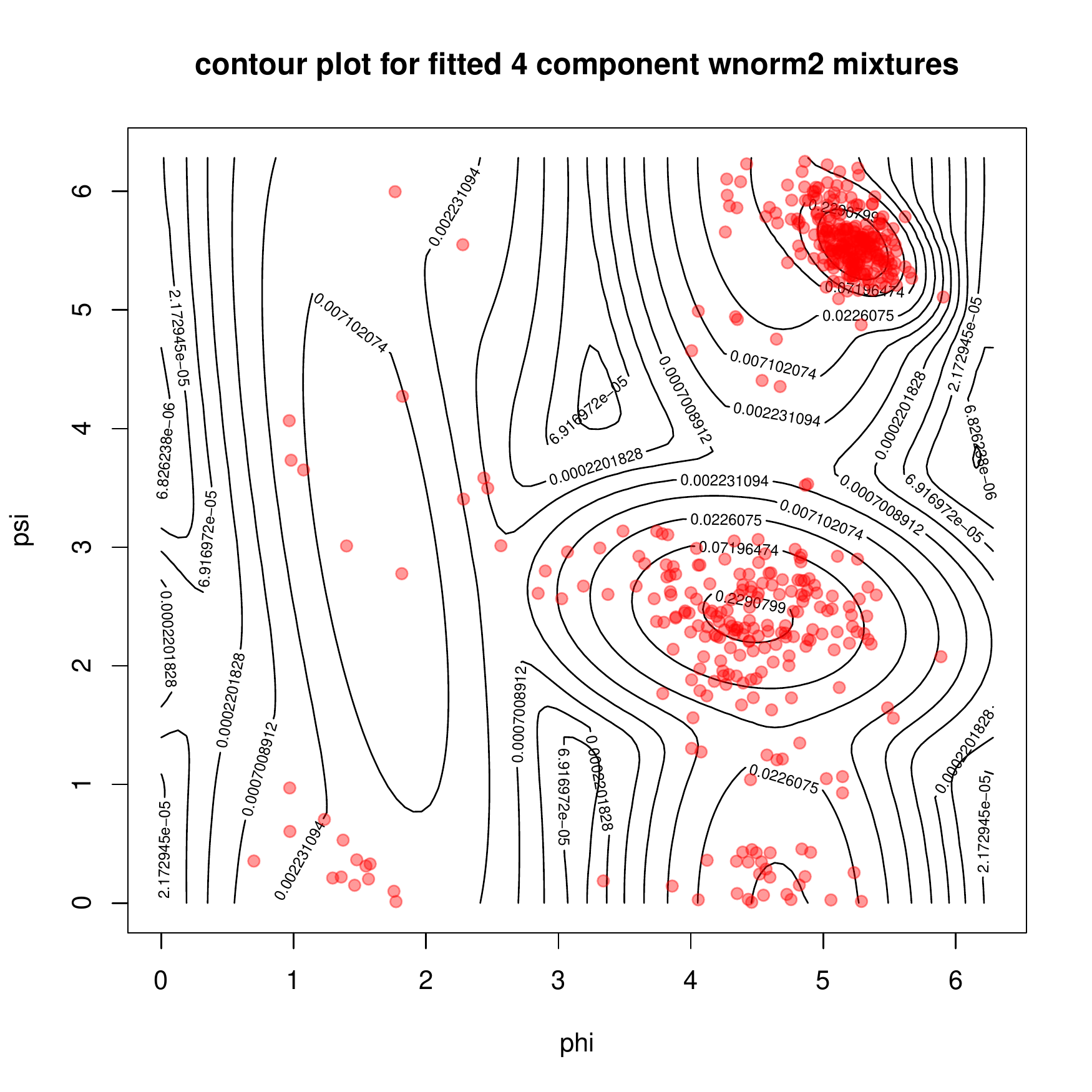}}
	\caption{Contour plots for fitted 4 component \code{wnorm2} mixture model with parameters estimated via MCMC-based MAP estimation (left) and  posterior mean estimation (right).}
	\label{wnorm2conplots}
\end{figure}

\begin{figure}[!htpb]
	\centering
	\subcaptionbox{Approximate MAP estimate from MCMC samples. \label{wnorm2den_raw}}%
	[.485\linewidth]{\includegraphics[height=3.2in, width = 3.2in]{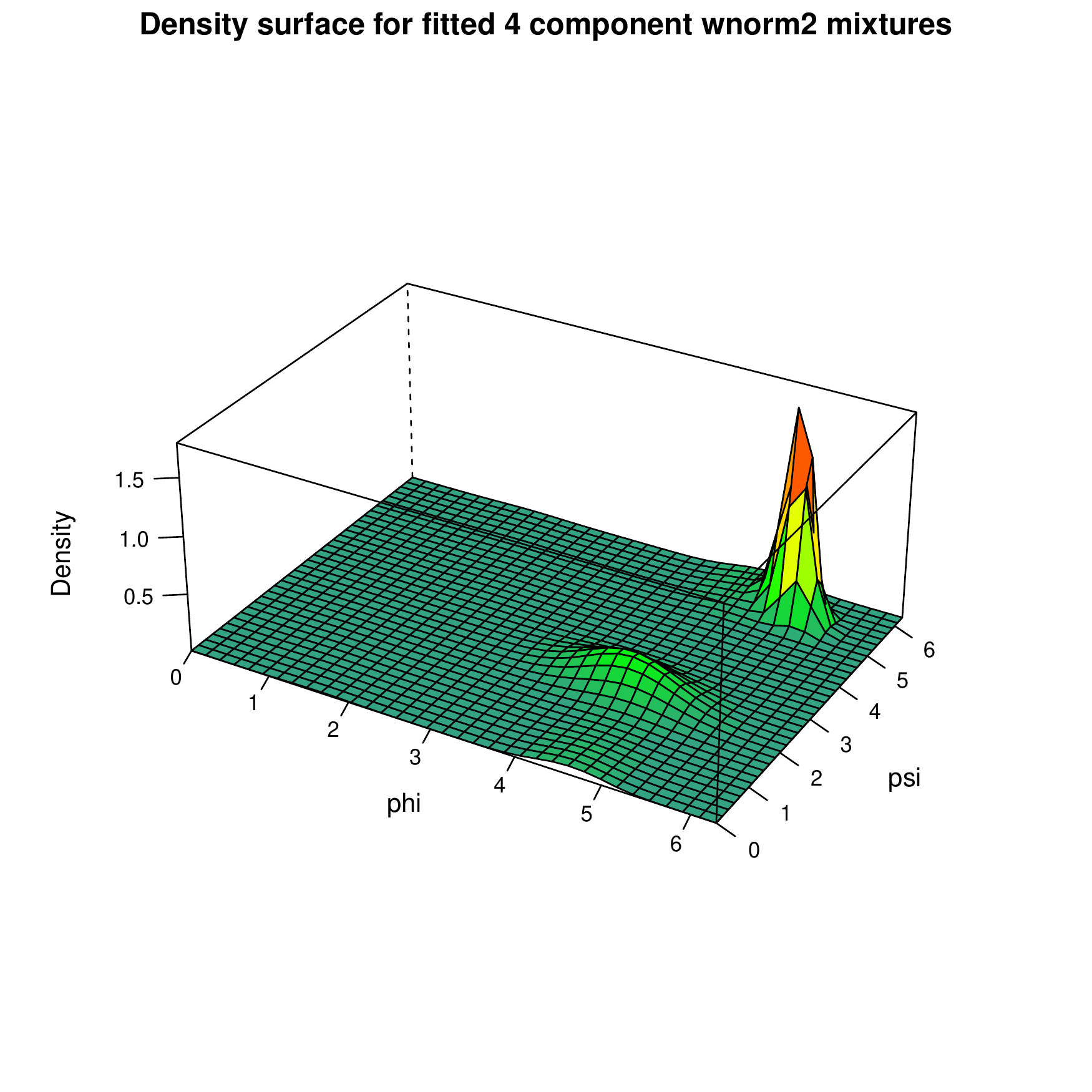}}
	\hfill
	\subcaptionbox{MCMC posterior mean after resolving label switching. \label{wnorm2den}}%
	[.485\linewidth]{\includegraphics[height=3.2in, width = 3.2in]{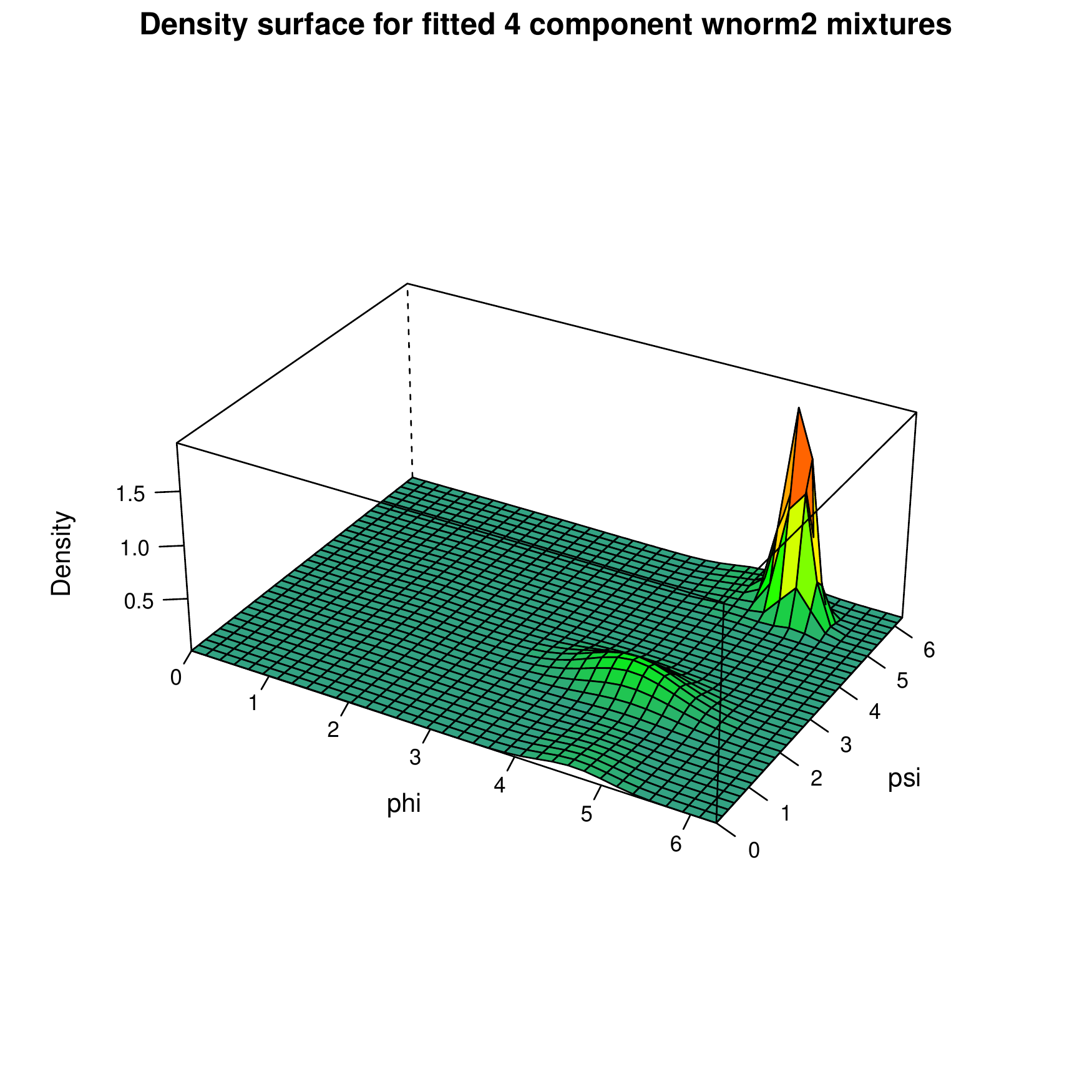}}
	\caption{Density surfaces for fitted 4 component \code{wnorm2} mixture model with parameters estimated via MCMC-based MAP estimation (left) and  posterior mean estimation (right).}
	\label{wnorm2denplots}
\end{figure}

The contours and density surfaces  displayed in Figures~\ref{wnorm2conplots} and \ref{wnorm2denplots}) are noticeably  similar.  They are also broadly similar to the fitted \code{vmsin} and \code{vmcos} mixture model density contours and surfaces. The estimated credible intervals along with MCMC posterior means for fitted 4 component wnorm2 mixture are obtained using the S3 function \code{summary.angmcmc} as follows.

\begin{Sinput}
R> summary(fit.wnorm2.best)
\end{Sinput}

\begin{Soutput}
                          1                  2
pmix      0.44 (0.38, 0.49)  0.16 (0.12, 0.21)
kappa1 35.57 (28.05, 44.49) 7.13 (4.55, 10.90)
kappa2 28.03 (20.93, 36.91)  1.48 (0.75, 2.44)
kappa3  12.32 (7.38, 18.02) 0.20 (-0.90, 1.12)
mu1       5.23 (5.20, 5.26)  4.66 (4.54, 4.78)
mu2       5.55 (5.52, 5.58)  6.17 (5.98, 6.28)
                        3                    4
pmix    0.34 (0.29, 0.38) 0.063 (0.043, 0.087)
kappa1  3.68 (2.86, 4.65)   8.05 (2.42, 17.30)
kappa2 7.68 (5.62, 10.32)   0.46 (0.034, 0.91)
kappa3  1.00 (0.10, 2.00)   1.13 (-1.36, 3.09)
mu1     4.43 (4.35, 4.51)    1.62 (1.37, 1.94)
mu2     2.43 (2.37, 2.50)   3.06 (0.044, 6.23)
\end{Soutput}

\paragraph{Comparative analysis of the three fitted models.}

So far, we have considered mixtures of \code{vmsin}, \code{vmcos} and \code{wnorm2} densities, and have fitted them to \code{tim8} data. The associated optimal number of components were determined via leave one out information criterion (LOOIC) in incremental fitting schemes. We then plotted the fitted density contours and  surfaces to assess the goodness of fit, and noticed that these plots are broadly similar across the three optimal fitted mixture models (of \code{vmsin}, \code{vmcos} and \code{wnorm2} densities). It is natural to then consider the question of which among these three fitted bivariate mixture models best explains the data. This can again be answered via LOOIC. For \code{angmcmc} objects LOOIC can be conveniently computed using the S3 function \code{looic} from package \pkg{loo}. However, here we do not need to recompute them since they were already computed during incremental model fitting, and can be extracted via the convenience function \code{bestcriterion} as follows:
\begin{Sinput}
R> vmsin.4.looic <- bestcriterion(fit.vmsin)
R> wnorm2.4.looic <- bestcriterion(fit.wnorm2)
R> vmcos.4.looic <- bestcriterion(fit.vmcos)
\end{Sinput}

Now we compare the three models via \code{loo::compare} on the basis of their LOOIC's:
\begin{Sinput}
R> comp <- loo::compare(vmsin.4.looic, vmcos.4.looic, wnorm2.4.looic)
R> comp
\end{Sinput}
\begin{Soutput}
               elpd_diff se_diff elpd_loo p_loo  looic 
vmsin.4.looic     0.0       0.0  -826.0     25.5 1652.0
wnorm2.4.looic   -0.9       3.6  -826.9     26.5 1653.8
vmcos.4.looic    -5.3       3.1  -831.3     25.2 1662.7
\end{Soutput}

The documentation of \code{loo} from package \pkg{loo} v2.1.0 says ``When comparing two fitted models, we can estimate the difference in their expected predictive accuracy by the difference in \code{elpd_loo} or \code{elpd_waic} (or multiplied by -2, if desired, to be on the deviance scale). When that difference, \code{elpd_diff}, is positive then the expected predictive accuracy for the second model is higher. A negative \code{elpd_diff} favors the first model. When using \code{compare()} with more than two models, the values in the \code{elpd_diff} and \code{se_diff} columns of the returned matrix are computed by making pairwise comparisons between each model and the model with the best ELPD (i.e., the model in the first row)''. 

Thus the above output provides a ranking among the three models based on their (estimated) expected log predictive density (elpd) values; a higher elpd indicates a better predictive accuracy and thus a better fit.  The fitted \code{vmsin} model appears to have the highest elpd (see the column \code{elpd_diff} in the above output), followed by the fitted \code{wnorm2} model and the fitted \code{vmcos} model. However, these elpd's are estimates, and the variabilities of these estimates need to be considered when making comparisons. To address this, we make use of the standard errors of the differences provided in the column \code{se_diff} and construct approximate 95\% credible interval estimates of the pairwise elpd differences (viz., \code{elpd_diff} $\pm$ 2 \code{se_diff}) for the fitted model pairs (\code{wnorm2}, \code{vmsin}) and  (\code{vmcos}, \code{vmsin}). A elpd difference is considered to be significant (at the 95\% level) if the corresponding interval estimate does not contain zero. The \proglang{R} commands are as follows.
\begin{Sinput}
R> find_ci <- function(x, digits = 1) {
+    round(c(lower = unname(x[1] - 2*x[2]),
+            upper = unname(x[1] + 2*x[2])),
+          digits = digits)
+ }
R> t(apply(comp[-1, c("elpd_diff", "se_diff")], 1, find_ci))
\end{Sinput}
\begin{Soutput}
               lower upper
wnorm2.4.looic  -8.1   6.3
vmcos.4.looic  -11.5   0.9
\end{Soutput}        
This shows that approximate 95\% interval estimates of the elpd differences between the fitted best \code{vmsin} and the best \code{wnorm2} model, and the best \code{vmsin} and the best \code{vmcos} model, are $(-8.1, 6.3)$ and $(-11.5, 0.9)$ respectively, both containing zero. It therefore follows that all three of the fitted (four component) mixture models are not significantly different in terms of their goodness of fit to these data.

\subsection{Fitting mixture models on the wind (univariate) data}

The \code{wind} data contains 239 observations on wind direction in radians  measured at Saturna Island, British Columbia, Canada, during October 1-10, 2016. As a result of a severe storm that occurred during that period, the data shows significant variability with an interesting bi- (or possibly tri-) modality. Figure~\ref{windhist} shows a histogram of the data, constructed by applying the default \code{hist} function on \code{wind[, "angle"]}.

\begin{figure}[hptb]
	\centering
	\includegraphics[width=4in]{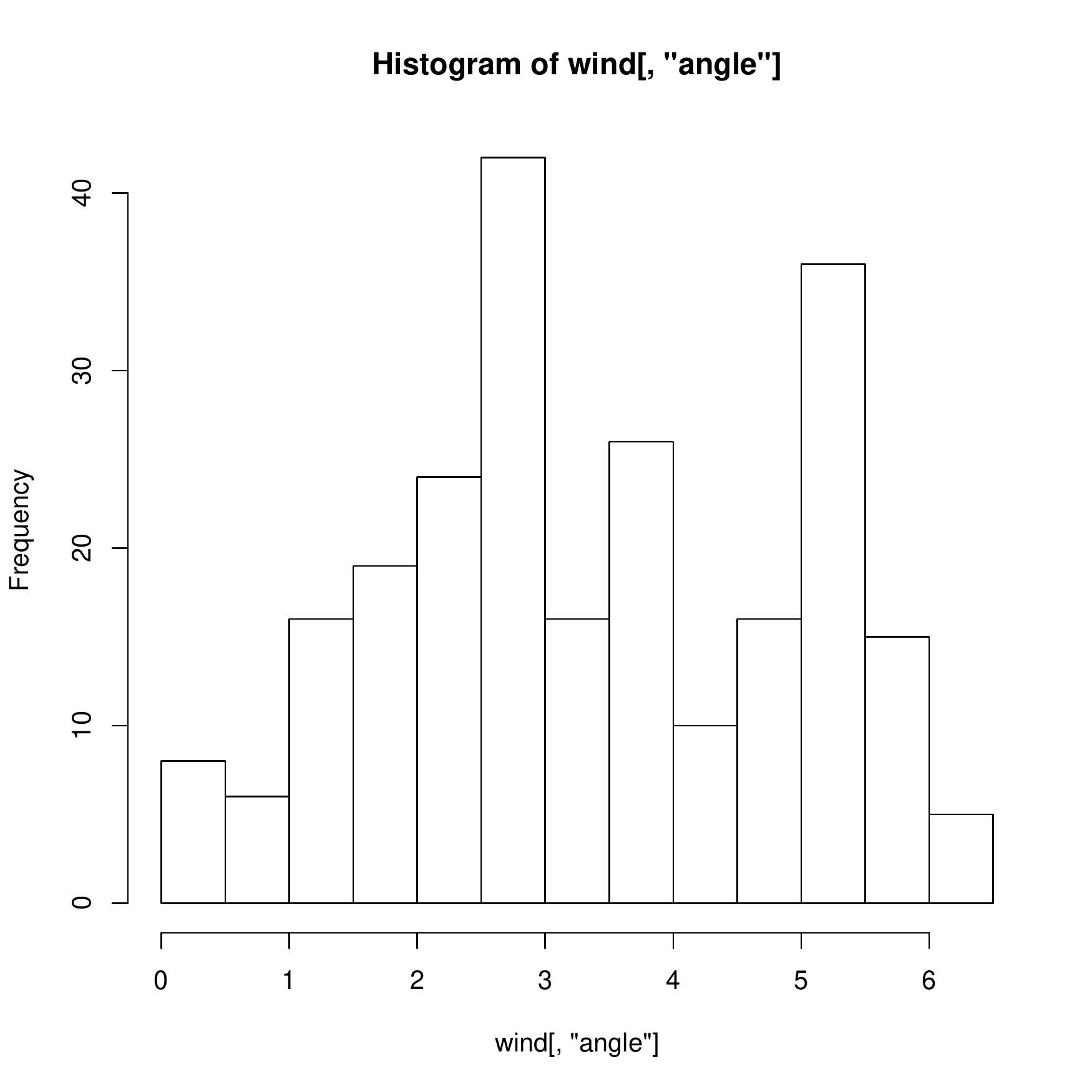}
	\caption{Histogram for the wind dataset.}
\end{figure} \label{windhist}

Similar to the bivariate case,  we  use \code{fit_incremental_angmix} to fit mixtures of \code{vm} and \code{wnorm} separately with incremental number of components (starting at 1) and determine an optimum size in each case.  To fit  each mixture model, we first generate 20,000 MCMC samples for the parameters, with the (default) first half taken as burn-in. Except for \code{n.iter}, defaults for all other arguments are used in these examples. After generating the MCMC samples, we assess their convergence via LPD and parameter trace plots. Following, we visualize the fits via  density curves constructed using the S3 function \code{densityplot} (which requires \pkg{lattice}). Finally, we compute point and interval estimates for each parameter using the S3 function \code{summary.angmcmc}.

\paragraph{Fitting the vm mixture model.}
We start with \code{vm},  The \proglang{R} commands are as follows:
\begin{Sinput}
R> set.seed(12321)
R> fit.vm <- fit_incremental_angmix(model="vm", data = wind[, "angle"],
+                                   crit = "LOOIC",
+                                   start_ncomp = 1, max_ncomp = 10,
+                                   n.iter = 2e4,
+                                   n.chains = 3)
\end{Sinput}

The function stops at 3 components and determines the optimal number of components to be 2. After it stops, we extract the \code{angmcmc} object corresponding to the best model from its output, and inspect its LPD and parameter traces for convergence (omitted for brevity). 
\begin{Sinput}
R> fit.vm.best <- bestmodel(fit.vm)
R> lpdtrace(fit.vm.best)
R> paramtrace(fit.vm.best)
\end{Sinput}



We first use \code{fix_label} to undo label switching, and then assess goodness of fit through density curves fitted using MAP and posterior mean estimation:  
\begin{Sinput}
R> fit.vm.best <- fix_label(fit.vm.best)
R> lattice::densityplot(fit.vm.best, fn = "MODE")
R> lattice::densityplot(fit.vm.best, fn = mean)
\end{Sinput}
The plots are displayed in Figures~\ref{vmden_raw} and \ref{vmden}, which show noticeable similarity. 
\begin{figure}[!htpb]
	\centering
	\subcaptionbox{Approximate MAP estimate from MCMC samples. \label{vmden_raw}}%
	[.485\linewidth]{\includegraphics[height=3.2in, width = 3.2in]{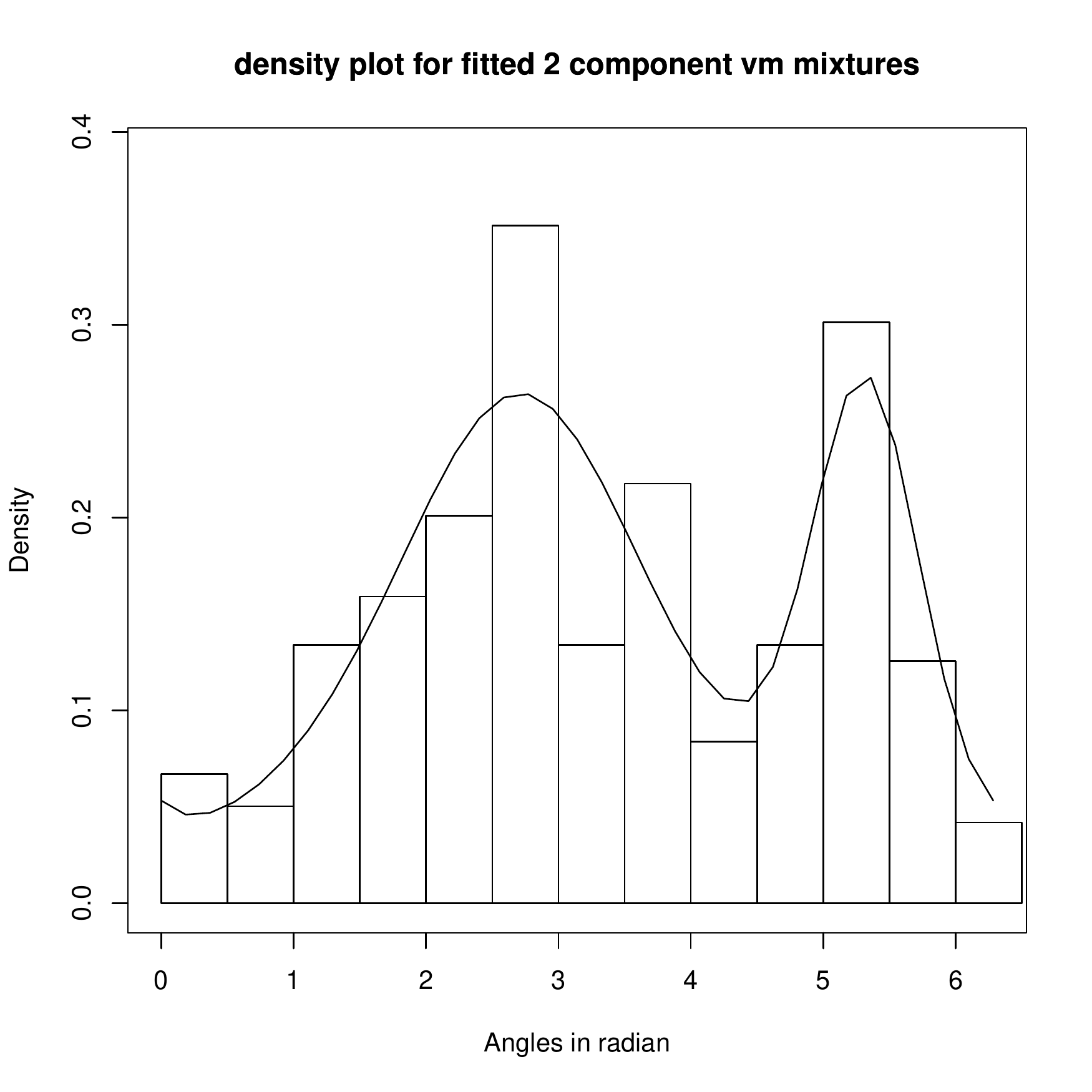}}
	\hfill
	\subcaptionbox{MCMC posterior mean after resolving label switching. \label{vmden}}%
	[.485\linewidth]{\includegraphics[height=3.2in, width = 3.2in]{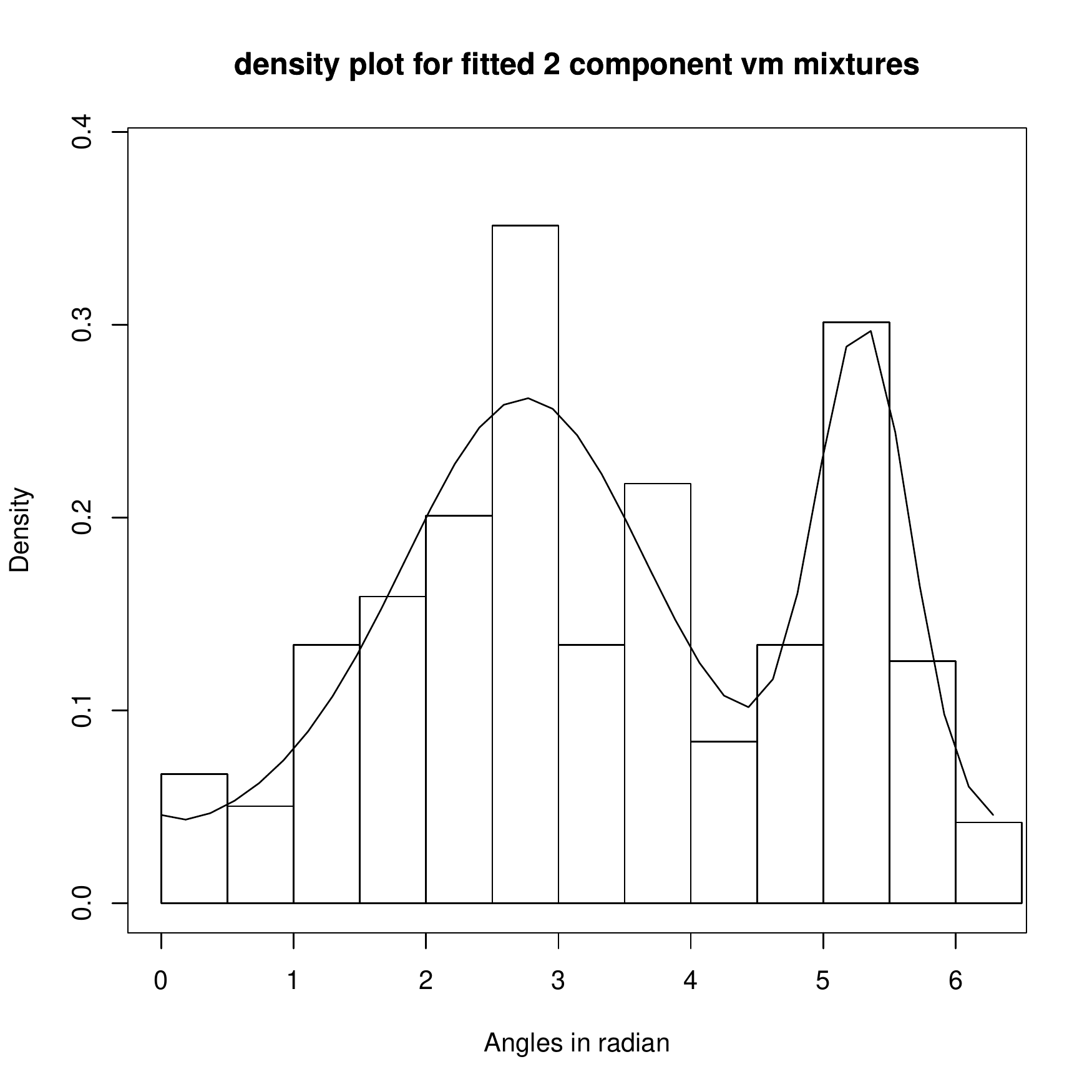}}
	\caption{Density curves for fitted 2 component \code{vm} mixture model with parameters estimated via MCMC-based MAP estimation (left) and  posterior mean estimation (right).}
	\label{vmdenplots}
\end{figure}

Finally, we compute MCMC posterior mean and associated 95\% credible interval using S3 function \code{summary}:
\begin{Sinput}
R> summary(fit.vm.best)
\end{Sinput}
\begin{Soutput}
                       1                 2
pmix   0.24 (0.13, 0.43) 0.76 (0.57, 0.87)
kappa 7.81 (1.34, 22.40) 1.01 (0.59, 1.71)
mu     5.30 (5.08, 5.49) 2.75 (2.49, 3.04)
\end{Soutput}

%
%

\paragraph{Fitting the wnorm mixture model.}
Next, we do similar exercises with \code{wnorm} model. The following are the \proglang{R} codes used.
\begin{Sinput}
R> set.seed(12321)
R> fit.wnorm <- fit_incremental_angmix(model="wnorm", data = wind[, "angle"],
+                                      crit = "LOOIC",
+                                      start_ncomp = 1, max_ncomp = 10,
+                                      n.iter = 2e4,
+                                      n.chains = 3)
R> fit.wnorm.best <- bestmodel(fit.wnorm)
R> lpdtrace(fit.wnorm.best)
R> paramtrace(fit.wnorm.best)
R> fit.wnorm.best <- fix_label(fit.wnorm.best)
R> lattice::densityplot(fit.wnorm.best, fn = "MODE")
R> lattice::densityplot(fit.wnorm.best, fn = mean)
\end{Sinput}

Similar to the \code{"vm"} case, here also the function stops at 3 components and determines the optimal number of components to be 2.  The LPD and parameter trace plots are  omitted for brevity. The density curves fitted using MAP and poterior mean estimates are shown in Figures~\ref{wnormden_raw} and \ref{wnormden} respectively, which are  noticeably similar. They are also  broadly similar to the plots associated with the fitted \code{vm} mixture densities shown in Figure~\ref{vmdenplots}.

\begin{figure}[!htpb]
	\centering
	\subcaptionbox{Approximate MAP estimate from MCMC samples. \label{wnormden_raw}}%
	[.485\linewidth]{\includegraphics[height=3.2in, width = 3.2in]{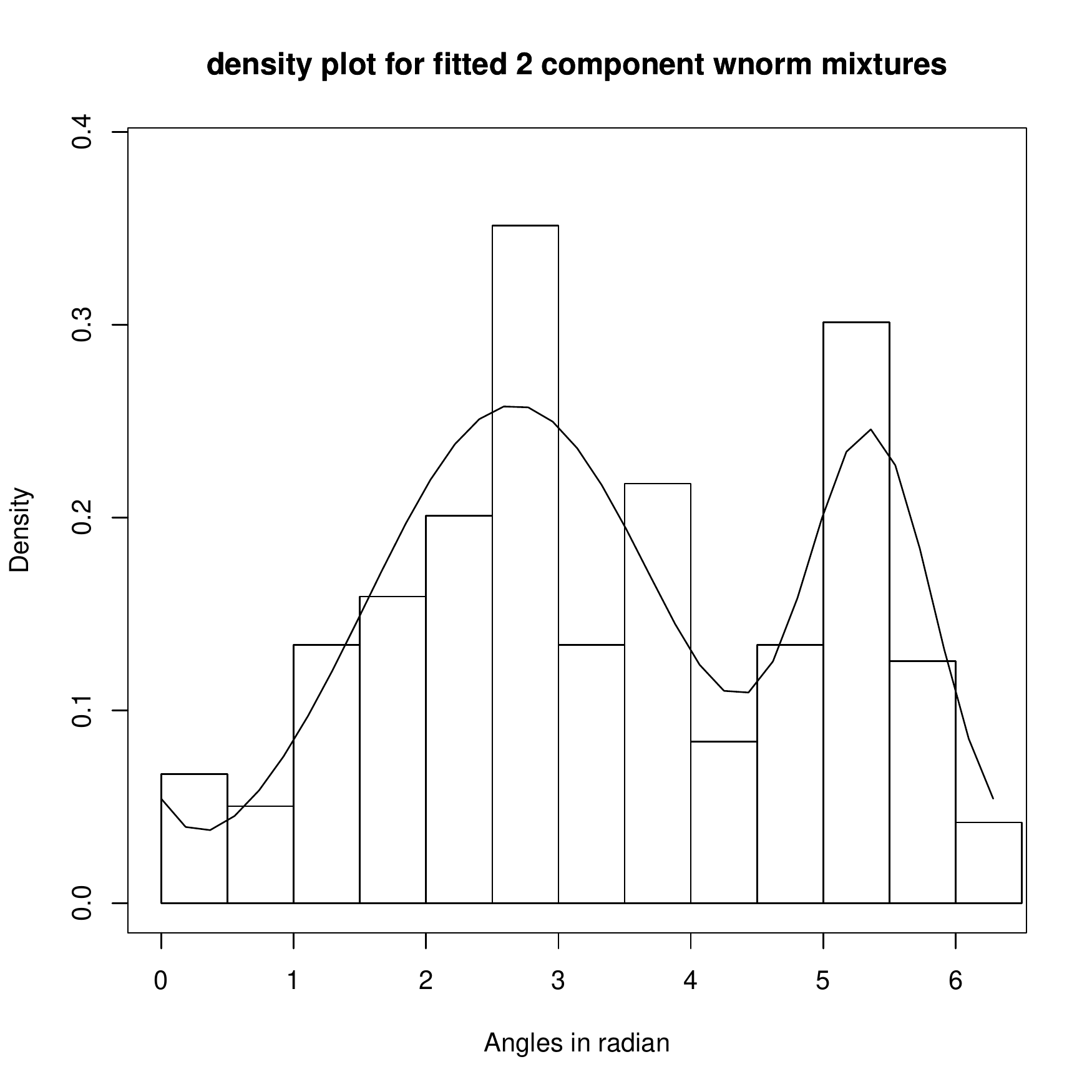}}
	\hfill
	\subcaptionbox{Estimated posterior mean after resolving label switching. \label{wnormden}}%
	[.485\linewidth]{\includegraphics[height=3.2in, width = 3.2in]{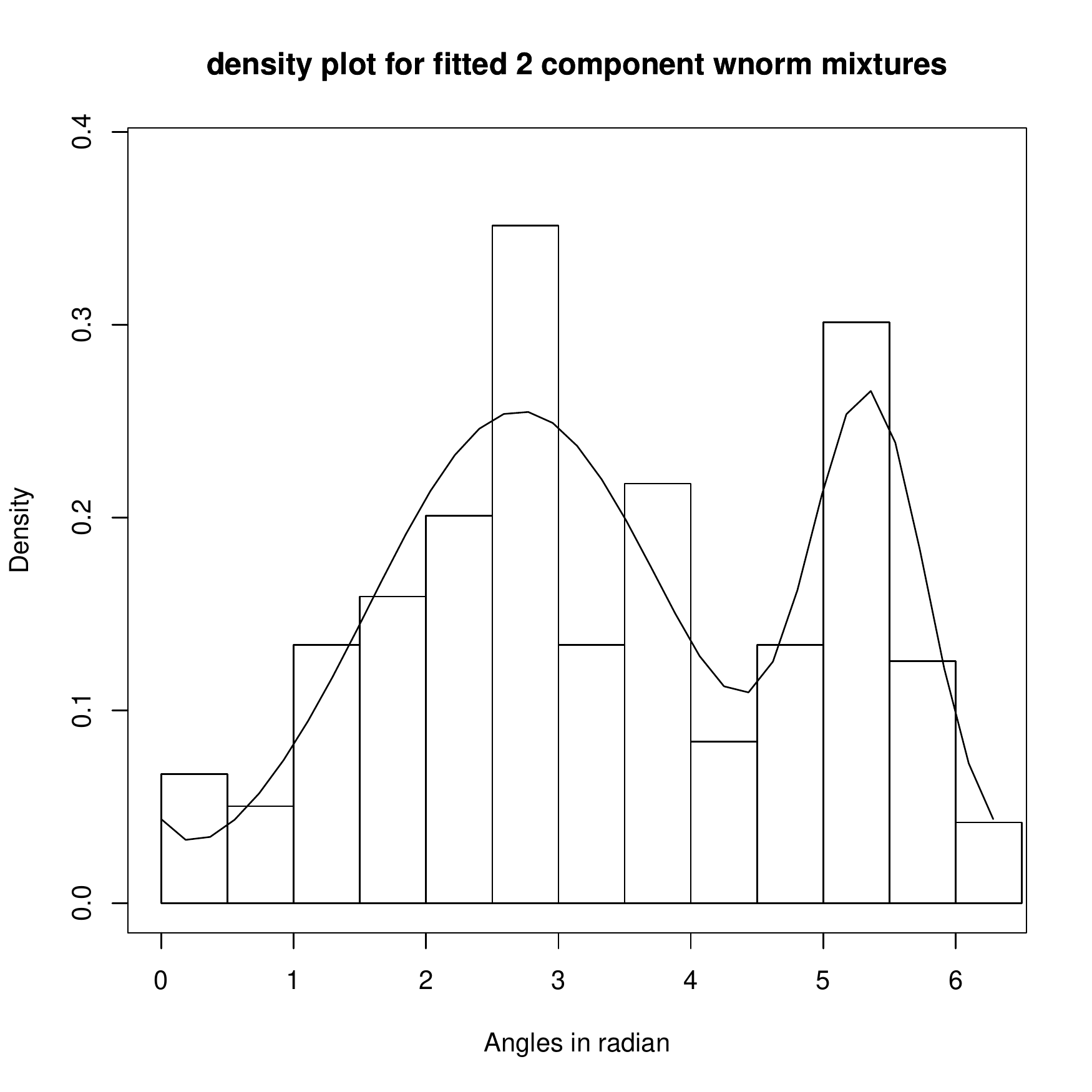}}
	\caption{Density curves for fitted 2 component \code{wnorm} mixture model with parameters estimated via MCMC-based MAP estimation (left) and  posterior mean estimation (right).}
	\label{wnormdenplots}
\end{figure}

Finally we compute the MCMC posterior mean and 95\% credible interval using the S3 function \code{summary}.
\begin{Sinput}
R> summary(fit.wnorm.best)
\end{Sinput}
\begin{Soutput}
                       1                 2
pmix   0.28 (0.16, 0.43) 0.72 (0.57, 0.84)
kappa 5.09 (1.16, 14.21) 0.78 (0.43, 1.35)
mu     5.35 (5.12, 5.53) 2.71 (2.45, 3.00)
\end{Soutput} 

%

\paragraph{Comparison between the two models.}
Similar to the bivariate case, we  compare the fitted \code{vm} and \code{wnorm} mixture models using their LOOIC values. We first extract the LOOICs using the convenience function \code{bestcriterion}:
\begin{Sinput}
R> vm.2.looic <- bestcriterion(fit.vm)
R> wnorm.2.looic <- bestcriterion(fit.wnorm)
\end{Sinput}
Then we compare the two models based on their estimated expected log predictive densities, by using \code{loo::compare} on \code{vm.2.looic} and \code{wnorm.2.looic}:
\begin{Sinput}
R> loo::compare(wnorm.2.looic, vm.2.looic)
\end{Sinput}
\begin{Soutput}
elpd_diff        se 
      1.0       1.0 
\end{Soutput}
Clearly an approximate 95\% credible interval estimate for the elpd difference, obtained by \code{elpd_diff} $\pm$ 2 \code{se}, contains zero. This implies that the fitted \code{vm} model and the fitted \code{wnorm} model do not have a significant difference in terms of their goodness of fit to these data.
%
%
%
%


\section{Concluding remarks and future work}

Angular data, both univariate and bivariate, arise naturally in a variety of modern scientific problems, and their analyses require appropriate use of rigorous statistical tools and distributions specifically developed for such data. The lack of comprehensive software implementing such methods (in \proglang{R} or otherwise) has hindered their applicability in practice -- especially for bivariate angular models and mixtures thereof.

The package \bambi is our contribution to this area, providing a platform that implements a set of formal statistical tools and methods for analyzing such data, and is readily accessible to  practitioners. 
There are various directions in which the software could be extended in future releases. Some possible features under consideration include the following.
\begin{itemize}
\item Implementation of additional angular distributions, such as wrapped Cauchy.
\item Additional methods of density evaluation and random simulation from fitted models on the basis of MCMC samples. 
\item Visualizations of bivariate angles with toroidal plots.
\item Use of parallel tempering or related methods during MCMC simulations for faster exploration of the posterior density.
\item Proper handling of overfitting heterogeneity that takes place in high dimensional mixture models when some of the component specific parameters in two different components are identical. \cite{fruhwirth:2011} suggests the use of sparse priors for the component specific location parameters to deal with this problem. 
\end{itemize}

\section{Acknowledgment}
The authors thank the two anonymous referees who reviewed this work, and provided valuable, thorough and constructive suggestions on both the package and the manuscript. The authors are also thankful to Anahita Nodehi from Tarbiat Modares University, Iran, for notifying them of a few typos and providing helpful suggestions on the R package documentation.

\bibliography{all_citations}

\begin{appendices}

\section{The normalizing constant for von Mises cosine density} \label{append_c_vmc_proof}

\begin{prop} \label{prop_const}
The normalizing constant for the density (\ref{vmc}) is given by	
\[
C_c(\kappa_1, \kappa_2, \kappa_3) = \left[(2\pi)^2 \left\lbrace I_0(\kappa_1)I_0(\kappa_2)I_0(\kappa_3) + 2 \sum_{n=1}^\infty  I_n(\kappa_1) I_n(\kappa_2)  I_n(\kappa_3)  \right\rbrace \right]^{-1}	
\]
\end{prop}

\begin{proof}
Without loss of generality, we first assume that the mean parameters in the density (\ref{vmc}) are all zero, i.e., $\mu_1 = \mu_2 = 0$. Therefore, our objective boils down to evaluate the integral
\begin{align} \label{int_vmcos}
C_c(\kappa_1, \kappa_2, \kappa_3)^{-1} =  \mathcal{I} = \int_{0}^{2\pi} \int_{0}^{2\pi} \exp(\kappa_1 \cos x + \kappa_2 \cos y + \kappa_3 \cos (x-y))\: dx \: dy.
\end{align}

\noindent Now from equation 9.6.34 of \cite{abramowitz:1964}, it follows that
\begin{align*}
\exp(\kappa_1 \cos x)  &= I_0(\kappa_1) + 2 \sum_{l=1}^\infty I_l(\kappa_1) \cos(lx) \\
\exp(\kappa_2 \cos y)  &= I_0(\kappa_2) + 2 \sum_{m=1}^\infty I_m(\kappa_2) \cos(my) \\
\text{and } \exp(\kappa_3 \cos (x-y))  &= I_0(\kappa_3) + 2 \sum_{n=1}^\infty I_n(\kappa_3) \cos(n(x-y)).
\end{align*}

\noindent Therefore, the integrand in (\ref{int_vmcos}) can be written as
\begin{align*}
& I_0(\kappa_1)I_0(\kappa_2)I_0(\kappa_3) + 2 \{I_0(\kappa_2) + I_0(\kappa_3)\} \sum_{l=1}^\infty I_l(\kappa_1) \cos(lx) \\
& \quad + 2 \{I_0(\kappa_3) + I_0(\kappa_1)\} \sum_{m=1}^\infty I_m(\kappa_2) \cos(my) \\
& \quad + 2 \{I_0(\kappa_1) + I_0(\kappa_2)\} \sum_{n=1}^\infty I_n(\kappa_3) \cos(n(x-y)) \\
& \quad + 8 \sum_{l=1}^\infty \sum_{m=1}^\infty \sum_{n=1}^\infty I_l(\kappa_1)  I_m(\kappa_2) I_n(\kappa_3)  \cos(lx)  \cos(my)  \cos(n(x-y)). \numbereqn \label{integrad_1}
\end{align*}

\noindent Note that for any positive integer $q$,
\[
\int_{0}^{2\pi} \cos(qz) \: dz = \int_{0}^{2\pi} \sin(qz) \: dz = 0
\]
which implies, for a positive integer $n$,
\begin{align*}
\int_{0}^{2\pi} \int_{0}^{2\pi} \cos(n(x-y)) \: dx \: dy &= \int_{0}^{2\pi} \cos(nx) \: dx \int_{0}^{2\pi} \cos(ny) \: dy \\
&\quad + \int_{0}^{2\pi} \sin(nx) \: dx \int_{0}^{2\pi} \sin(ny) \: dy = 0.
\end{align*}
(Equality of the double and the iterative integrals are ensured by the Fubini theorem, which is applicable as the integrands and the range of integrals are all finite.)

\noindent Thus the (double) integrals of the second, third and fourth terms in (\ref{integrad_1}) are all zero. Hence,
\begin{align*}
\mathcal{I} &= (2\pi)^2  I_0(\kappa_1)I_0(\kappa_2)I_0(\kappa_3)  \\
&\quad + 8 \int_{0}^{2\pi} \int_{0}^{2\pi} \sum_{l=1}^\infty \sum_{m=1}^\infty \sum_{n=1}^\infty   I_l(\kappa_1) I_m(\kappa_2)  I_n(\kappa_3)  \cos(lx)  \cos(my)  \cos(n(x-y)) \:dx \:dy. \numbereqn \label{int_vmcos_2}
\end{align*}

\noindent Now, for the second term in (\ref{int_vmcos_2}), first note that
\begin{align*}
& \quad \quad \int_{0}^{2\pi} \int_{0}^{2\pi} \sum_{l=1}^\infty \sum_{m=1}^\infty \sum_{n=1}^\infty \left|I_l(\kappa_1)  I_m(\kappa_2) I_n(\kappa_3)  \cos(lx)  \cos(my)  \cos(n(x-y)) \right| \: dx \: dy \\
& \leq \int_{0}^{2\pi} \int_{0}^{2\pi} \sum_{l=1}^\infty \sum_{m=1}^\infty \sum_{n=1}^\infty I_l(\kappa_1)  I_m(\kappa_2) I_n(|\kappa_3|) \: dx \: dy \\
&=  \sum_{l=1}^\infty \sum_{m=1}^\infty \sum_{n=1}^\infty \int_{0}^{2\pi} \int_{0}^{2\pi} I_l(\kappa_1)  I_m(\kappa_2) I_n(|\kappa_3|) \: dx \: dy \quad  \text{(by Fubini-Tonelli)} \\
&= (2\pi)^2 \left( \sum_{l=1}^\infty I_l(\kappa_1)\right) \left(\sum_{m=1}^\infty  I_m(\kappa_2) \right) \left( \sum_{n=1}^\infty   I_n(|\kappa_3|) \right) \\
&< \infty
\end{align*}
where the equality in the third line follows from the Fubini-Tonelli theorem for non-negative integrands. Therefore, the Fubini theorem for general integrands can be applied to ensure interchangeability of the sums and the integrals in the second term in (\ref{int_vmcos_2}). In particular, one can write
\begin{align*}
& \int_{0}^{2\pi} \int_{0}^{2\pi} \sum_{l=1}^\infty \sum_{m=1}^\infty \sum_{n=1}^\infty   I_l(\kappa_1) I_m(\kappa_2)  I_n(\kappa_3)  \cos(lx)  \cos(my)  \cos(n(x-y)) \:dx \:dy \\
&= \sum_{l=1}^\infty \sum_{m=1}^\infty \sum_{n=1}^\infty   I_l(\kappa_1) I_m(\kappa_2)  I_n(\kappa_3) \int_{0}^{2\pi} \int_{0}^{2\pi} \cos(lx)  \cos(my)  \cos(n(x-y)) \:dx \:dy. \numbereqn \label{intsum}
\end{align*}

Now, for any positive integers $l, m, n$,
\begin{align*}
\cos(lx)  \cos(my)  \cos(n(x-y)) &= \cos(lx) \cos(nx)  \cos(my)  \cos(ny) \\
& \quad + \cos(lx) \sin(nx)  \cos(my) \sin(ny).
\end{align*}

\noindent Observe that for any two positive integers $p$ and $q$,
\begin{align*}
\int_{0}^{2\pi} \cos(pz) \cos(qz) \: dz = \pi \one_{\{p=q\}} \text{ and } \int_{0}^{2\pi} \cos(pz) \sin(qz) \: dz = 0.
\end{align*}

\noindent Therefore, for any positive integers $l, m, n$,
\[
\int_{0}^{2\pi}\int_{0}^{2\pi} \cos(lx) \cos(nx)  \cos(my)  \cos(ny) \: dx \:dy = \pi \one_{\{l = n\}} \pi \one_{\{m = n\}} = \pi^2 \one_{\{l = m = n\}}
\]
and
\[
\int_{0}^{2\pi}\int_{0}^{2\pi} \cos(lx) \sin(nx)  \cos(my) \sin(ny) \: dx \:dy = 0.
\]
which implies,
\begin{align} \label{lastterm}
\int_{0}^{2\pi} \int_{0}^{2\pi} \cos(lx)  \cos(my)  \cos(n(x-y)) \:dx \:dy  = \pi^2 \one_{\{l = m = n\}}.
\end{align}

Therefore, combining (\ref{int_vmcos_2}), (\ref{intsum}) and (\ref{lastterm}), we get
\begin{align*}
\mathcal{I} &= (2\pi)^2  I_0(\kappa_1)I_0(\kappa_2)I_0(\kappa_3) + 8 \pi^2 \sum_{l=1}^\infty \sum_{m=1}^\infty \sum_{n=1}^\infty   I_l(\kappa_1)  I_n(\kappa_3) I_m(\kappa_2)  \one_{\{l = m = n\}} \\
&= (2\pi)^2 \left\lbrace I_0(\kappa_1)I_0(\kappa_2)I_0(\kappa_3) + 2 \sum_{n=1}^\infty  I_n(\kappa_1) I_n(\kappa_2)  I_n(\kappa_3)  \right\rbrace
\end{align*}
This completes the proof.	
\end{proof}

\section{Circular variance and correlation coefficients} \label{appen_vmsin_vmocs_varcor}

\subsection{von Mises sine model} \label{appen_vmsin_varcor}

	Let $(\psi_1, \psi_2) \sim \vms(\mu_1, \mu_2, \kappa_1, \kappa_2, \kappa_3)$. Then 
		\begin{enumerate}
			\item the Fisher-Lee circular correlation coefficient (\ref{rho_fl_defn}) between $\psi_1$ and $\psi_2$ is given by 
	        \begin{equation*}
			\rho_{\fl} (\psi_1, \psi_2) 
		     = \frac{\left(\frac{1}{\bar C_s} \frac{\partial \bar C_s}{\partial \kappa_3} \right) \left(\frac{1}{\bar C_s} \frac{\partial^2 \bar C_s}{\partial \kappa_1 \partial \kappa_2} \right)}{\sqrt{\left(\frac{1}{\bar C_s} \frac{\partial^2 \bar C_s}{\partial \kappa_1^2} \right) \left(1 - \frac{1}{\bar C_s} \frac{\partial^2 \bar C_s}{\partial \kappa_1^2} \right) \left(\frac{1}{\bar C_s} \frac{\partial^2 \bar C_s}{\partial \kappa_2^2} \right) \left(1 - \frac{1}{\bar C_s} \frac{\partial^2 \bar C_s}{\partial \kappa_2^2} \right)}}.   
			\end{equation*}
			
			\item the Jammalamadaka-Sarma circular correlation coefficient (\ref{rho_js_defn}) between $\Theta$ and $\Phi$ is given by 
			\begin{equation*}
			\rho_{\js} (\psi_1, \psi_2) 
			= \frac{\frac{1}{\bar C_s} \frac{\partial \bar C_s}{\partial \kappa_3} }{\sqrt{\left(1 - \frac{1}{\bar C_s} \frac{\partial^2 \bar C_s}{\partial \kappa_1^2} \right)  \left(1 - \frac{1}{\bar C_s} \frac{\partial^2 \bar C_s}{\partial \kappa_2^2} \right)}}.   
			\end{equation*}
			
			\item the circular variance for $\psi_i$, $i = 1, 2$ is given by
			\begin{equation*}
			\var(\psi_i) = 1 - \frac{1}{\bar C_s} \frac{\partial \bar C_s}{\partial \kappa_i}.
			\end{equation*}
		\end{enumerate}	
	
	   Here $\bar C_s = 1/C_s$, where $C_s$ is the normalizing constant of the von Mises sine distribution as defined in \eqref{c_vms}. Infinite series expressions for partial derivatives of $\bar C_s$ constant are provided as follows.
	       \begin{align*}
			\frac{\partial \bar C_s}{\partial \kappa_1} &= 4 \pi^2  \sum_{m=0}^{\infty} \binom{2m}{m} \left(\frac{\kappa_3^2}{4\kappa_1 \kappa_2}\right)^m I_{m+1}(\kappa_1) I_m(\kappa_2)  \\
			\frac{\partial \bar C_s}{\partial \kappa_2} &= 4 \pi^2  \sum_{m=0}^{\infty} \binom{2m}{m} \left(\frac{\kappa_3^2}{4\kappa_1 \kappa_2}\right)^m I_{m}(\kappa_1) I_{m+1}(\kappa_2)  \\
			\frac{\partial \bar C_s}{\partial \kappa_3} &=  8 \pi^2  \sum_{m=1}^{\infty} m \binom{2m}{m} \frac{\kappa_3^{2m-1}}{(4\kappa_1 \kappa_2)^m} I_{m}(\kappa_1) I_{m}(\kappa_2) \label{del_C_lambda_expr} \\
			\frac{\partial^2 \bar C_s}{\partial \kappa_1^2} &= 4 \pi^2  \sum_{m=0}^{\infty} \binom{2m}{m} \left(\frac{\kappa_3^2}{4\kappa_1 \kappa_2}\right)^m     \\
			&\qquad \qquad \left(\frac{I_{m+1}(\kappa_1)}{\kappa_1} + I_{m+2}(\kappa_1)\right) I_m(\kappa_2)  \\
			\frac{\partial^2 \bar C_s}{\partial \kappa_2^2} &= 4 \pi^2  \sum_{m=0}^{\infty} \binom{2m}{m} \left(\frac{\kappa_3^2}{4\kappa_1 \kappa_2}\right)^m   \\
			& \qquad \qquad  I_m(\kappa_1) \left(\frac{I_{m+1}(\kappa_2)}{\kappa_2} + I_{m+2}(\kappa_2)\right) \\
			\frac{\partial^2 \bar C_s}{\partial \kappa_1 \: \partial \kappa_2} &= 4 \pi^2  \sum_{m=0}^{\infty} \binom{2m}{m} \left(\frac{\kappa_3^2}{4\kappa_1 \kappa_2}\right)^m  I_{m+1}(\kappa_1) I_{m+1}(\kappa_2) 
			\end{align*}   
   
\subsection{von Mises cosine model} \label{appen_vmcos_varcor}

	Let $(\psi_1, \psi_2) \sim \vmc(\mu_1, \mu_2, \kappa_1, \kappa_2, \kappa_3)$. Then 
		\begin{enumerate}
			\item the Fisher-Lee circular correlation coefficient (\ref{rho_fl_defn}) between $\psi_1$ and $\psi_2$ is given by 
	        \begin{equation*}
			\rho_{\fl} (\psi_1, \psi_2) 
		     = \frac{\left(\frac{1}{\bar C_c} \left\{\frac{\partial \bar C_c}{\partial \kappa_3} -  \frac{\partial^2 \bar C_c}{\partial \kappa_1 \partial \kappa_2} \right\} \right) \left(\frac{1}{\bar C_c} \frac{\partial^2 \bar C_c}{\partial \kappa_1 \partial \kappa_2} \right)}{\sqrt{\left(\frac{1}{\bar C_c} \frac{\partial^2 \bar C_c}{\partial \kappa_1^2} \right) \left(1 - \frac{1}{\bar C_c} \frac{\partial^2 \bar C_c}{\partial \kappa_1^2} \right) \left(\frac{1}{\bar C_c} \frac{\partial^2 \bar C_c}{\partial \kappa_2^2} \right) \left(1 - \frac{1}{\bar C_c} \frac{\partial^2 \bar C_c}{\partial \kappa_2^2} \right)}}.   
			\end{equation*}
			
			\item the Jammalamadaka-Sarma circular correlation coefficient (\ref{rho_js_defn}) between $\Theta$ and $\Phi$ is given by 
			\begin{equation*}
			\rho_{\js} (\psi_1, \psi_2) 
			= \frac{\frac{1}{\bar C_c} \left\{\frac{\partial \bar C_c}{\partial \kappa_3} -  \frac{\partial^2 \bar C_c}{\partial \kappa_1 \partial \kappa_2} \right\} }{\sqrt{\left(1 - \frac{1}{\bar C_c} \frac{\partial^2 \bar C_c}{\partial \kappa_1^2} \right)  \left(1 - \frac{1}{\bar C_c} \frac{\partial^2 \bar C_c}{\partial \kappa_2^2} \right)}}.   
			\end{equation*}
			
			\item the circular variance for $\psi_i$, $i = 1, 2$ is given by
			\begin{equation*}
			\var(\psi_i) = 1 - \frac{1}{\bar C_c} \frac{\partial \bar C_c}{\partial \kappa_i} .
			\end{equation*}
		\end{enumerate}

	   Here $\bar C_c = 1/C_c$ is the reciprocal of the von Mises cosine normalizing constant, as given in \eqref{c_vmc}. Infinite series expressions for partial derivatives of $\bar C_c$ are given as follows.
	   
	   		\begin{align*}
	   		\frac{\partial \bar C_c}{\partial \kappa_1} =  4\pi^2  & \left\{ I_1(\kappa_1) I_0(\kappa_2) I_0(\kappa_3)  +  \right.   \\
	   		& \qquad  \sum_{m = 1} ^\infty  \left.    I_m(\kappa_2) I_m(\kappa_3) \left[ I_{m+1}(\kappa_1) + I_{m-1}(\kappa_1) \right] \right\}   \\
	   		\frac{\partial \bar C_c}{\partial \kappa_2} = 4\pi^2 & \left\{  I_0(\kappa_1) I_1(\kappa_2) I_0(\kappa_3)  +  \right.   \\
	   		& \qquad  \sum_{m = 1} ^\infty  \left.    I_m(\kappa_1) I_m(\kappa_3) \left[ I_{m+1}(\kappa_2) + I_{m-1}(\kappa_2) \right] \right\}  \\
	   		\frac{\partial \bar C_c}{\partial \kappa_3} = 4\pi^2 & \left\{ I_0(\kappa_1) I_0(\kappa_2) I_1(\kappa_3)  +  \right.   \\ 
	   		&   \qquad \sum_{m = 1} ^\infty \left. I_m(\kappa_1) I_m(\kappa_2) \left[ I_{m+1}(\kappa_3) + I_{m-1}(\kappa_3) \right] \right \rbrace. \\
	   		\frac{\partial^2 \bar C_c}{\partial \kappa_1^2} = 2\pi^2 & \left\{  I_0(\kappa_2) I_0(\kappa_3)[I_0(\kappa_1) + I_2(\kappa_1)] + \right.   \\ 
	   		& \qquad  \sum_{m = 1} ^\infty  \left.   I_{m}(\kappa_2) I_{m}(\kappa_3) [I_{m-2}(\kappa_1) + 2I_{m}(\kappa_1) + I_{m+2}(\kappa_1)]   \right\} \\ 
	   		\frac{\partial^2 \bar C_c}{\partial \kappa_2^2} = 2\pi^2 & \left\{  I_0(\kappa_1) I_0(\kappa_3)[I_0(\kappa_2) + I_2(\kappa_2)] + \right.   \\ 
	   		& \qquad  \sum_{m = 1} ^\infty  \left.   I_{m}(\kappa_1) I_{m}(\kappa_3) [I_{m-2}(\kappa_2) + 2I_{m}(\kappa_2) + I_{m+2}(\kappa_2)]   \right\} \\
	   		\frac{\partial^2 \bar C_c}{\partial \kappa_1 \partial \kappa_2} = 2\pi^2 & \left\{ 2 I_1(\kappa_1) I_1(\kappa_2) I_0(\kappa_3) + \right.   \\ 
	   		&  \sum_{m = 1} ^\infty  \left.  I_m(\kappa_3) \left[ I_{m+1}(\kappa_1) + I_{m-1}(\kappa_1) \right] \left[ I_{m+1}(\kappa_2) + I_{m-1}(\kappa_2) \right]\right\} 
	   		\end{align*}

\section{Gradients} \label{grads}
For notational simplicity we shall omit the subscripts $i$ and $j$. Note that, in the sequel, $\thetab$ stands for the parameter vector for one generic component and not the entire parameter vector of all components.

\subsection{Wrapped normal models}
\begin{enumerate}
\item \textit{Univariate case.} Here $\thetab^\top = (\kappa, \mu)$, and
\begin{align*}
\frac{\partial f_{\wn}(\psi| \thetab)}{\partial \kappa} &= \frac{1}{2 \kappa^{1/2} \sqrt{2\pi}} \sum_{\omega\in \Z} \exp \left[-\frac{\kappa}{2}(\psi - \mu - 2\pi \omega)^2 \right] \left[  { 1 - \kappa (\psi - \mu - 2\pi\omega)^2} \right] \\
\frac{\partial f_{\wn}(\psi| \thetab)}{\partial \mu} &= \frac{\kappa^{3/2}}{\sqrt{2\pi}} \sum_{\omega\in \Z} \exp \left[-\frac{\kappa}{2}(\psi - \mu - 2\pi \omega)^2 \right] (\psi - \mu - 2\pi\omega).
\end{align*}

\item \textit{Bivariate case.} Here $\thetab^\top = (\kappa_1, \kappa_2, \kappa_3, \mu_1, \mu_2)$, $\psib^\top = (\psi_1, \psi_2)$ and
\begin{align*}
\frac{\partial f_{\bwn}(\psib| \thetab)}{\partial \kappa_1} &= \frac{1}{4\pi \sqrt{\kappa_{12.3}}} \: \sum_{(\omega_1, \omega_2) \in \Z^2}  E_{\omega_1, \omega_2} \left[ \kappa_2 - \kappa_{12.3} (\psi_1 - \mu_1 - 2\pi\omega_1)^2 \right] \\
\frac{\partial f_{\bwn}(\psib| \thetab)}{\partial \kappa_2} &= \frac{1}{4\pi \sqrt{\kappa_{12.3}}} \: \sum_{(\omega_1, \omega_2) \in \Z^2}  E_{\omega_1, \omega_2} \left[ \kappa_1 - \kappa_{12.3} (\psi_2 - \mu_2 - 2\pi\omega_2)^2 \right]  \\
\frac{\partial f_{\bwn}(\psib| \thetab)}{\partial \kappa_3} &= \frac{1}{2\pi \sqrt{\kappa_{12.3}}} \: \sum_{(\omega_1, \omega_2) \in \Z^2}  E_{\omega_1, \omega_2} \left[ \kappa_3 - \kappa_{12.3} (\psi_1 - \mu_1 - 2\pi\omega_1)(\psi_2 - \mu_2 - 2\pi\omega_2) \right]  \\
\frac{\partial f_{\bwn}(\psib| \thetab)}{\partial \mu_1} &= \frac{\sqrt{\kappa_{12.3}}}{2\pi } \: \sum_{(\omega_1, \omega_2) \in \Z^2}  E_{\omega_1, \omega_2} \left[ \kappa_1(\psi_1 - \mu_1 - 2\pi\omega_1)  + \kappa_3(\psi_2 - \mu_2 - 2\pi\omega_2)\right]   \\
\frac{\partial f_{\bwn}(\psib| \thetab)}{\partial \mu_2} &= \frac{\sqrt{\kappa_{12.3}}}{2\pi } \: \sum_{(\omega_1, \omega_2) \in \Z^2}  E_{\omega_1, \omega_2} \left[ \kappa_3(\psi_1 - \mu_1 - 2\pi\omega_1)  + \kappa_2(\psi_2 - \mu_2 - 2\pi\omega_2)\right]
\end{align*}
where
\begin{align*}
E_{\omega_1, \omega_2}  =   \exp  & \left[-\frac{1}{2} \left\lbrace  \kappa_1 (\psi_1 - \mu_1 - 2\pi \omega_1)^2 + \kappa_2 (\psi_2 - \mu_2 - 2\pi \omega_2)^2  \right. \right.\\
& \qquad \quad \left. \left. + 2 \kappa_3 (\psi_1 - \mu_1 - 2\pi \omega_1) (\psi_2 - \mu_2 - 2\pi \omega_2) \right\rbrace \right]
\end{align*}
and $\kappa_{12.3} = \kappa_1\kappa_2 - \kappa_3^2$.
\end{enumerate}

\subsection{von Mises models}
\begin{enumerate}
\item \textit{Univariate case.} Here $\thetab^\top = (\kappa, \mu)$ and
\begin{align*}
\frac{\partial \log f_{\vm}(\psi| \thetab)}{\partial \kappa} &= \cos(\psi - \mu) - \frac{I_1(\kappa)}{I_0(\kappa)} \\
\frac{\partial \log f_{\vm}(\psi| \thetab)}{\partial \mu} &= \kappa \sin(\psi - \mu).
\end{align*}
\item \textit{Bivariate sine model.} Here $\thetab^\top = (\kappa_1, \kappa_2, \kappa_3, \mu_1, \mu_2)$, $\psib^\top = (\psi_1, \psi_2)$ and
\begin{align*}
\frac{\partial \log f_{\vms}(\psib| \thetab)}{\partial \kappa_1} &= \cos(\psi_1 - \mu_1) - \frac{\partial \bar{C}_{s}(\kappa_1, \kappa_1, \kappa_3)/\partial \kappa_1}{\bar{C}_s(\kappa_1, \kappa_1, \kappa_3)} \\
\frac{\partial \log f_{\vms}(\psib| \thetab)}{\partial \kappa_2} &= \cos(\psi_2 - \mu_2) - \frac{\partial\bar{C}_{s}(\kappa_1, \kappa_1, \kappa_3)/\partial \kappa_2}{\bar{C}_s(\kappa_1, \kappa_1, \kappa_3)} \\
\frac{\partial \log f_{\vms}(\psib| \thetab)}{\partial \kappa_3} &= \sin(\psi_1 - \mu_1) \sin(\psi_2 - \mu_2) - \frac{\partial\bar{C}_{s}(\kappa_1, \kappa_1, \kappa_3)/\partial\kappa_3}{\bar{C}_s(\kappa_1, \kappa_1, \kappa_3)} \\
\frac{\partial \log f_{\vms}(\psib| \thetab)}{\partial \mu_1} &= \kappa_1 \sin(\psi_1 - \mu_1) - \kappa_3 \cos(\psi_1 - \mu_1) \sin(\psi_2 - \mu_2) \\
\frac{\partial \log f_{\vms}(\psib| \thetab)}{\partial \mu_2} &= \kappa_2 \sin(\psi_2 - \mu_2) - \kappa_3 \sin(\psi_1 - \mu_1) \cos(\psi_2 - \mu_2)
\end{align*}
where $\bar{C}_s(\kappa_1, \kappa_1, \kappa_3) = 1/C_s(\kappa_1, \kappa_1, \kappa_3)$ and expressions for the partial derivatives are provided in Appendix~\ref{appen_vmsin_varcor}.

\item \textit{Bivariate cosine model.} Here $\thetab^\top = (\kappa_1, \kappa_2, \kappa_3, \mu_1, \mu_2)$, $\psib^\top = (\psi_1, \psi_2)$ and
\begin{align*}
\frac{\partial \log f_{\vmc}(\psib| \thetab)}{\partial \kappa_1} &= \cos(\psi_1 - \mu_1) - \frac{\partial \bar{C}_{c}(\kappa_1, \kappa_1, \kappa_3)/\partial \kappa_1}{\bar{C}_c(\kappa_1, \kappa_1, \kappa_3)} \\
\frac{\partial \log f_{\vmc}(\psib| \thetab)}{\partial \kappa_2} &= \cos(\psi_2 - \mu_2) - \frac{\partial \bar{C}_{c}(\kappa_1, \kappa_1, \kappa_3)/\partial\kappa_2}{\bar{C}_c(\kappa_1, \kappa_1, \kappa_3)} \\
\frac{\partial \log f_{\vmc}(\psib| \thetab)}{\partial \kappa_3} &= \cos(\psi_1 - \mu_1 - \psi_2 + \mu_2) - \frac{\partial \bar{C}_{c}(\kappa_1, \kappa_1, \kappa_3)/\partial\kappa_3}{\bar{C}_c(\kappa_1, \kappa_1, \kappa_3)} \\
\frac{\partial \log f_{\vmc}(\psib| \thetab)}{\partial \mu_1} &= \kappa_1 \sin(\psi_1 - \mu_1) + \kappa_3 \sin(\psi_1 - \mu_1 - \psi_2 + \mu_2) \\
\frac{\partial \log f_{\vmc}(\psib| \thetab)}{\partial \mu_2} &= \kappa_2 \sin(\psi_2 - \mu_2) - \kappa_3 \sin(\psi_1 - \mu_1 - \psi_2 + \mu_2)
\end{align*}
where $\bar{C}_c(\kappa_1, \kappa_1, \kappa_3) = 1/C_c(\kappa_1, \kappa_1, \kappa_3)$ and infinite series expressions for the partial derviatives are provided in Appendix~\ref{appen_vmcos_varcor}.

\end{enumerate}

\section{Trace plots for 4 component vmsin}	

\label{vmsinparamtraces}

\begin{figure}[!htpb]
	\centering 
	\subcaptionbox{}%
	[.485\linewidth]{\includegraphics[height=2.6in, width = 2.6in]{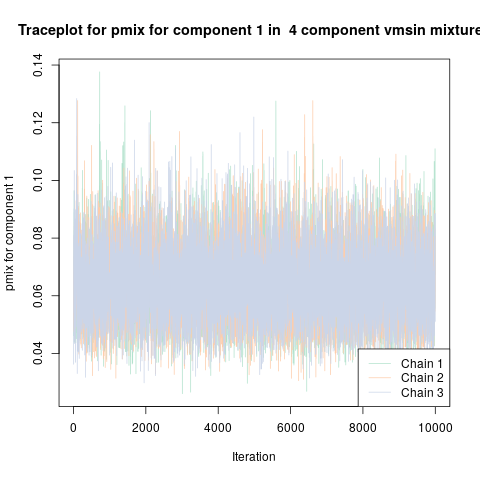}}
	\hfill
	\subcaptionbox{}%
	[.485\linewidth]{\includegraphics[height=2.6in, width = 2.6in]{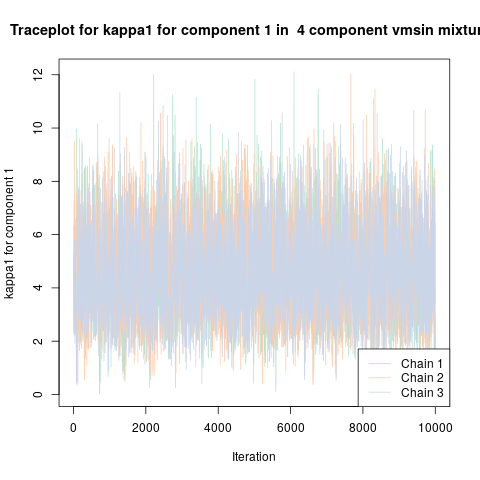}} \\
	\subcaptionbox{}%
	[.485\linewidth]{\includegraphics[height=2.6in, width = 2.6in]{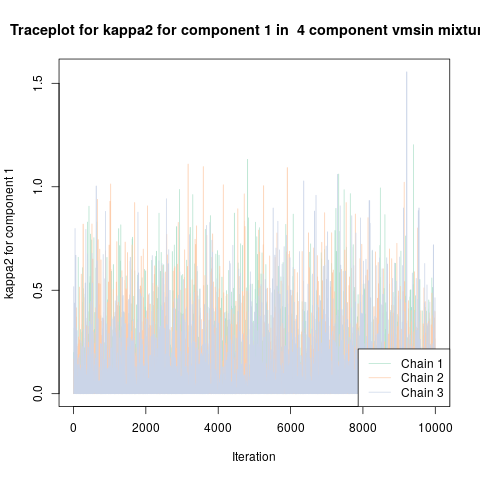}}
	\hfill
	\subcaptionbox{}%
	[.485\linewidth]{\includegraphics[height=2.6in, width = 2.6in]{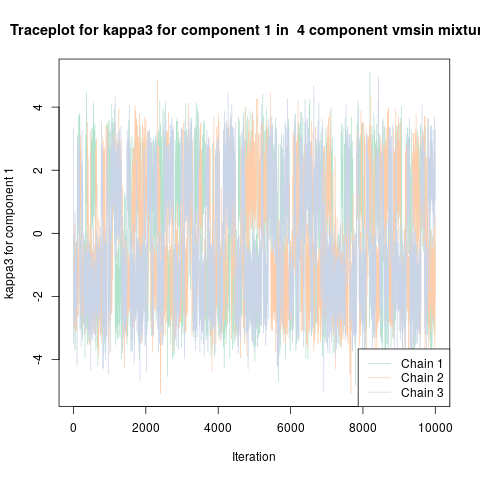}} \\
	\subcaptionbox{}
	[.485\linewidth]{\includegraphics[height=2.6in, width = 2.6in]{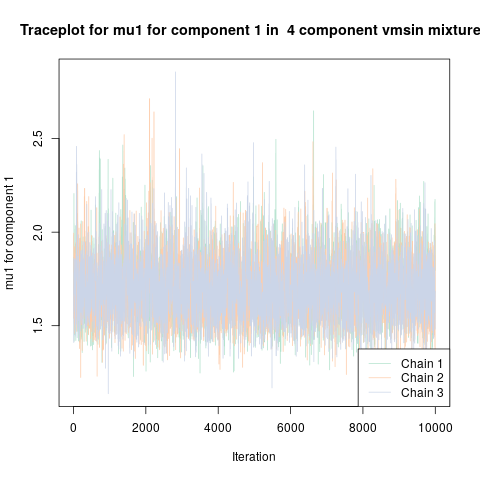}}
	\hfill
	\subcaptionbox{}%
	[.485\linewidth]{\includegraphics[height=2.6in, width = 2.6in]{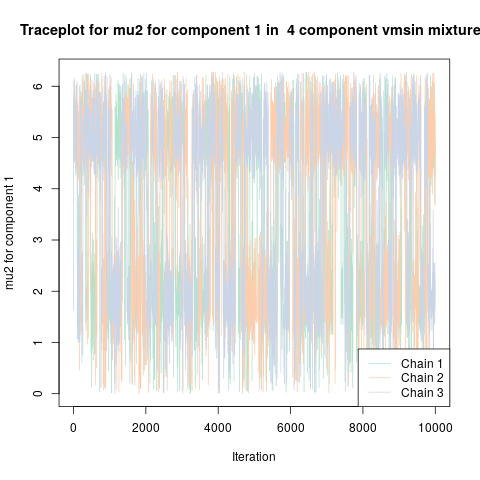}} 
	\caption{Trace plots for parameters in the first component for the Markov chain associated with the best fitted vmsin mixture model.}
	\label{vmsin_paramtrace_plots_comp1}
\end{figure}

\begin{figure}[!htpb]
	\centering 
	\subcaptionbox{}%
	[.485\linewidth]{\includegraphics[height=2.6in, width = 2.6in]{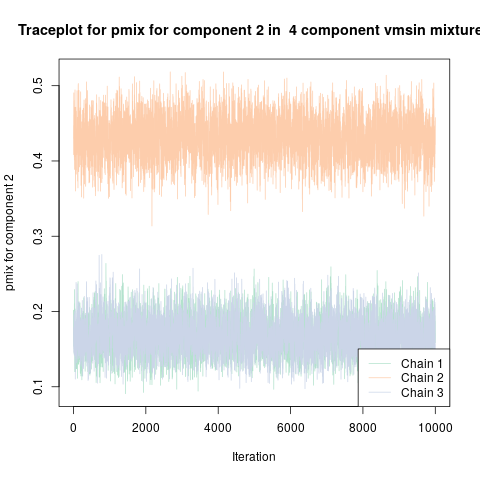}}
	\hfill
	\subcaptionbox{}%
	[.485\linewidth]{\includegraphics[height=2.6in, width = 2.6in]{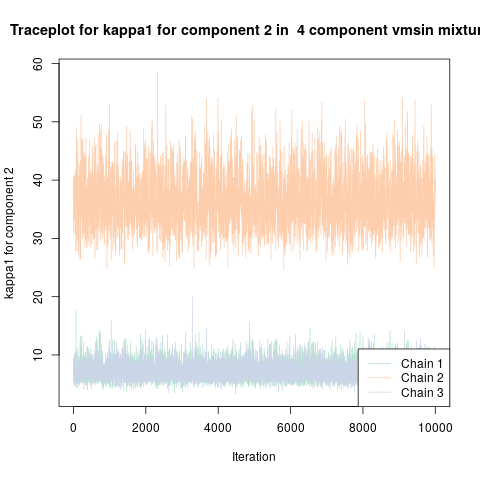}} \\
	\subcaptionbox{}%
	[.485\linewidth]{\includegraphics[height=2.6in, width = 2.6in]{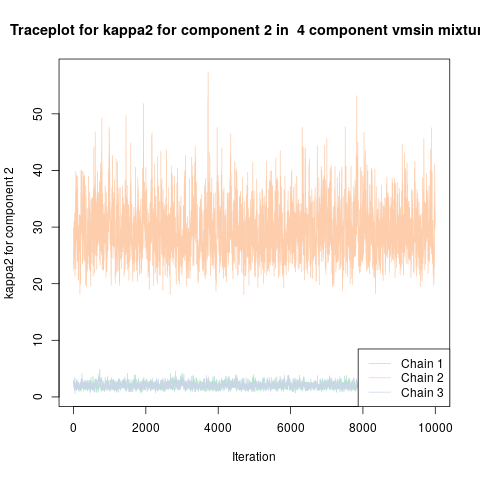}}
	\hfill
	\subcaptionbox{}%
	[.485\linewidth]{\includegraphics[height=2.6in, width = 2.6in]{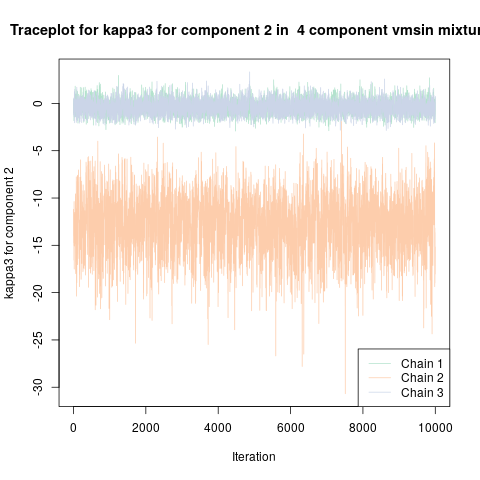}} \\
	\subcaptionbox{}
	[.485\linewidth]{\includegraphics[height=2.6in, width = 2.6in]{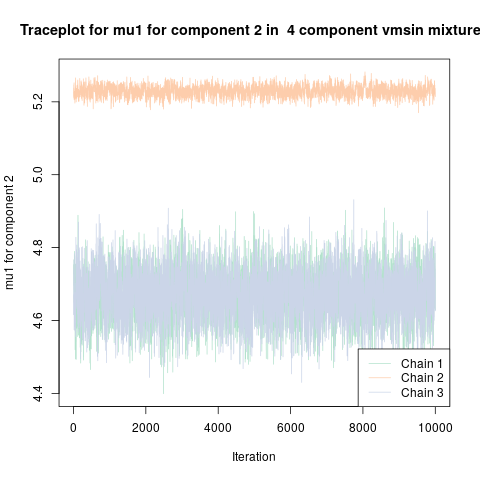}}
	\hfill
	\subcaptionbox{}%
	[.485\linewidth]{\includegraphics[height=2.6in, width = 2.6in]{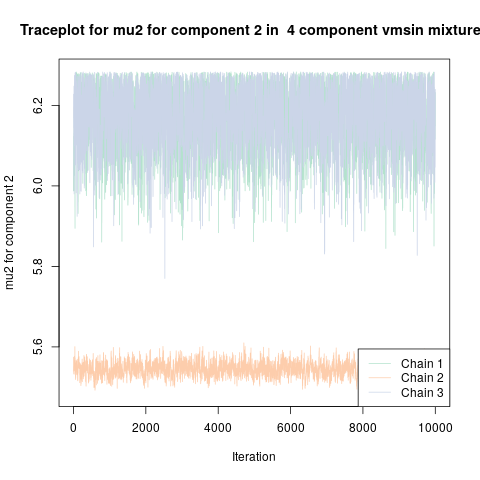}} 
	\caption{Trace plots for parameters in the second component for the Markov chain associated with the best fitted vmsin mixture model.}
	\label{vmsin_paramtrace_plots_comp2}
\end{figure}

\begin{figure}[!htpb]
	\centering 
	\subcaptionbox{}%
	[.485\linewidth]{\includegraphics[height=2.6in, width = 2.6in]{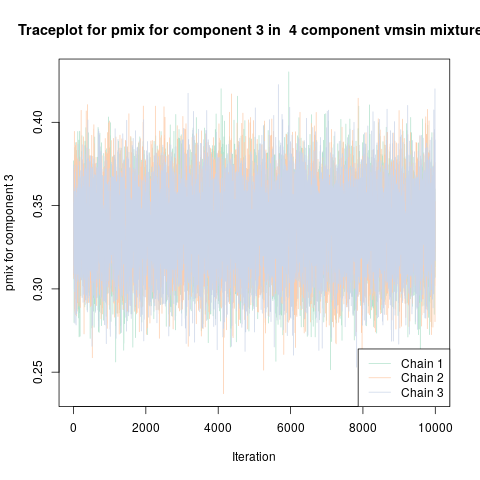}}
	\hfill
	\subcaptionbox{}%
	[.485\linewidth]{\includegraphics[height=2.6in, width = 2.6in]{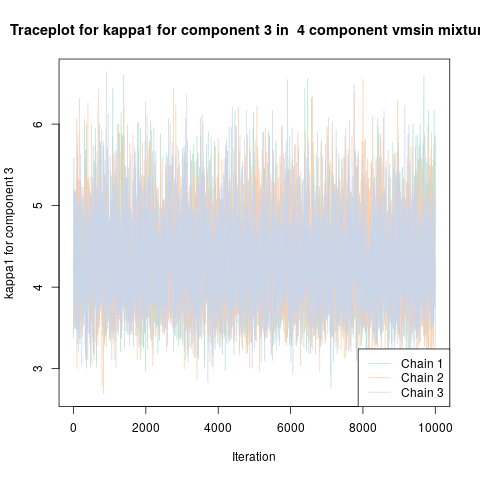}} \\
	\subcaptionbox{}%
	[.485\linewidth]{\includegraphics[height=2.6in, width = 2.6in]{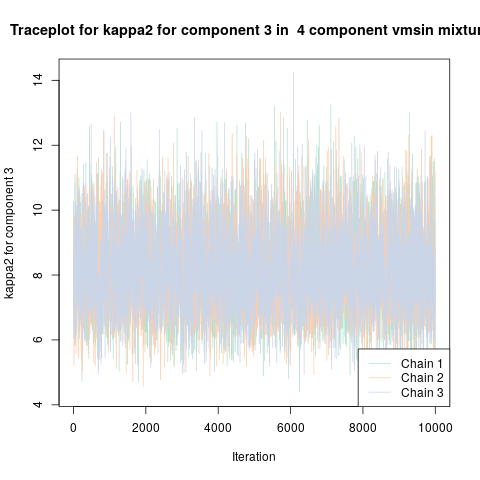}}
	\hfill
	\subcaptionbox{}%
	[.485\linewidth]{\includegraphics[height=2.6in, width = 2.6in]{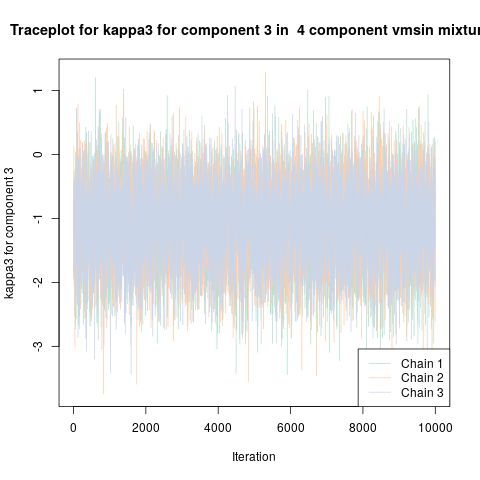}} \\
	\subcaptionbox{}
	[.485\linewidth]{\includegraphics[height=2.6in, width = 2.6in]{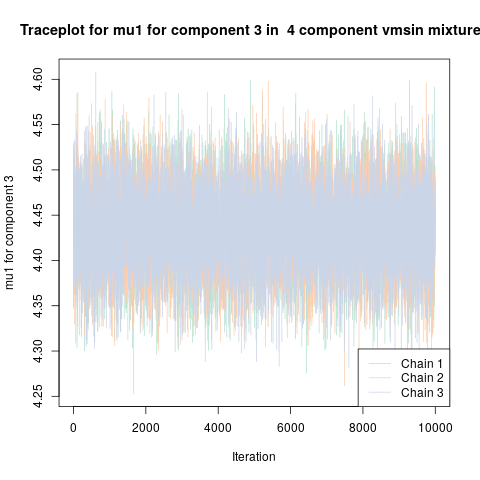}}
	\hfill
	\subcaptionbox{}%
	[.485\linewidth]{\includegraphics[height=2.6in, width = 2.6in]{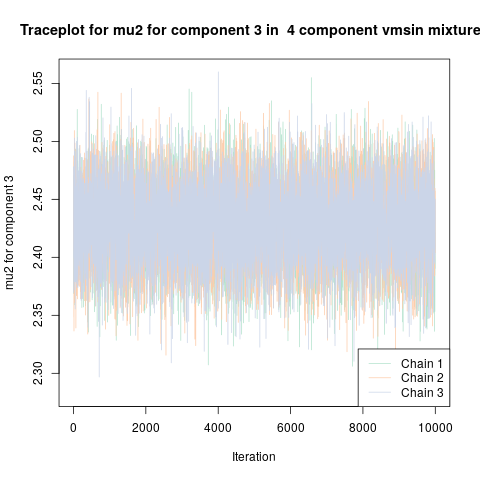}} 
	\caption{Trace plots for parameters in the third component for the Markov chain associated with the best fitted vmsin mixture model.}
	\label{vmsin_paramtrace_plots_comp3}
\end{figure}

\begin{figure}[!htpb]
	\centering 
	\subcaptionbox{}%
	[.485\linewidth]{\includegraphics[height=2.6in, width = 2.6in]{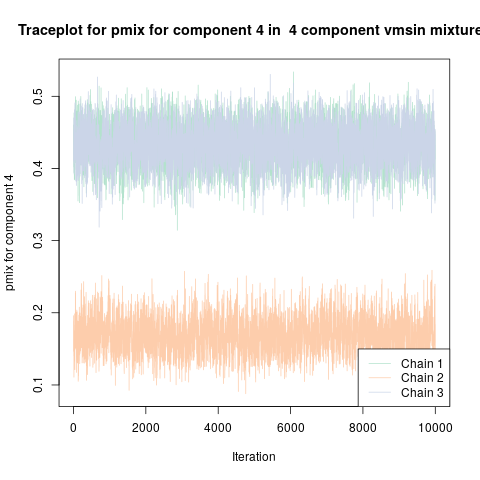}}
	\hfill
	\subcaptionbox{}%
	[.485\linewidth]{\includegraphics[height=2.6in, width = 2.6in]{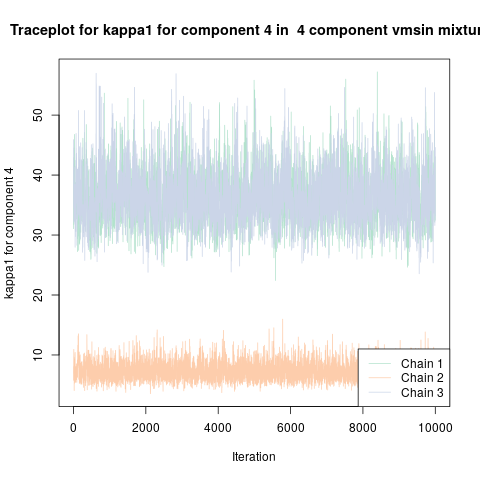}} \\
	\subcaptionbox{}%
	[.485\linewidth]{\includegraphics[height=2.6in, width = 2.6in]{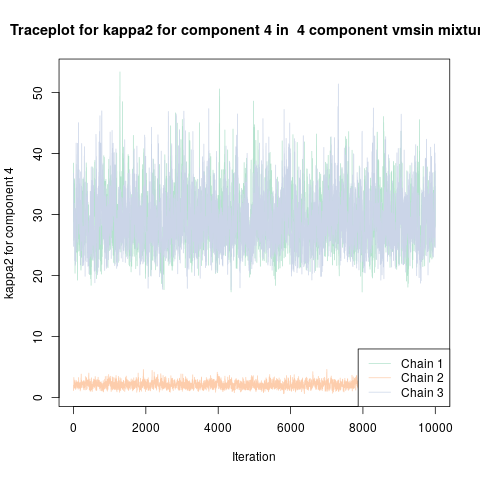}}
	\hfill
	\subcaptionbox{}%
	[.485\linewidth]{\includegraphics[height=2.6in, width = 2.6in]{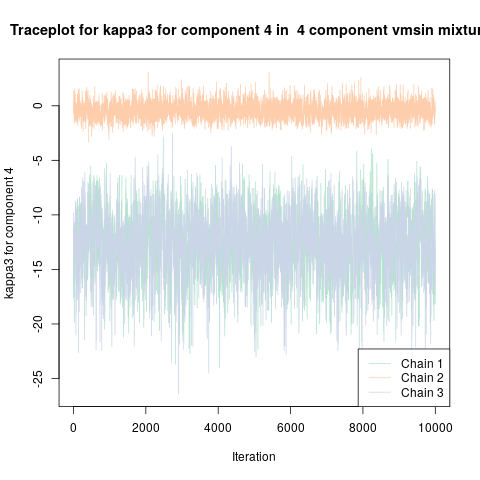}} \\
	\subcaptionbox{}
	[.485\linewidth]{\includegraphics[height=2.6in, width = 2.6in]{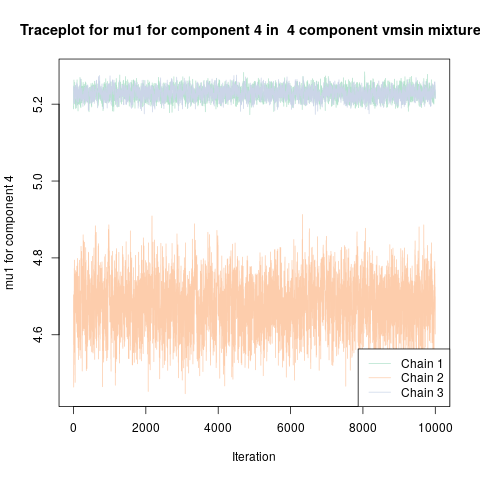}}
	\hfill
	\subcaptionbox{}%
	[.485\linewidth]{\includegraphics[height=2.6in, width = 2.6in]{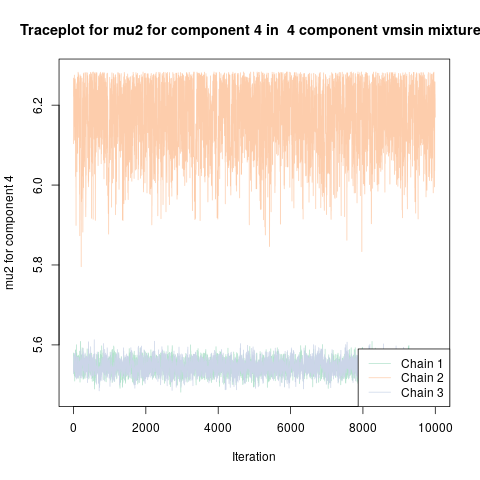}} 
	\caption{Trace plots for parameters in the fourth component for the Markov chain associated with the best fitted vmsin mixture model.}
	\label{vmsin_paramtrace_plots_comp4}
\end{figure}

\FloatBarrier

\section{Trace plots for 4 component vmsin with label switchings fixed}	

\label{vmsinparamtraces_fix}

\begin{figure}[!htpb]
	\centering 
	\subcaptionbox{}%
	[.485\linewidth]{\includegraphics[height=2.6in, width = 2.6in]{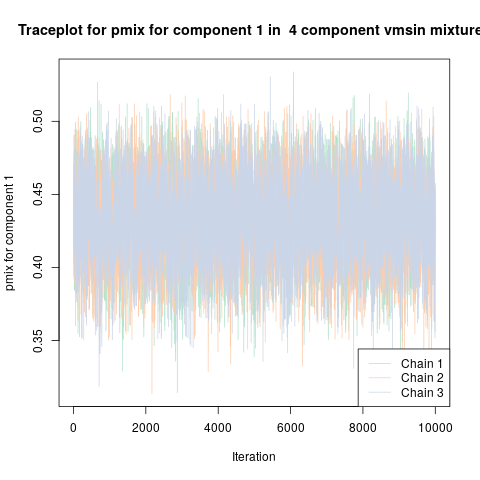}}
	\hfill
	\subcaptionbox{}%
	[.485\linewidth]{\includegraphics[height=2.6in, width = 2.6in]{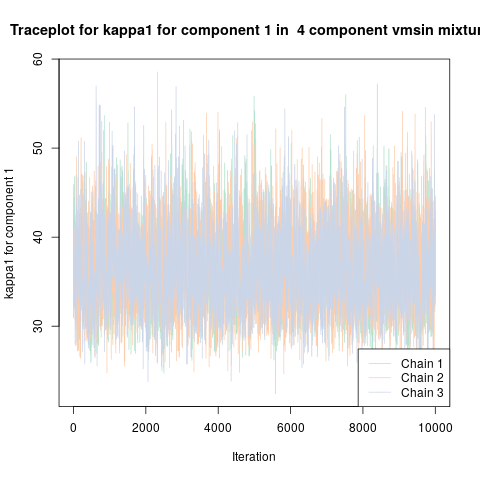}} \\
	\subcaptionbox{}%
	[.485\linewidth]{\includegraphics[height=2.6in, width = 2.6in]{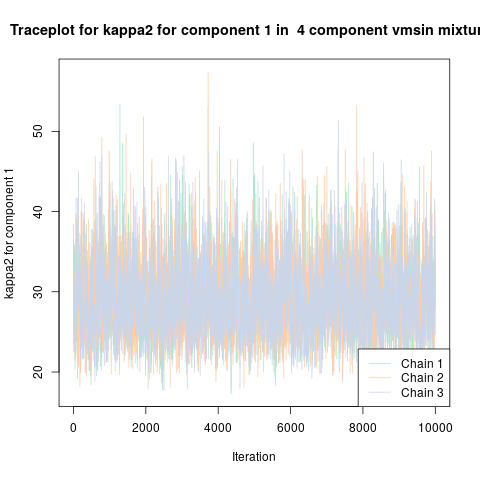}}
	\hfill
	\subcaptionbox{}%
	[.485\linewidth]{\includegraphics[height=2.6in, width = 2.6in]{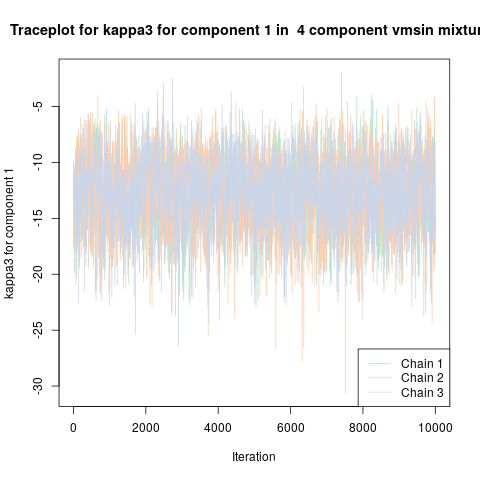}} \\
	\subcaptionbox{}
	[.485\linewidth]{\includegraphics[height=2.6in, width = 2.6in]{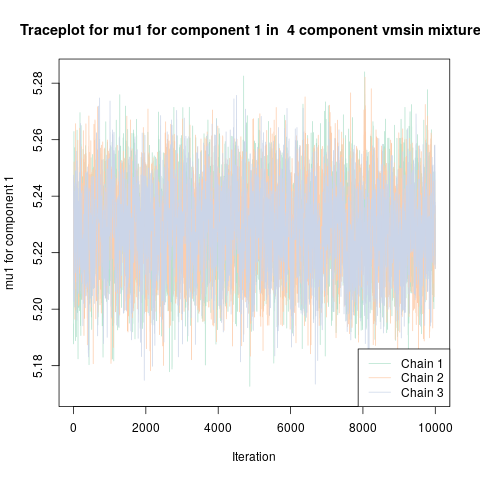}}
	\hfill
	\subcaptionbox{}%
	[.485\linewidth]{\includegraphics[height=2.6in, width = 2.6in]{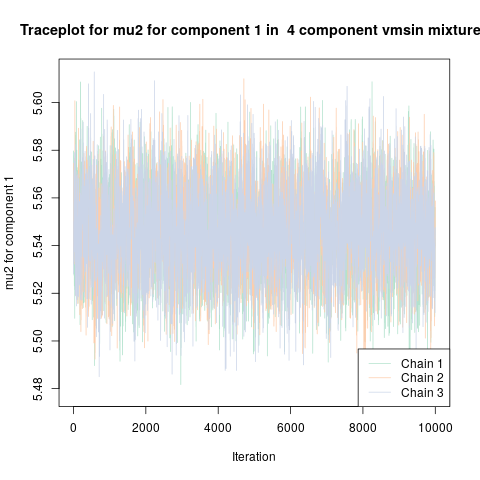}} 
	\caption{Trace plots for parameters in the first component for the Markov chain associated with the best fitted vmsin mixture model, after undoing label switching.}
	\label{vmsin_paramtrace_fix_plots_comp1}
\end{figure}

\begin{figure}[!htpb]
	\centering 
	\subcaptionbox{}%
	[.485\linewidth]{\includegraphics[height=2.6in, width = 2.6in]{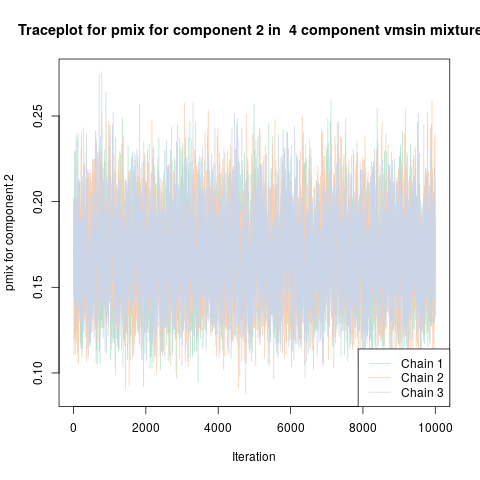}}
	\hfill
	\subcaptionbox{}%
	[.485\linewidth]{\includegraphics[height=2.6in, width = 2.6in]{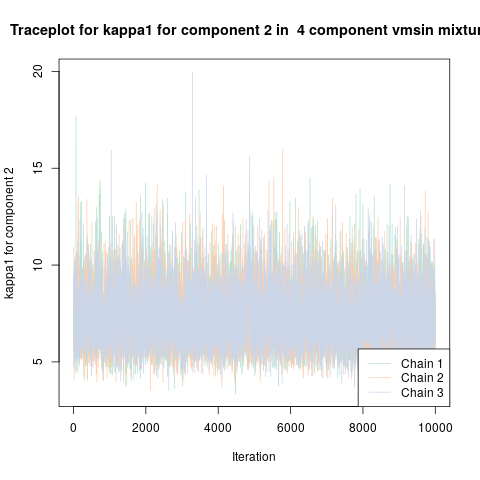}} \\
	\subcaptionbox{}%
	[.485\linewidth]{\includegraphics[height=2.6in, width = 2.6in]{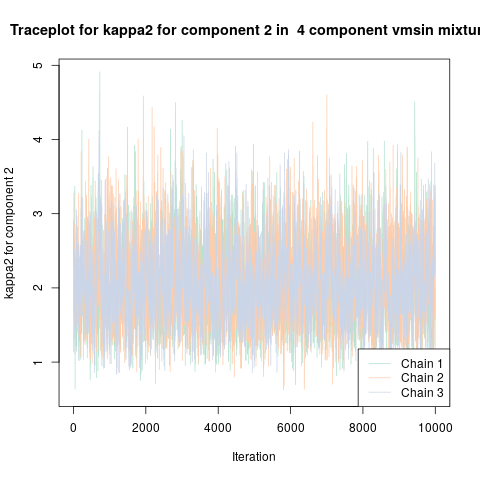}}
	\hfill
	\subcaptionbox{}%
	[.485\linewidth]{\includegraphics[height=2.6in, width = 2.6in]{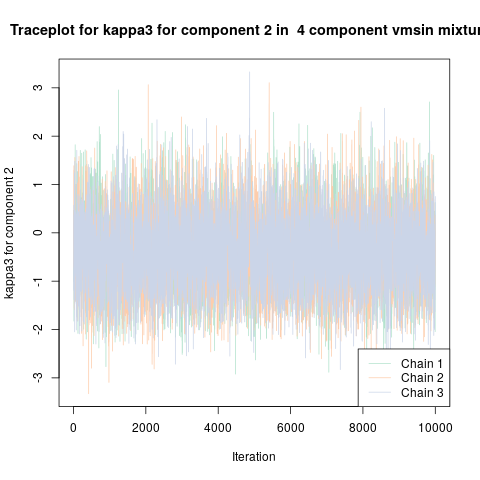}} \\
	\subcaptionbox{}
	[.485\linewidth]{\includegraphics[height=2.6in, width = 2.6in]{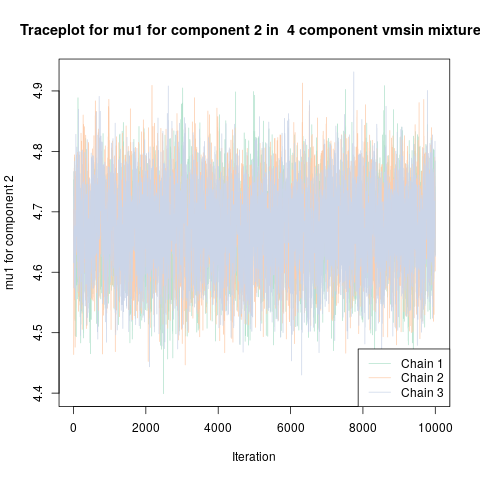}}
	\hfill
	\subcaptionbox{}%
	[.485\linewidth]{\includegraphics[height=2.6in, width = 2.6in]{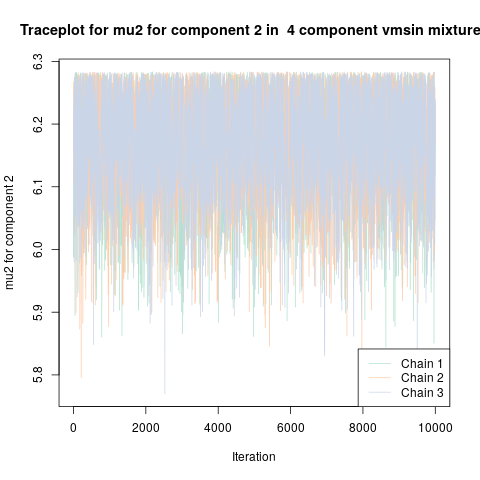}} 
	\caption{Trace plots for parameters in the second component for the Markov chain associated with the best fitted vmsin mixture model, after undoing label switching.}
	\label{vmsin_paramtrace_fix_plots_comp2}
\end{figure}

\begin{figure}[!htpb]
	\centering 
	\subcaptionbox{}%
	[.485\linewidth]{\includegraphics[height=2.6in, width = 2.6in]{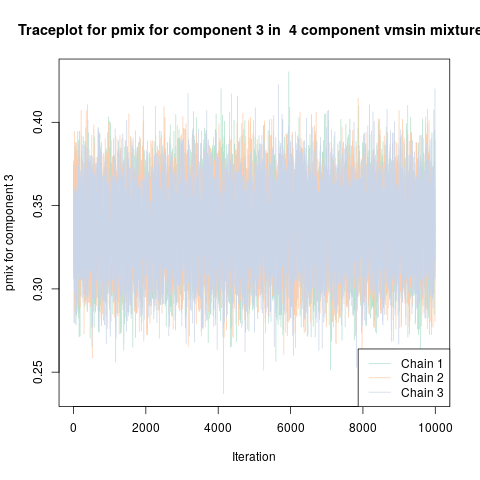}}
	\hfill
	\subcaptionbox{}%
	[.485\linewidth]{\includegraphics[height=2.6in, width = 2.6in]{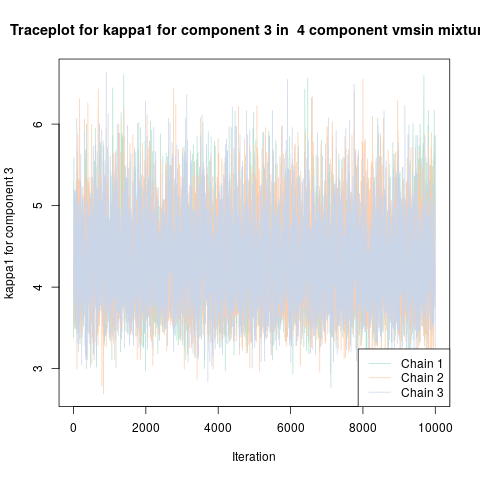}} \\
	\subcaptionbox{}%
	[.485\linewidth]{\includegraphics[height=2.6in, width = 2.6in]{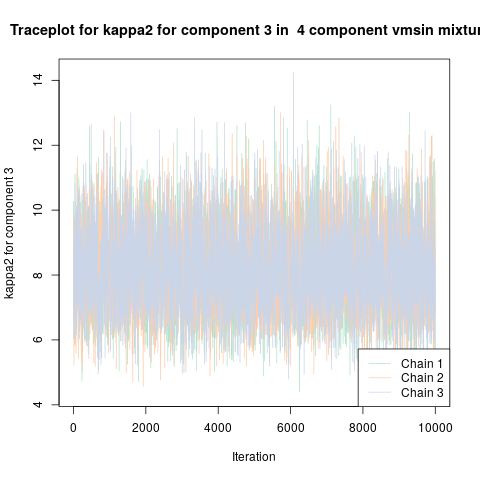}}
	\hfill
	\subcaptionbox{}%
	[.485\linewidth]{\includegraphics[height=2.6in, width = 2.6in]{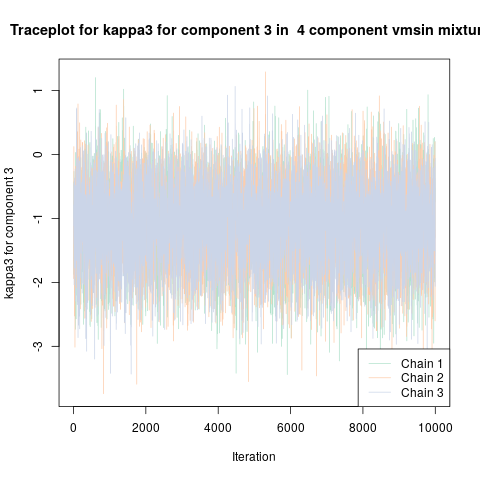}} \\
	\subcaptionbox{}
	[.485\linewidth]{\includegraphics[height=2.6in, width = 2.6in]{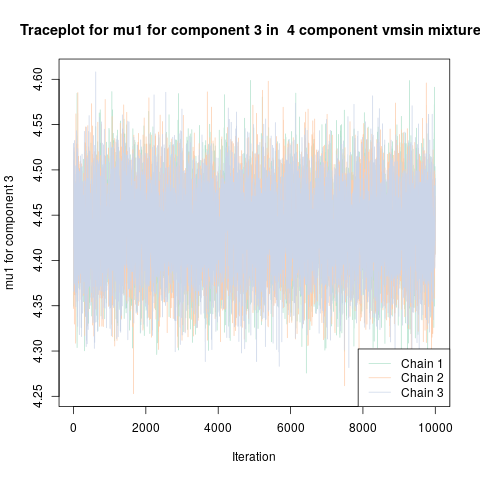}}
	\hfill
	\subcaptionbox{}%
	[.485\linewidth]{\includegraphics[height=2.6in, width = 2.6in]{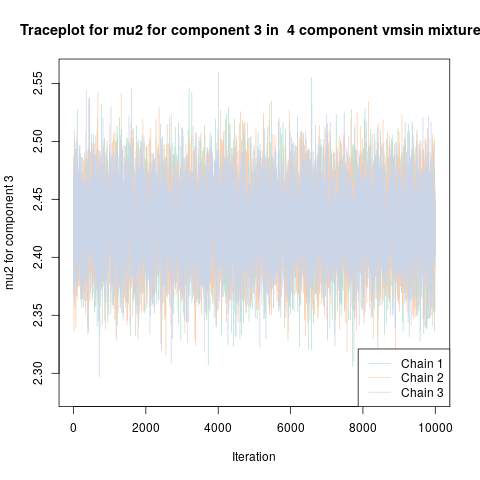}} 
	\caption{Trace plots for parameters in the third component for the Markov chain associated with the best fitted vmsin mixture model, after undoing label switching.}
	\label{vmsin_paramtrace_fix_plots_comp3}
\end{figure}

\begin{figure}[!htpb]
	\centering 
	\subcaptionbox{}%
	[.485\linewidth]{\includegraphics[height=2.6in, width = 2.6in]{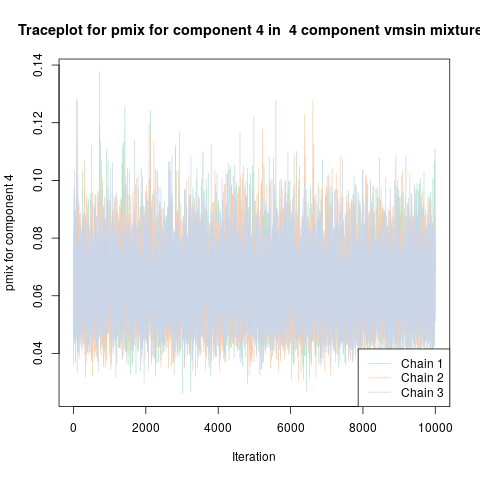}}
	\hfill
	\subcaptionbox{}%
	[.485\linewidth]{\includegraphics[height=2.6in, width = 2.6in]{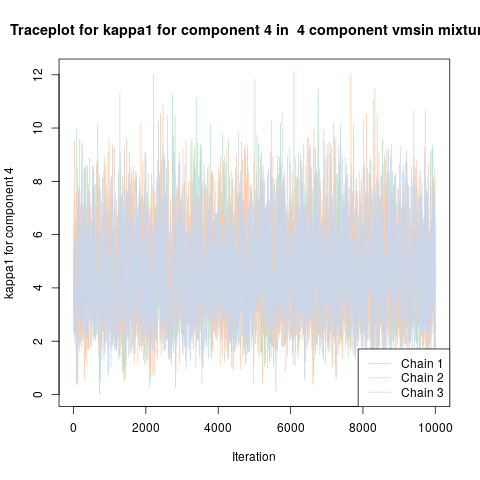}} \\
	\subcaptionbox{}%
	[.485\linewidth]{\includegraphics[height=2.6in, width = 2.6in]{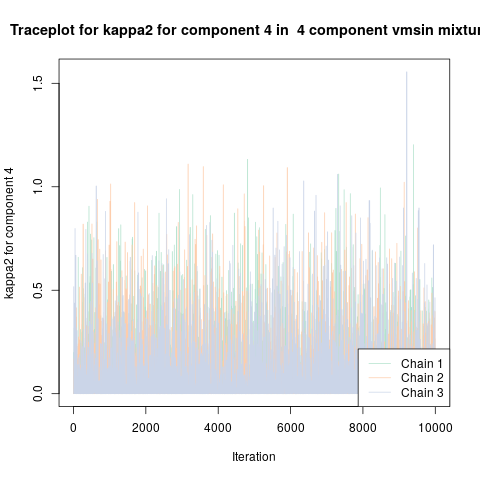}}
	\hfill
	\subcaptionbox{}%
	[.485\linewidth]{\includegraphics[height=2.6in, width = 2.6in]{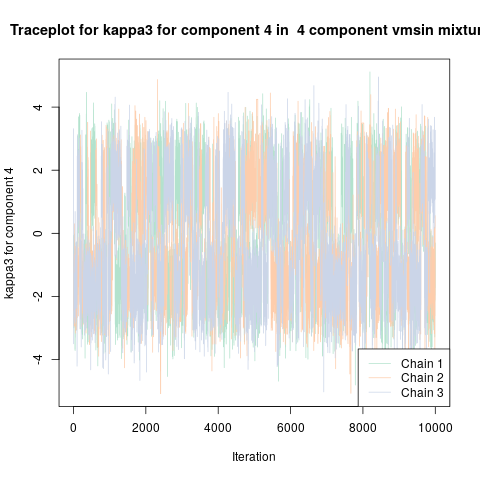}} \\
	\subcaptionbox{}
	[.485\linewidth]{\includegraphics[height=2.6in, width = 2.6in]{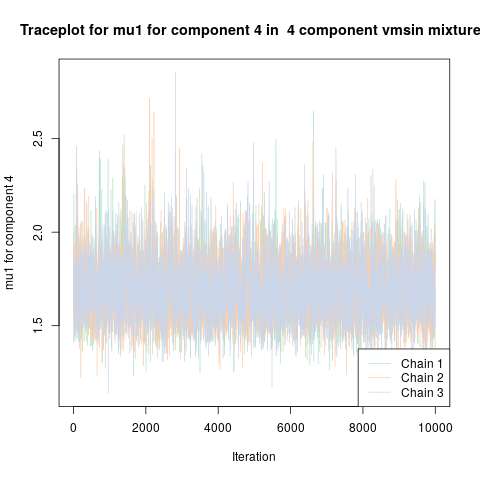}}
	\hfill
	\subcaptionbox{}%
	[.485\linewidth]{\includegraphics[height=2.6in, width = 2.6in]{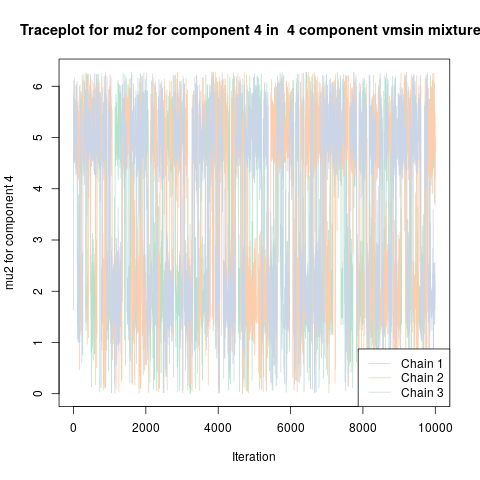}} 
	\caption{Trace plots parameters in the fourth component for the Markov chain associated with the best fitted vmsin mixture model, after undoing label switching.}
	\label{vmsin_paramtrace_fix_plots_comp4}
\end{figure}

\FloatBarrier

\end{appendices}

\end{document}